\documentclass[12pt]{amsart}
\usepackage[latin1]{inputenc}
\usepackage{mathrsfs}
\usepackage{mathtools}
\usepackage{amsmath}
\usepackage{amsfonts}
\usepackage{amssymb}
\usepackage{graphicx}
\usepackage{fourier}
\usepackage{dutchcal}
\usepackage[colorlinks=true,citecolor=green,linkcolor=blue]{hyperref}
\usepackage{enumerate}
\usepackage{esint}
\usepackage{bm}
\usepackage{pgfplots}

\usepackage{xcolor} 
\usepackage{verbatim}        
\usepackage{float}

\usepackage{subcaption}
\usepackage{epstopdf}

\usepackage{scalerel}
\usepackage{upgreek}

\let\nu\upupsilon
\let\phi\upvarphi
\let\eta\upeta
\let\psi\uppsi
\let\mu\upmu
\let\epsilon\varepsilon
\let\varrho\upvarrho

\usepackage[margin=1in]{geometry}

\usepackage[english]{babel}
\newtheorem{theorem}{Theorem}[section]
\newtheorem{lemma}[theorem]{Lemma}

\theoremstyle{definition}
\newtheorem{definition}[theorem]{Definition}

\newtheorem{corollary}[theorem]{Corollary}

\theoremstyle{remark}
\newtheorem{remark}[theorem]{Remark}

\newcommand{\norm}[1]{\left\lVert#1\right\rVert}
\newcommand{\abs}[1]{\left\lvert#1\right\rvert}
\newcommand{\pa}[1]{\left( #1 \right)}
\newcommand{\rpa}[1]{\left[ #1 \right]}
\newcommand{\br}[1]{\left\lbrace #1\right\rbrace}

\newcommand{\R}{\mathbb{R}}
\newcommand{\C}{\mathbb{C}}
\newcommand{\Z}{\mathbb{Z}}
\newcommand{\N}{\mathbb{N}}

\newcommand{\PP}{\mathscr{P}}
\newcommand{\X}{\mathcal{X}}
\newcommand{\wh}{\widehat}
\newcommand{\wt}{\widetilde}
\newcommand{\F}{\mathcal{F}}

\newcommand{\T}{\mathcal{T}}
\newcommand{\A}{\mathcal{A}}
\newcommand{\B}{\mathcal{B}}

\newcommand{\NNN}{\nonumber\\}

\newcommand{\tm}[1]{\mathrm{#1}}
\numberwithin{equation}{section}

\definecolor{bostonuniversityred}{rgb}{0.8, 0.0, 0.0}

\definecolor{byzantium}{rgb}{0.44, 0.16, 0.39}

\newcommand{\ii}{\mathfrak{i}}

\title{A notion of partial order in the Choose the Leader model}
\author[A. Einav]{Amit Einav}
\address{Amit Einav \hfill\break 
	School of Mathematical Sciences, Durham University, Upper Mountjoy Campus, Stockton Road DH1 3LE Durham, United Kingdom}
\email{amit.einav@durham.ac.uk}
\author{Yue Jiang}
\address{Yue Jiang \hfill\break 
	School of Mathematical Sciences, Durham University, Upper Mountjoy Campus, Stockton Road DH1 3LE Durham, United Kingdom}
\email{yue.jiang@durham.ac.uk}

\begin{document}
\maketitle

\begin{abstract}
	In this work we continue the study of non-chaotic asymptotic correlations in many element systems and discuss the emergence of a new notion of asymptotic correlation -- partial order -- in the Choose the Leader (CL) system. Similarly to the newly defined notion of order, partial order refers to alignment of the elements in the system -- though it allows for deviation from total adherence. 
	Our presented work revolves around the definition of partial order and shows its emergence in the CL model in its original critical scaling. Furthermore, we discuss the propagation of partial order in the CL model and give a quantitative estimate to the convergence to this state. 
	This new notion (as well as that of order) opens the door to exploring old and new (probabilistic) models of biological and societal nature in a more realistic way.
\end{abstract}
{\fontsize{5}{4}\selectfont
	}
\\{\indent\fontsize{5}{4}\selectfont{MSC Subject Classification:}
	82C22, 60F99, 35Q82, 35B40}

\section{Introduction}\label{sec:intro}
\subsection{Background: mean field models and limits}\label{subsec:background_and_MFL}
Systems of many interacting elements are prevalent in most fields of natural sciences, as well as in economical and societal settings. Understanding such systems, theoretically or in practice, can be quite challenging and computationally costly.
The late 19th century, a time where great progress in the study of kinetic gases and statistical mechanics was achieved, saw a fundamental shift from the (then) conventional approach in studying such systems: instead of considering the behaviour of individuals in the system, people began to explore the behaviour of an \textit{average} element. Since then numerous probabilistic and analytic ideas and approaches have been developed and used to deal with such systems.

One framework to investigate systems with many element is known as \textit{mean field models and limits}. In this approach such systems are represented by an average (probabilistic) model given by a PDE for the density measure of the system, the so-called \textit{master equation} (or Liouville equation), together with a correlation condition which captures the expected emerging phenomena. 

One of the first examples of a mean field model, as well as a notion of asymptotic correlation, was Kac's particle model (see \cite{K1956}). Motivated by a desire to provide a probabilistic justification to the validity of the Boltzmann equation, Kac proposed a simplified average model of a dilute gas consisting of $N$ \textit{indistinguishable} particles which undergo binary collisions. He showed that under the assumption of choatitcity, i.e. emerging independence of finitely many particles as the number of particles in the system goes to infinity, the Boltzmann equation arises as limit of the equation that describes the behaviour of one average particle in the original system.

Kac's model and approach had ramification beyond their immediate success. In the last few decades, the mean field limit approach has moved beyond particle systems and permeated into the realms of biology, chemical interactions, and even sociology. Prime examples are models of swarming, neuronal networks, and societal consensus/opinion variations (see \cite{APD2021,BFFT2012,BCC2011,DLMT2017,FL2016} as well as probabilistic models described in the review papers \cite{CD(I)2022,CD(II)2022} and the many references within). 

While casting a wide net, until the recent work of the first author, \cite{E2024}, the \textit{only} asymptotic correlation relation used in most mean field models was chaoticity. Since in many biological and societal situations we expect to see a development of strong dependence, and even adherence -- opposite phenomena to that which chaoticity encapsulates, there is a need to explore additional notions of asymptotic correlation. 

The intuition that there might be something besides chaos was validated in the work of Carlen, Degond, and Wennberg from 2013, \cite{CDW2013}, where the authors constructed a swarming model, the Choose the Leader (CL) model, which showed a departure from chaoticity. This was the starting point of \cite{E2024} where a new notion of \textit{order} was defined.

\subsection{Choose the Leader (CL) model and the notion of order}\label{subsec:CL_order}

The CL model is a Kac-like swarming model which describes a herd of $N$ animals that moves in a planar domain. Each animal is represented by its velocity which is assumed to be of length $1$. As these velocities can be thought of as elements of the unit circle in $\C$, using the identity $v=e^{i\theta}$ allows us to replace the velocity variables in the model with their relative angle $\theta$, considered on the interval $\rpa{-\pi,\pi}$ where $-\pi$ and $\pi$ are identified as the same angle.

The ``collision'' process in the CL model is as follows: at a random time, given by a Poisson stream with a rate $\lambda>0$, a pair of animals is chosen at random in a uniform way. One of the animals, chosen again at random with probability $1/2$, adapts its velocity to the other animal's up to a small amount of ``noise''. Mathematically, if the $i$-th and $j$-th animals interacted and the $j$-th animal decided to adapt it velocity to the $i$-th's animal (i.e. it decided to follow the $i-$th animal's lead) we find that after the interaction finished, the original velocity angles of these animals, $\pa{\theta_i,\theta_j}$, changes to $\pa{\theta_i,\theta_i+\mathcal{Z}}$ where $\mathcal{Z}$ is an independent random variable with values in $\rpa{-\pi,\pi}$ distributed in accordance to a given probability density function $h$.

Following on the above description we find that the master equation for the CL model, i.e. the evolution equation for the probability measure of the ensemble of animals on the $N-$dimensional torus $\T^N=[-\pi,\pi]^{N}$, is: 
\begin{equation}\nonumber
	\begin{split}
		\partial_t F_N\pa{\theta_1,\dots,\theta_N} =& \frac{2\lambda}{N-1}\sum_{i<j}\Bigg\{\frac{h\pa{\theta_i-\theta_j}}{2}\Bigg( \rpa{F_N}_{\wt{j}}\pa{\theta_1,\dots,\wt{\theta}_{j},\dots,\theta_N}\\
		&+\rpa{F_N}_{\wt{i}}\pa{\theta_1,\dots,\wt{\theta}_{i},\dots,\theta_N}\Bigg)
		-F_N\pa{\theta_1,\dots,\theta_N}\Bigg\},
	\end{split}
\end{equation}
with
\begin{equation}\nonumber
	\rpa{F_N}_{\wt{j}}\pa{\theta_1,\dots,\wt{\theta}_{j},\dots,\theta_N}=\int_{-\pi}^{\pi}F_N\pa{\theta_1,\dots,\theta_N}\frac{d\theta_j}{2\pi}, 
\end{equation}
and where we have used the notation $\pa{\theta_1,\dots,\wt{\theta}_{j},\dots,\theta_N}$ for the $(N-1)$-dimensional vector which is attained by removing $\theta_j$ from the original $N$-dimensional vector $\pa{\theta_1,\dots,\theta_N}$. We will continue and use this notation throughout this work.

The underlying process that governs the CL model leads us to believe that since the interactions between the animals creates \textit{strong} correlation we can't expect to observe chaoticity. We intuitively expect to see the herd moving in some random direction in relative unison -- the more $h$ is concentrated around $\theta=0$, the more aligned the herd should be. However, Carlen et al.\ showed the following in \cite{CDW2013}:
\begin{theorem}\label{thm:CDE_chaos}
	Assume that $\br{F_{N}(0)}_{N\in\N}$ is $f$-chaotic. Then for any $t>0$ the family of solutions to the CL master equation with initial data $\br{F_{N}(0)}_{N\in\N}$, $\br{F_{N}(t)}_{N\in\N}$, is $f(t)$-chaotic where $f(t)$ satisfies the equation
	\begin{equation}\nonumber
		\partial_t f\pa{\theta,t} = \pa{h\ast f}\pa{\theta,t} - f\pa{\theta,t}.
	\end{equation}
\end{theorem}

The reason behind the fact that chaos remains a valid asymptotic correlation in the CL model is that as it stands, the CL model retains the Kac-like assumption of \textit{sparsity}: in a given unit of time we expect to see only \textit{one} interaction between a couple of animals in the entire herd. This implies that while the interactions are extremely strong they are still very unlikely to happen between most pairs of animals. 

Carlen et al.\ suggested that in order to see chaos breaks we need to rescale both the time variable $t$ (to allow enough correlations to build up) and the strength of the interaction (represented by $h$)\footnote{It is worth to mention that rescaling the interaction with respect to $N$ is not unlikely when one considers biological and societal settings -- the number of people/animals in the system may affect how they align/confer to each other.}. 

While scaling by a factor of $N$ seems a natural choice for the time rescaling, as it guarantees that in a (rescaled) unit time many pairs of animals have interacted, the scaling of the interaction is less clear. Motivated by a standard scaling on the line we replace $h$ with 
\begin{equation}\nonumber
	g_{\epsilon_N}\pa{\theta}=\frac{1}{\epsilon_N G_{\epsilon_N}}g\pa{\frac{\theta}{\epsilon_N}},\qquad \theta \in [-\pi,\pi], 
\end{equation}
where $g$ is an even probability density function on $\R$ which we call the \textit{interaction generating function},
$$G_{\epsilon_N}=\frac{1}{2\pi}\int^{\frac{\pi}{\epsilon_N}}_{-\frac{\pi}{\epsilon_N}}g(x)dx,$$
and $\br{\epsilon_N}_{N\in\N}$, the spatial scaling parameters, is a positive sequence that goes to zero.

The rescaled CL master equation can be found in \cite{CCDW2013} and is given by
\begin{equation} \label{eq:master_CL_rescaled}
	\begin{split}
		&\partial_t F_N\pa{\theta_1,\dots,\theta_N,t} = \frac{2\lambda N}{N-1}\sum_{i<j}\Bigg\{\frac{1}{2}g_{\epsilon_N}\pa{\theta_i-\theta_j} \\
		&\pa{\rpa{F_N}_{\wt{j}}\pa{\theta_1,\dots,\wt{\theta}_{j},\dots,\theta_N}+\rpa{F_N}_{\wt{i}}\pa{\theta_1,\dots,\wt{\theta}_{i},\dots,\theta_N}}
		-F_N\pa{\theta_1,\dots,\theta_N}\Bigg\}.
	\end{split}
\end{equation}

As our interactions in the above become more and more concentrated around $\theta=0$ when $N$ increases, we can intuitively imagine that the herd will become more and more aligned -- more \textit{ordered}. This observation led the first author of the presented work to define a new notion of asymptotic correlation in \cite{E2024}\footnote{The definition given here is equivalent to that of \cite{E2024} where the limiting measure was $\mu_0\pa{\theta_1}\prod_{i=1}^{k-1} \delta\pa{\theta_i-\theta_{i+1}}$. The reason we recast it here is motivated by the new concept of \textit{partial order} which will shortly be defined.}:
\begin{definition}\label{def:order}
	Let $\mathcal{X}$ be a Polish space with a group operation $+$. We say that a sequence of symmetric probability measures\footnote{Symmetric probability measures are probability measures which are invariant under permutation of the variables.}, $\mu_N\in \PP\pa{\mathcal{X}^N}$ with $N\in\N$, is $\mu_0-$ordered for some probability measure $\mu_0\in \PP\pa{\mathcal{X}}$ if for any $k\in\N$
	\begin{equation}\label{eq:def_of_order}
		\Pi_k\pa{\mu_{N}}\pa{\theta_1,\dots,\theta_k} \underset{N\to\infty}{\overset{\text{weak}}{\longrightarrow}}\mu_0\pa{\theta_1}\prod_{i=2}^{k} \delta\pa{\theta_i-\theta_1},
	\end{equation}
	where $\Pi_k \pa{\mu_N}$ is the $k-$th marginal of $\mu_N$ and $\delta$ is the delta measure concentrated at the additive zero.
\end{definition}

Is this definition appropriate in the setting of the CL model? Exploring the BBGKY hierarchy of \eqref{eq:master_CL_rescaled}, i.e. the systems of equations for the marginals of $F_N$, (representing the evolution of a fixed number of animals), gives us a closed system of equations which unveils potential choices for spatial scaling $\epsilon_N$. In \cite{E2024}, the author has used simple Fourier analysis to identify three possible relationships between the time and interaction scaling under the assumption that the interaction generating function, $g$, has a finite moment of order $3$:
\begin{enumerate}[(i)]
	\item $N\epsilon_N^2 \underset{N\to\infty}{\longrightarrow}\infty$: in this case the effects of the interactions are \textit{slower} than the time scaling and consequently the adherence of the herd doesn't manage to build up fast enough. We expect that chaos will prevail in this case\footnote{In fact, we expect to achieve the time independent chaotic equilibrium in this case as our rescaled time goes to infinity with $N$.}.
	\item\label{item:diffusive} $N\epsilon_N^2=1$: in this case the effects of the interactions and time scaling are ``balanced'' in a diffusive manner. Carlen et al.\ have explored this scaling relationship in \cite{CCDW2013} and \cite{CDW2013} where they showed the breaking of chaoticity in this case\footnote{The same analysis holds if $\lim_{N\to\infty}N\epsilon_N^2=C$ with $0<C<\infty$.}.	 
	\item\label{item:order} $N\epsilon_N^2 \underset{N\to\infty}{\longrightarrow}0$: in this case the effects of the interactions are \textit{faster} than the time scaling and consequently we expect the rise of \textit{order}.
\end{enumerate} 

The last case was explored extensively in \cite{E2024}. Using the conventional notation $F_{N,k}=\Pi_k\pa{F_N}$ the author has shown the following:

\begin{theorem}\label{thm:main_order}
	Let $\br{F_N(t)}_{N\in\N}$ be the family of symmetric solutions to \eqref{eq:master_CL_rescaled} with initial data $\br{F_N(0)}_{N\in\N}$. Assume in addition that $\lim_{N\to\infty}N\epsilon_N^2=0$ and that $\br{F_{N,k}\pa{0}}_{N\in\N}$ converges weakly as $N$ goes to infinity to a family $f_{k,0} \in \PP\pa{\mathcal{T}^k}$ for any $k\in\N$. 
	Then for any $t>0$ and any $k\in\N$, $\br{F_{N,k}(t)}_{N\in\N}$ converges weakly as $N$ goes to infinity to $f_k(t)\in \PP\pa{\T^k}$ which satisfies
	\begin{equation}\label{eq:limit_of_F_N_K_no_order}
		\begin{aligned}
			&f_k\pa{\theta_1,\dots,\theta_k,t}= e^{-\lambda k\pa{k-1}t}f_{k,0}\pa{\theta_1,\dots,\theta_k}\\
			&+2\lambda\int_{0}^t e^{-\lambda k\pa{k-1}\pa{t-s}}\pa{ \sum_{i<j\leq k}f_{k-1}\pa{\theta_1,\dots, \wt{\theta_i},\dots,\theta_k,s}\delta\pa{\theta_i-\theta_j}}ds,
		\end{aligned}
	\end{equation}
	where the sum over an empty set is defined to be $0$.
	
	In particular, we have that $\br{f_k(t)}_{k\in\N}$ converges weakly as $t$ goes to infinity to an $f_{1,0}-$ordered family
	$$\lim_{t\to\infty}f_{k}(\theta_1,\dots, \theta_k,t) = f_{1,0}\pa{\theta_1}\prod_{j=2}^{k}\delta\pa{\theta_{i}-\theta_i},$$
	where the product over an empty set is defined to be $~1$. Moreover, if $\br{F_N(0)}_{N\in\N}$ is $f_{1,0}-$ordered then 
	$$f_k\pa{\theta_1,\dots,\theta_k,t}=f_{1,0}\pa{\theta_1}\prod_{j=2}^{k}\delta\pa{\theta_{i}-\theta_1}$$ 
	for all $t>0$.
\end{theorem}

We will continue to use the convention that the sum over an empty set is defined to be $0$ and the product over an empty set is defined to be $1$ throughout this work.

Theorem \ref{thm:main_order} not only shows the propagation of order in the case where $\lim_{N\to\infty}N\epsilon_N^2=0$, it also shows that order will be \textit{generated} at large times -- though we need to be careful here as the result is obtained by first taking $N$ to infinity and then taking $t$ to infinity. 

Motivated by the above, it was enquired in \cite{E2024} whether or not a similar phenomena occurs in the critical case where $N\epsilon_N^2=1$. Unfortunately, or fortunately, order does not hold in that case:

\begin{theorem}\label{thm:main_in_between}
	Let $\br{F_N(t)}_{N\in\N}$ be the family of symmetric solutions to \eqref{eq:master_CL_rescaled} with initial data $\br{F_N(0)}_{N\in\N}$. Assume in addition that $N\epsilon_N^2=1$. Then $\br{F_N(t)}_{N\in\N} $  is neither chaotic nor ordered for any $t>0$.\\
	Moreover, if $\br{F_{N,1}(0)}_{N\in\N}$ and $\br{F_{N,2}(0)}_{N\in\N}$ converge weakly to $f_{1,0}\in \PP\pa{\T}$  and $f_{2,0}\in \PP\pa{\T^2}$ respectively then for all $t>0$, $\br{F_{N,1}(t)}_{N\in\N}$ and $\br{F_{N,2}(t)}_{N\in\N}$ converge to some $f_1(t)\in\PP\pa{\T}$ and $f_2(t)\in \PP\pa{\T^2}$ respectively and
	\begin{equation}\nonumber 
		\lim_{t\to\infty} f_1(\theta_1,t) =1,
	\end{equation}
	while
	\begin{equation}\nonumber 
		\lim_{t\to\infty}f_2(\theta_1,\theta_2,t) = \mathcal{H}\pa{\theta_2-\theta_1}
	\end{equation}
	with 
	$$\mathcal{H}\pa{\theta}= \sum_{n\in\Z}\frac{2}{m_2 n^2+2}e^{in \theta}=1+4\sum_{n\in\N}\frac{\cos\pa{n\theta}}{m_2 n^2+2},$$
	where $m_2=\int_{\R}x^2g(x)dx$ with $g$ being the interaction generating function. 
\end{theorem}

\begin{figure}[H]\label{fig1}
\centering
\includegraphics[scale=0.3]{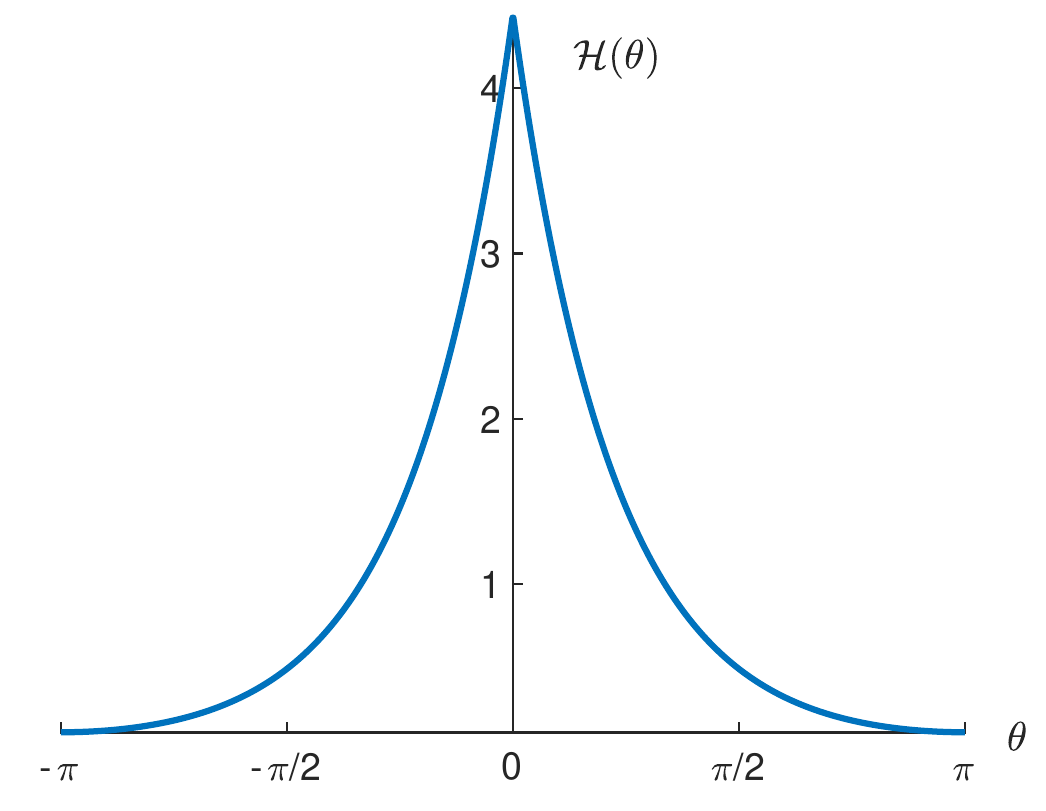}
\caption{\small{A plot of an approximation of $\mathcal{H}$ with $m_2=1$.}}
\end{figure}
The fact that $\mathcal{H}$ is concentrated around $0$ shows that while order doesn't emerge in the critical case, some partial adherence is observed. The possibility of such correlation arising in this case was conjectured in \cite{E2024} and is the goal and achievement of the presented work.

\subsection{Partial order and our main results} Motivated by the definition of order, Definition \ref{def:order}, and Theorem \ref{thm:main_in_between}, one seemingly reasonable option to define a notion of partial order is to replace the delta measures in \eqref{eq:def_of_order} with another measure, say $\nu$, which is somewhat concentrated around $\theta=0$\footnote{It is worth to note that \textit{this} is the reason why we have defined the notion of order in \eqref{eq:def_of_order} relative to a fixed variable (i.e. considering $\delta\pa{\theta_i-\theta_1}$ and not $\delta\pa{\theta_{i+1}-\theta_i}$ as in  the original definition in \cite{E2024}). By doing so we are able to more readily generalise this notion to partial order where we compare to a \textit{fixed} random leader.}. Some care is needed here: the fact that we used the delta measures in \eqref{eq:def_of_order} automatically implies that every element is on the same footing -- choosing $\theta_1$ as the variable for which we know the limiting distribution $\mu_0$ was arbitrary since any other element behaves exactly the same -- the measure that expresses the behaviour of $k$ elements is concentrated on the diagonal. In other words for any $i\not=j$
$$\mu_0\pa{\theta_i}\prod_{l\not=i}^{k} \delta\pa{\theta_l-\theta_i}=\mu_0\pa{\theta_j}\prod_{l\not=j}^{k} \delta\pa{\theta_l-\theta_j}.$$
When moving to the case where the measure of the $k$ elements is no longer concentrated on the diagonal, we need to take into account the intuition that \textit{every element} is a potential ``leader'', and since the original elements were indistinguishable -- every element has the same likelihood to be chosen. This motivates us to consider partial order as the existence of two measures, $\mu_0$ and $\nu$, such that 
	\begin{equation}\label{eq:naive_partial_order}
	\Pi_k\pa{\mu_{N}}\pa{\theta_1,\dots,\theta_k} \underset{N\to\infty}{\overset{\text{weak}}{\longrightarrow}} \frac{1}{k}\sum_{i=1}^k \mu_0\pa{\theta_i}\prod_{j\not=i}^{k} \nu\pa{\theta_j-\theta_i},
\end{equation}
where $\nu$ is an even measure on $\X$, i.e. $\nu(x)=\nu(-x)$. Following on Theorem \ref{thm:main_in_between}, one could hope that $\nu$ has a probability density function with respect to the underlying measure on $\T$, $\frac{d\theta}{2\pi}$.  \\
It turns out that the above is not attainable in the simple case of the CL model  \textit{unless} $\nu=\frac{d\theta}{2\pi}$ (this will be shown in \S \ref{sec:preliminaries}). In a sense, the relationship expressed in \eqref{eq:naive_partial_order} is too simplistic as it states that the correlation relation between any ``follower'' and ``leader'' is completely decoupled from all the other followers and independent of their number. In general, we need to allow for a more complex relation.

The above realisations motivate our notion of partial order. Before we define it formally, we remind the reader the notions of translation of measures and even measures:

\begin{definition}\label{def:translation_and_evenness}
	Let $\mathcal{X}$ be a Polish space with a continuous group operation $+$ and its inverse. Given a probability measure $\mu\in\PP\pa{\mathcal{X}}$ and an element $c\in\mathcal{X}$ we define the translated probability measure 
	\begin{equation}\label{eq:transalted_measure}
		\begin{split}
			\mu\pa{\cdot-c} = {\tau_c}_\# \mu,
		\end{split}
	\end{equation}
	where $\tau_c(x) = x+c$. In other words
	\begin{equation}\label{eq:translation_measure_for_integrals}
		\begin{split}
			\int_{\mathcal{X}}f\pa{x}d\mu\pa{x-c}=\int_{\mathcal{X}}f\pa{x+c}d\mu\pa{x}.
		\end{split}
	\end{equation}
	for all $f\in L^1\pa{\mathcal{X},\mu}$.\\
	We say that a probability measure $\mu\in \PP\pa{\mathcal{X}^{k}}$ with $k\in\N$ is even if 
		\begin{equation}\label{eq:even}
	 {\ii}_{\#}\mu=\mu
	\end{equation}
	where $\ii\pa{x_1,\dots,x_k}=\pa{-x_1,\dots,-x_k}$. In other words
	\begin{equation}\label{eq:even_for_integrals}
		\begin{split}
			\int_{\X}f\pa{x_1,\dots,x_k}d\mu(x)= \int_{\X}f\pa{-x_1,\dots,-x_k}d\mu(x)
		\end{split}
	\end{equation}
	for any $f\in L^1\pa{\X^k,\mu}$.
\end{definition}

\begin{definition}[Partial Order]\label{def:partial_order}
	Let $\mathcal{X}$ be a 
	Polish space with a continuous group operation $+$ and its inverse. We say that a sequence of symmetric probability measures, $\mu_N\in \PP\pa{\mathcal{X}^N}$ with $N\in\N$, is $\br{\pa{\eta_{k},\nu_{k-1}}}_{k\in\N}-$partially ordered where $\eta_{k}\in \PP\pa{\mathcal{X}}$ for any $k\in\N$ and $\nu_k\in \PP\pa{\mathcal{X}^{k}}$ are even and symmetric for any $k\in\N$
	 if 
	\begin{equation}\label{eq:def_of_partial_order}
		\Pi_k\pa{\mu_{N}}\pa{\theta_1,\dots,\theta_k} \underset{N\to\infty}{\overset{\text{weak}}{\longrightarrow}}
		\begin{cases}
			\eta_1\pa{\theta_1},& k=1,\\
			\frac{1}{k}\sum_{i=1}^k \eta_{k}\pa{\theta_i}\nu_{k-1}\pa{\theta_1-\theta_i,\dots,\widetilde{\theta_i-\theta_i},\dots, \theta_k-\theta_i}, & k\geq 2,
		\end{cases}
	\end{equation}
	for any $k\in\N$ where
	\begin{equation}\nonumber
		\begin{aligned}
			\eta_k&\pa{\cdot_i}\nu_{k-1}\pa{\cdot_1-\cdot_i,\dots,\widetilde{\cdot_i-\cdot_i},\dots, \cdot_k-\cdot_i}\pa{E}\\
			&=\int_{\X}\chi_E\pa{x_1+x_i,\dots, x_i, \dots, x_k+x_i}d\nu_{k-1}\pa{x_1,\dots,\wt{x_i},\dots,x_k}d\eta_k(x_i).
		\end{aligned}
	\end{equation}
\end{definition}
\begin{remark}\label{rem:why_totally_}
	We assumed that $\nu_k$ is even to capture the intuition that $\nu_k$ only cares about the ``distance'' between the ensemble $\pa{\theta_1,\dots,\wt{\theta_i},\dots,\theta_k}$ and $\theta_i$ in the sense that 
	$$\nu_{k-1}\pa{\theta_1-\theta_i,\dots,\widetilde{\theta_i-\theta_i},\dots, \theta_j-\theta_i,\dots, \theta_k-\theta_i}=\nu_{k-1}\pa{\theta_i-\theta_1,\dots,\widetilde{\theta_i-\theta_i},\dots, \theta_i-\theta_k}.$$
\end{remark}

In the presented work we will focus solely on the setting of the CL model, i.e. on the family of probability spaces $\br{\pa{\T^k, \frac{d\theta_1\dots d\theta_k}{\pa{2\pi}^k}}}_{k\in\N}$ where $\frac{d\theta_1\dots d\theta_k}{\pa{2\pi}^k}$ is the underlying measure when considering probability density functions.

Our first result shows the rise of partial order in the critical scaling of the CL model, $N\epsilon_N^2=1$.

\begin{theorem}\label{thm:main_partial_order}
	Let $\br{F_N(t)}_{N\in\N}$ be the family of symmetric solutions to \eqref{eq:master_CL_rescaled} with initial data $\br{F_N(0)}_{N\in\N}$. Assume in addition that $N\epsilon_N^2=1$, the interaction generating function $g$ has a finite moment of order $l \in \N\setminus \br{1,2}$, i.e. for some $l\in\N$, $l\geq 3$
	$$m_l =\int_{\R}\abs{x}^l g(x)dx <\infty,$$ 
	and $\br{F_{N,k}\pa{0}}_{N\in\N}$ converges weakly as $N$ goes to infinity to $f_{k,0} \in \PP\pa{\mathcal{T}^k}$ for any $k\in\N$. \\
	Then for any $t>0$ and any $k\in\N$, $\br{F_{N,k}(t)}_{N\in\N}$ converges weakly as $N$ goes to infinity to a symmetric $f_k(t)\in \PP\pa{\T^k}$ whose Fourier coefficients are given by
	\begin{equation}\label{eq:limit_of_F_N_K_partial_order}
		\begin{aligned}
			&\wh{f_k}\pa{n_1,\dots,n_k,t}= e^{-\frac{\lambda \pa{2k\pa{k-1}+m_2\pa{\sum_{r=1}^k n_r^2}}t}{2}}\wh{f_{k,0}}\pa{n_1,\dots,n_k}\\
			&+2\lambda\sum_{i<j\leq k}\int_{0}^t e^{-\frac{\lambda \pa{2k\pa{k-1}+m_2\pa{\sum_{r=1}^k n_r^2}}\pa{t-s}}{2}}\wh{f_{k-1}}\pa{n_1,\dots,\wt{n_i},\dots,n_i+n_j,\dots, \wt{n_j},\dots, n_k,s}ds,
		\end{aligned}
	\end{equation}
	where the position of $n_i+n_j$ is arbitrary due to the symmetry of $f_k$.
	 
	In addition we have that for any $k\in\N$
	$$\lim_{t\to\infty}f_{k}(\theta_1,\dots, \theta_k,t) = f_{k,\infty}(\theta_1,\dots, \theta_k) $$
	where $f_{k,\infty}\in \PP\pa{\T^k}$ is the probability measure whose Fourier coefficients are given by
	\begin{equation}\label{eq:formula_for_f_infty}
		\begin{split}
			&\wh{f_{1,\infty}}(n_1)=\delta_0\pa{n_1},\\
			&\wh{f_{k,\infty}}\pa{n_1,\dots,n_{k}}=\frac{\delta_0\pa{\sum_{r=1}^{k}n_r}}{k}\sum_{l=1}^{k}\nu_{k-1}\pa{n_1,\dots,\wt{n_l},\dots,n_{k}},\quad k\geq 2,
		\end{split}
	\end{equation}
	with 
	$$\delta_0\pa{n} = \begin{cases}
		1,& n=0,\\
		0,& n\ne0,
	\end{cases}$$
	and where  $\nu_k$ is an even and symmetric probability measure on $\T^k$ which is absolutely continuous with respect to $\frac{d\theta_1\dots d\theta_k}{\pa{2\pi}^k}$.  Moreover, the probability density of $\nu_k$, $\upnu_k$, is continuous and \eqref{eq:formula_for_f_infty} can be rewritten as
	\begin{equation}\label{eq:formula_for_f_infty_spatial}
		f_{k,\infty}\pa{\theta_1,\dots,\theta_k} = \begin{cases}
			\frac{d\theta_1}{2\pi},& k=1,\\
			\pa{\frac{1}{k}\sum_{l=1}^k \upnu_{k-1}\pa{\theta_1-\theta_l,\dots \wt{\theta_l-\theta_l},\dots,\theta_k-\theta_l}}\frac{d\theta_1\dots d\theta_k}{\pa{2\pi}^k}, & k\geq 2.
		\end{cases}
	\end{equation}
	Furthermore, $\nu_k$ is the unique probability measure on $\PP\pa{\T^k}$ whose Fourier coefficients are given by\footnote{Notice that when $k=1$
		$$\xi_1\pa{n_1,\dots,\wt{n_1}, - \sum_{r\not= 1}n_r, \dots, n_{1}} = \xi_1(0)=\xi_1(n_1)=1, $$
		so the formula still makes sense.}
	\begin{equation}\label{eq:explicit_nu_k}
		\begin{split}
			& \wh{\nu_{k}}\pa{n_1,\dots,n_{k}} =\sum_{i=1}^k\frac{2(k+1)\xi_k\pa{n_1,\dots,\wt{n_i}, - \sum_{r\not= i}n_r, \dots, n_{k}}}{2k(k+1)+m_2\sum_{r=1}^{k}n_r^2 + m_2\pa{\sum_{r=1}^{k}n_r}^2 }
		\end{split}
	\end{equation}
	where $\xi_k$ satisfies the recursive relation
	\begin{equation}\label{eq:explicit_xi_k}
		\begin{split}
			&\wh{\xi_1}(n_1)=1,\\
			& \wh{\xi_{k}}\pa{n_1,\dots,n_{k}} =
			\frac{4\sum_{i< j \leq k}\xi_{k-1}\pa{n_1,\dots,\wt{n_i},\dots,n_i+n_j,\dots, \wt{n_j},\dots, n_{k}}}{2k(k-1)+m_2\sum_{r=1}^{k}n_r^2 },\quad k\geq 2.
		\end{split}
	\end{equation}
	 Lastly, the family $\br{f_{N,\infty}}_{N\in\N}$,  is $\br{\pa{\frac{d\theta}{2\pi},\nu_{k-1}}}_{k\in\N}-$partially ordered.
\end{theorem}

Much like Theorem \ref{thm:main_order}, the above theorem indicates that partial order is generated as $t$ goes to infinity. A natural question at this point would be: does partial order also propagate in our setting? Unfortunately, the answer to this question is in the negative in general:
\begin{theorem}\label{thm:lack_of_propagation}
	Let $f_1(t)\in \PP\pa{\T}$ and $f_2(t)\in \PP\pa{\T^2}$ be the limits as $N$ goes to infinity of the solutions $\br{F_{N,1}(t)}_{N\in\N}$ and $\br{F_{N,2}(t)}_{N\in\N}$ to \eqref{eq:master_CL_rescaled} with $N\epsilon_N^2=1$ as described in Theorem \ref{thm:main_partial_order}. Then
	\begin{enumerate}[(i)]
		\item\label{item:no_partial_order_for_some_time} If there exists $n_0\in \Z $ such that $\abs{n_0} > \sqrt{\frac{2}{m_2}}$ and 
		$$\wh{f_{2,0}}\pa{n_0,n_0} + \frac{2\wh{f_{1,0}}\pa{2n_0}}{m_2n_0^2-2}\ne 0$$
		then there exists no $~\eta(t)\in \PP\pa{\T}$ and an even $~\nu(t)\in\PP\pa{\T} $ such that
		\begin{equation}\label{eq:f_2_for_lack}
			f_2\pa{\theta_1,\theta_2,t} = \frac{1}{2}\pa{\mu\pa{\theta_1,t}\nu\pa{\theta_2-\theta_1,t}+\mu\pa{\theta_2,t}\nu\pa{\theta_2-\theta_1,t}}
		\end{equation}
		for all $t\in \pa{0,\infty}$. 
		\item\label{item:no_partial_order_for_all_time} If there exists a polynomial function $p$ such that 
		$$ \limsup_{n\to \pm\infty}\abs{p(n)\pa{\wh{f_{2,0}}\pa{n,n} + \frac{2\wh{f_{1,0}}\pa{2n}}{m_2n^2-2}}}>0,$$
		then \eqref{eq:f_2_for_lack} can't hold for any $~\eta(t)\in \PP\pa{\T}$ and an even $~\nu(t)\in \PP\pa{\T}$, and any $t\in \pa{0,\infty}$.
	\end{enumerate}
	In particular, under the same conditions of Theorem \ref{thm:main_partial_order}, the initial data 
	$$F_{N}\pa{\theta_1,\dots,\theta_N,0}=\mu_0(\theta_1)\prod_{i=2}^N\delta\pa{\theta_N-\theta_1}$$
	is $~\br{\pa{\mu_0, \prod_{i=1}^{k-1} \delta}}_{k\in\N}-$partially ordered and the condition of \eqref{item:no_partial_order_for_all_time} holds when $\mu_0\pa{\theta}$ corresponds to the probability density with Fourier coefficients $\wh{\mu_0}(n) = \frac{2}{2+n^2}$\footnote{A close look shows that $\nu_1$ from Theorem \ref{thm:main_partial_order} satisfies this when $m_2=1$ and as such there indeed exists such probability measure.}.
	Consequently, partial order does not propagate in the rescaled CL model when $N\epsilon_N^2=1$ in general.
\end{theorem}

With a lack of propagation property for our model, we turn our attention to the generation of partial order (as indicated in Theorem \ref{thm:main_partial_order}) and ask ourselves if it bears any meaning for the family of marginal solutions to \eqref{eq:master_CL_rescaled}, $\br{F_{N,k}(t)}_{N\in\N}$. After all, the precise claim of Theorem \ref{thm:main_partial_order} is that 
$$\lim_{t\to\infty}\pa{\lim_{N\to\infty}F_{N,k}(t)}=f_{k,\infty}$$
is partially ordered. As the above is an iterated limit, issues may arise that will prevent us from saying that $F_{N,k}(t)$ and $f_{k,\infty}$ are ``close'' for large enough $N$ and $t$. The last theorem of our work addresses these concerns and shows that the $F_{N,k}(t)$ converges weakly to $f_{k,\infty}$ \textit{uniformly} in $N$ and $t$ under a few additional conditions on the interaction generating function. The weak convergence will be quantitatively estimated by the Fourier coefficients of the appropriate measures (more on the connection between weak convergence and the Fourier coefficients can be found in \S\ref{sec:preliminaries}). We will need the following two definitions to be able to state our theorem: 

\begin{definition}\label{def:}
	A probability density $g\in \PP\pa{\R,dx}$ is said to be a \textit{strong} interaction generating function if $g$ is an even function with a finite moment of order $l\in \N \setminus \br{1,2}$, and if $g\in L^p\pa{\R}$ for some $p>1$.
\end{definition}

\begin{definition}\label{def:k_level_set_of_Z_r}
	Let $k\in\N$ be given and let $r\in\N$ be such that $r\leq k$. The $k-$th level set of $\N^r$ is the set 
	\begin{equation}\nonumber 
		\mathfrak{L}_{k}\pa{r} = \br{\bm{p}=\pa{p_1,\dots,p_r}\in \N^r\;|\; \sum_{l=1}^r p_l=k}.
	\end{equation}
	For a given $\sigma\in S^k$, where $S^{k}$ is the permutation group of $\br{1,\dots,k}$, and an element of $\bm{p}\in \mathfrak{L}_{k}\pa{r}$  we define the map
	$$s_{\bm{p}}^{\sigma}:\Z^k \to \Z^r$$
	by
	$$s_{\bm{p}}^{\sigma}\pa{n_1,\dots,n_k} = \pa{\sum_{m=1}^{p_1}n_{\sigma(m)},\sum_{m=p_1+1}^{p_1+p_2}n_{\sigma(m)},\dots, \sum_{m=p_1+\dots+p_{r-1}+1}^{k}n_{\sigma(m)}}.$$ 
\end{definition}

\begin{theorem}\label{thm:quantitative_convergence}
		Let $\br{F_N(t)}_{N\in\N}$ be the family of symmetric solutions to \eqref{eq:master_CL_rescaled} with a strong interaction generating function $g$ and with initial data $\br{F_N(0)}_{N\in\N}$. Assume in addition that $N\epsilon_N^2=1$ and that $\br{F_{N,k}\pa{0}}_{N\in\N}$ converges weakly as $N$ goes to infinity to $f_{k,0} \in \PP\pa{\mathcal{T}^k}$ for any $k\in\N$. \\
		Then, there exist explicit $\mathfrak{N}_0\in \N$ and $\gamma>0$ that depend only on the interaction generating function $g$, and explicit constants $\mathcal{C}_k,\mathcal{D}_k$ such that for any 
		$$N\geq \max\pa{\mathfrak{N}_0, \frac{\pa{2m_l}^{\frac{l}{2}}}{\pi^2},\pa{\frac{32m_l }{\pi^l\max\pa{8,m_2}}}^{\frac{2}{l-2}},2k, \pa{\frac{96m_l k}{\pi^lm_2}}^{\frac{2}{l-2}}}$$
		we have that for any  $\pa{n_1,\dots,n_{k}}\in \Z^{k}$
		\allowdisplaybreaks
		\begin{align}
			\nonumber&\abs{\wh{F_{N,k}}\pa{n_1,\dots,n_{k},t}-\wh{f_{k,\infty}}\pa{n_1,\dots,n_{k}}} \leq 
			\abs{\wh{F_{N,k}}\pa{n_1,\dots,n_{k},0}-\wh{f_{k,0}}\pa{n_1,\dots,n_{k}}} \\
				\label{eq:quantitative_convergence} & + \pa{\frac{N}{N-1}}^{k-1}\pa{k-1}\max_{\fontsize{5}{4}\selectfont{\begin{matrix}
							r\in \br{1,\dots,k-1},\\
							\bm{p}\in \mathfrak{L}_{k}\pa{r},\;\;\sigma\in S^k
					\end{matrix}}, }\abs{\wh{F_{N,r}}\pa{s_{\bm{p}}^{\sigma}\pa{n_1,\dots,n_k},0}-\wh{f_{r,0}}\pa{s_{\bm{p}}^{\sigma}\pa{n_1,\dots,n_k}}}\\
			\nonumber& +2e^{-\lambda k(k-1)t}e^{-\lambda \gamma \sqrt[\kappa(l)~]{N} t}+
			\frac{\mathcal{C}_k}{\sqrt[\kappa(l)~]{N}}+	\mathcal{D}_k\pa{\frac{ e^{-2\lambda t}}{1-e^{-2\lambda t}}+e^{-\frac{\lambda m_2 t}{2}}}
		\end{align}
		where the maximum over an empty set is defined to be zero and 
		\begin{equation}\nonumber
			\kappa(l) = \begin{cases}
				3,& l=3,\\
				2, & l\geq 4.
			\end{cases}
		\end{equation}
\end{theorem}

\subsection{The organisation of the paper} \sloppy In section \S\ref{sec:preliminaries} we will recast weak convergence on $\T^k$ 
in terms of Fourier coefficients and use that to find an equivalent definition to partial order in the setting of the CL model. We will also discuss the existence of partially ordered states and state a few technical lemmas which will serve us in our analysis. In section \S\ref{sec:partial_order} we will consider the measures arising as the mean field limits of the $k-$th marginals of the solutions to the CL master equation and show the generation of partial order, proving Theorem \ref{thm:main_partial_order}. We will discuss the lack of propagation of partial order and prove Theorem \ref{thm:lack_of_propagation} in \S\ref{sec:propagation},  and then turn our attention to the quantitative convergence of marginals of solutions of \eqref{eq:master_CL_rescaled} to our limiting partially ordered state  in \S\ref{sec:quantitative} where we will prove Theorem \ref{thm:quantitative_convergence}. We conclude the work with a few remarks and future prospects in \S\ref{sec:remarks}, followed by an Appendix where we prove results that were postponed to allow a better flow for the paper.

\section{Preliminaries}\label{sec:preliminaries}
In this section we will discuss the notion of partial order on $\T^N$, explore the existence of partially ordered states, and consider a few technical lemmas on the Fourier coefficients of the rescaled interaction density, $g_{\epsilon_N}$, which will be crucial to prove theorems \ref{thm:main_partial_order} and  \ref{thm:quantitative_convergence}.
\subsection{Weak convergence, partial order and the Fourier coefficients}
We start with the following simple observation whose proof can be found in \cite{E2024}:
\begin{lemma}\label{lem:weak_convergence_and_fourier}
	Let $\br{\eta_{N}}_{N\in\N}$ be a sequence of probability measures on $\T^k,$ and let $~\eta\in \PP\pa{\T^k}$. Then $\eta_N\underset{N\to\infty}{\overset{\text{weak}}{\longrightarrow}}\eta$ if and only if for any $\pa{n_1,\dots,n_k}\in\Z^k$
	\begin{equation}\nonumber 
		\begin{aligned}
			\wh{\eta_N}\pa{n_1,\dots,n_k} &= \int_{\T^k} e^{-i\sum_{j=1}^k n_j \theta_j} d\eta_N\pa{\theta_1,\dots,\theta_k}
			&\underset{N\to\infty}{\longrightarrow}\wh{\eta}\pa{n_1,\dots,n_k}. 
		\end{aligned}
	\end{equation}
\end{lemma}

With this at hand, we are able to recast the definition of partial order in the setting of the CL model in a more tractable form.

\begin{lemma}\label{lem:partial_order_in_fourier}
	The family $\mu_N\in \PP\pa{\T^N}$, with $N\in\N$, is $\br{\pa{\eta_{k},\nu_{k-1}}}_{k\in\N}-$partially ordered if and only if 
	for any $\pa{n_1,\dots,n_k}\in\Z^k$
	\begin{equation}\label{eq:convergence_of_fourier_coefficients_partial}
		\begin{aligned}
			\wh{\Pi_k\pa{\mu_N}}\pa{n_1,\dots,n_k}&\underset{N\to\infty}{\longrightarrow} \begin{cases}
				\wh{\eta_1}\pa{n_1},& k=1,\\
				\frac{\wh{\eta_k}\pa{\sum_{j=1}^k n_j}}{k}\sum_{l=1}^k \wh{\nu_{k-1}}\pa{n_1,\dots,\wt{n_l},\dots,n_{k}}, & k\geq 2,
			\end{cases}
		\end{aligned}
	\end{equation}
	with
	\begin{equation}\label{eq:even_fourier}
		\wh{\nu_{k}}\pa{n_1,\dots,n_{k}} = \wh{\nu_{k}}\pa{-n_1,\dots,- n_{k}}, 
	\end{equation}
	\begin{equation}\label{eq:symmetry_fourier}
			\wh{\nu_{k}}\pa{n_1,\dots,n_{k}} = \wh{\nu_{k}}\pa{n_{\sigma(1)},\dots, n_{\sigma(k)}}, 
	\end{equation}
	for any $\pa{n_1,\dots,n_{k}}\in \Z^{k}$ and any $\sigma \in S^{k}$.
\end{lemma}

The proof of this lemma relies on two observations.

\begin{lemma}\label{lem:fourier_of_partial_order}
	For a given $k\in\N$ we define $\mu_{\mathrm{po},k}\in \PP\pa{\T^k}$ by 
	$$\mu_{\mathrm{po},k}\pa{\theta_1,\dots,\theta_k} = \begin{cases}
		\eta_1\pa{\theta_1},& k=1,\\
		\frac{1}{k}\sum_{i=1}^k \eta_{k}\pa{\theta_i}\nu_{k-1}\pa{\theta_1-\theta_i,\dots,\widetilde{\theta_i-\theta_i},\dots, \theta_k-\theta_i}, & k\geq 2,
	\end{cases}$$
	where $\eta_k\in \PP\pa{\T}$ and $\nu_{k-1}\in \PP\pa{\T^{k-1}}$.
	Then for any $\pa{n_1,\dots,n_k}\in \Z^k$ we have that 
	$$\wh{\mu_{\mathrm{po},k}}\pa{n_1,\dots,n_k} = \begin{cases}
		\wh{\eta_1}\pa{n_1},& k=1,\\
		\frac{\wh{\eta_k}\pa{\sum_{j=1}^k n_j}}{k}\sum_{l=1}^k \wh{\nu_{k-1}}\pa{n_1,\dots,\wt{n_l},\dots,n_{k}}, & k\geq 2.
	\end{cases}$$
\end{lemma}

\begin{proof}
	The statement is immediate for the case $k=1$ and as such we may assume that $k\geq 2$. Using the definition of $\mu_{\mathrm{po},k}$ we find that for any $\pa{n_1,\dots,n_k}\in \Z^k$
	$$\wh{\mu_{\mathrm{po},k}}\pa{n_1,\dots,n_k} =\frac{1}{k}\sum_{l=1}^k\int_{\T^k}e^{-i\sum_{j=1}^k n_j \theta_j} d\eta_k\pa{\theta_l}d\nu_{k-1}\pa{\theta_1-\theta_l,\dots,\wt{\theta_l-\theta_l},\dots,\theta_k-\theta_l}$$
	$$=\frac{1}{k}\sum_{l=1}^k\int_{\T^k}e^{-in_l \theta_l-i \sum_{j=1,\;j\not=l}^k n_j \pa{\theta_j+\theta_l}} d\eta_k\pa{\theta_l}d\nu_{k-1}\pa{\theta_1,\dots,\wt{\theta_l},\dots,\theta_k}$$
	$$=\frac{1}{k}\sum_{l=1}^k\int_{\T^k}e^{-i\pa{\sum_{j=1}^k n_j} \theta_l}e^{-i \sum_{j=1,\;j\not=l}^k n_j \theta_j} d\eta_k\pa{\theta_l}d\nu_{k-1}\pa{\theta_1,\dots,\wt{\theta_l},\dots,\theta_k}$$
	$$=\frac{\wh{\eta_k}\pa{\sum_{j=1}^k n_j}}{k}\sum_{l=1}^k \wh{\nu_{k-1}}\pa{n_1,\dots,\wt{n_l},\dots,n_{k}},$$
	which is the desired result.
\end{proof}

\begin{lemma}\label{lem:properties_of_even_and_symmetry_in_fourier}
	Let $k\in\N$ be given.
	\begin{enumerate}[(i)]
		\item\label{item:even_in_fourier} $\rho_k\in\PP\pa{\T^k}$ is even if and only if 
			\begin{equation}\label{eq:even_fourier_lem}
			\wh{\rho_k}\pa{n_1,\dots,n_{k}} = \wh{\rho_k}\pa{-n_1,\dots,- n_{k}}, 
		\end{equation}
			for any $\pa{n_1,\dots,n_{k}}\in \Z^{k}$.
		\item\label{item:symmetry_in_fourier} 
		$\rho_k\in\PP\pa{\T^k}$ is symmetric if and only if 
		\begin{equation}\label{eq:symmetry_fourier_lem}
			\wh{\rho_k}\pa{n_1,\dots,n_{k}} = \wh{\rho_k}\pa{n_{\sigma(1)},\dots, n_{\sigma(k)}}, 
		\end{equation}
			for any $\pa{n_1,\dots,n_{k-1}}\in \Z^{k}$ and any $\sigma \in S^{k}$.
	\end{enumerate}
\end{lemma}

\begin{proof}
	We start by recalling that if $\mu,\nu\in \PP\pa{\T^k}$ satisfy
	$$\int_{\T^k}f\pa{\theta_1,\dots,\theta_k}d\mu\pa{\theta_1,\dots,\theta_k}=\int_{\T^k}f\pa{\theta_1,\dots,\theta_k}d\nu\pa{\theta_1,\dots,\theta_k}$$
	for all $f\in C_b \pa{\T^k}$ then $\mu=\nu$.
	Consequently, testing the properties of evenness and symmetry is equivalent to testing the integral versions of these properties. 
	\begin{enumerate}[(i)]
		\item We notice that as
		$$\int_{\T^k}e^{-i \sum_{l=1}^k n_l \pa{-\theta_l} }d\rho_k(\theta_1,\dots,\theta_k)=\wh{\rho_k}\pa{-n_1,\dots,-n_k},$$
		 We conclude that if $\rho_k$ is even then \eqref{eq:even_fourier_lem} must hold.\\
		 Conversely, assuming \eqref{eq:even_fourier_lem} holds and using the above we find that for any trigonometric polynomial $p\pa{\theta_1,\dots, \theta_k}$ 
		 \begin{equation}\nonumber
		 \begin{aligned}
		 	\int_{\T^k}p&\pa{-\theta_1,\dots, -\theta_k}d\rho_k\pa{\theta_1,\dots,\theta_k}
		 	=\int_{\T^k}p\pa{\theta_1,\dots, \theta_k}d\rho_k\pa{\theta_1,\dots,\theta_k}.
		\end{aligned}	
		 \end{equation}
		 As trigonometric polynomials are dense in $C_b\pa{\T^k}$ with respect to the supremum norm, we conclude that the above holds for any $f\in C_b\pa{\T^k}$. Consequently, $\rho_k$ is even. 
		 \item We notice that for any $\sigma\in S^k$
		 	$$\int_{\T^k}e^{-i \sum_{l=1}^k n_k \theta_{\sigma(k)}}d\rho_k(\theta_1,\dots,\theta_k)=\int_{\T^k}e^{-i \sum_{l=1}^k n_{\sigma^{-1}(k)} \theta_{k}}d\rho_k(\theta_1,\dots,\theta_k)=\wh{\rho_k}\pa{n_{\sigma^{-1}(1)},\dots,n_{\sigma^{-1}(k)}}.$$
		 	Consequently, if $\rho_k$ is symmetric we must have that 
		 	$$\wh{\rho_k}\pa{n_{\sigma^{-1}(1)},\dots,n_{\sigma^{-1}(k)}}=\wh{\rho_k}\pa{n_{1},\dots,n_{k}}$$
		 	for any $\pa{n_1,\dots,n_k}\in\Z^k$ and any $\sigma\in S^k$ which implies \eqref{eq:symmetry_fourier_lem}.
		 	The converse follows from similar arguments to those given to show part \eqref{item:even_in_fourier}. 
	\end{enumerate}
\end{proof}

\begin{proof}[Proof of Lemma \ref{lem:partial_order_in_fourier}]
	This is an immediate result of lemmas \ref{lem:weak_convergence_and_fourier}, \ref{lem:fourier_of_partial_order}, and \ref{lem:properties_of_even_and_symmetry_in_fourier} together with the fact that the Fourier coefficients of any measure on $\T^k$ determine the measure uniquely. 
\end{proof}

\subsection{The existence of a partially ordered state and the lack of a decoupled state on  $\T^k$}
A natural question when defining a new notion is whether or not it makes sense -- i.e. whether or not the definition is vacuous. Theorem \ref{thm:main_partial_order} guarantees the existence of a partially ordered state, but it will require a decent amount of work. There are, however, two immediate polar options we can consider:
\begin{itemize}
	\item \textit{Ordered states}. Choosing 
	$$\updelta_{k-1} \pa{\theta_1,\dots,\theta_{k-1}}= \prod_{i=1}^{k-1} \delta\pa{\theta_i},$$
	with $\delta$ being the delta measure concentrated at $0$, we find that the family $\br{\updelta_{k-1}}_{k\in\N\setminus\br{1}}$ is even and symmetric. If $\br{F_N}_{N\in\N}$ is $\mu_0-$ordered then 
	$$\Pi_k\pa{F_N}\pa{\theta_1,\dots,\theta_k} \underset{N\to\infty}{\overset{\text{weak}}{\longrightarrow}} \mu_0(\theta_1) \prod_{j=2}^k\delta\pa{\theta_j-\theta_1}$$
	$$=\frac{1}{k}\sum_{i=1}^k \mu_{0}\pa{\theta_i}\updelta_{k-1}\pa{\theta_1-\theta_i,\dots,\widetilde{\theta_i-\theta_i},\dots, \theta_k-\theta_i}. $$
	Consequently, any $\mu_0-$ordered family is $\br{\pa{\mu_0,\updelta_{k-1}}}_{k\in\N}-$partially ordered.
	\item \textit{The uniform state (time independent chaotic state)}. For any $k\in\N$ the measure
	$$\nu_{k}\pa{\theta_1,\dots,\theta_{k}} = \frac{d\theta_1\dots d\theta_{k}}{\pa{2\pi}^{k}}$$
	is even and symmetric, as well as translation invariant. The family $F_N = \frac{d\theta_1\dots d\theta_N}{\pa{2\pi}^N}$ satisfies
	$$\Pi_k\pa{F_N}\pa{\theta_1,\dots,\theta_k} = \frac{d\theta_1\dots d\theta_k}{\pa{2\pi}^k}  = \frac{1}{k}\sum_{i=1}^k \frac{d\theta_i}{2\pi}\nu_{k-1}\pa{\theta_1-\theta_i,\dots,\widetilde{\theta_i-\theta_i},\dots, \theta_k-\theta_i}.$$
	Consequently, $\frac{d\theta_1\dots d\theta_N}{\pa{2\pi}^N}$ is $\br{\pa{\frac{d\theta}{2\pi},\nu_{k-1}}}_{k\in\N}-$partially ordered.
\end{itemize}

The above states seem unlikely to appear in most situations where we expect to see true partial order emerging. Much like in the case of chaos and order, the family of limiting probabilities in the definition of partial order, Definition \ref{def:partial_order}, is a natural candidate for a partially ordered state. The next lemma shows us how we can use such families to build up new partially ordered states from existing ones.

\begin{lemma}\label{lem:existence_of_partial_order}
	For any $N\in\N$ define the probability measure $\mu_N\in \PP\pa{\T^N}$ by
	$$\mu_N\pa{\theta_1,\dots,\theta_N}=\begin{cases}
	\eta_1(\theta_1),& N=1,\\
	\frac{1}{N}\sum_{i=1}^N \eta_N(\theta_i)\nu_{N-1}\pa{\theta_1-\theta_i,\dots,\widetilde{\theta_i-\theta_i},\dots, \theta_N-\theta_i}, & N\geq 2,
	\end{cases}$$
	where $\eta_N\in \PP(\T)$ and $\nu_{N}\in \PP\pa{\T^{N}}$ is symmetric for any $N\in\N$. If there exist $\eta_\infty\in \PP(\T)$ such that 
	$$\eta_N(\theta) \underset{N\to\infty}{\overset{\text{weak}}{\longrightarrow}} \eta_{\infty}(\theta)$$
	and if $\br{\nu_N}_{N\in\N}$ is $\br{\pa{\zeta_{k,\infty},\rho_{k-1,\infty}}}_{k\in\N}-$partially ordered for an appropriate $~\zeta_{k,\infty}\in \PP\pa{\T}$ and $\varrho_{k,\infty}\in \PP\pa{\T^k}$ 
	then $\br{\mu_N}_{N\in\N}$ is $\br{\pa{\eta_\infty \ast \zeta_{k,\infty},\rho_{k-1,\infty}}}_{k\in\N}-$partially ordered.
	
	In particular
	\begin{enumerate}[(i)]
		\item\label{item:example_uniform} If $\eta_\infty(\theta) =\zeta_{k,\infty}(\theta)= \frac{d\theta}{2\pi}$ for all $k\in \N$ then $\br{\mu_N}_{N\in\N}$ is $\br{\pa{\frac{d\theta}{2\pi},\rho_{k-1,\infty}}}_{k\in\N}-$partially ordered.
		\item\label{item:example_delta} if $\eta_\infty(\theta) =\delta(\theta)$ then  $\br{\mu_N}_{N\in\N}$ is $\br{\pa{\zeta_{k,\infty},\rho_{k-1,\infty}}}_{k\in\N}-$partially ordered.
	\end{enumerate}
\end{lemma}

\begin{proof}
We start by noticing that as $\nu_{N}$ are symmetric, so are $\mu_N$. This can be verified by using lemmas \ref{lem:fourier_of_partial_order} and \ref{lem:properties_of_even_and_symmetry_in_fourier} as for any $N\geq 2$
$$\wh{\mu_N}\pa{n_1,\dots,n_N} = 
	\frac{\wh{\eta_N}\pa{\sum_{j=1}^N n_j}}{N}\sum_{l=1}^N \wh{\nu_{N-1}}\pa{n_1,\dots,\wt{n_l},\dots,n_{N}}$$
	$$=\frac{\wh{\eta_N}\pa{\sum_{j=1}^N n_{\sigma(j)}}}{N}\sum_{l=1}^N \wh{\nu_{N-1}}\pa{n_{1},\dots,\wt{n_{l}},\dots,n_{N}} $$
		$$=\frac{\wh{\eta_N}\pa{\sum_{j=1}^N n_{\sigma(j)}}}{N}\sum_{l=1}^N \wh{\nu_{N-1}}\pa{n_{\sigma(1)},\dots,\wt{n_{\sigma(l)}},\dots,n_{\sigma(N)}} =\wh{\mu_N}\pa{n_{\sigma(1)},\dots,n_{\sigma(N)}},$$
for all $\sigma\in S^N$ and $\pa{n_1,\dots,n_N}\in\Z^N$. 

In our setting, the $k-$th marginal of $\mu_N$ can be defined as the unique measure $\Pi_k\pa{\mu_N}\in \PP\pa{\T^k}$ such that for any $f\in C_b\pa{\T^k}$ 
$$\int_{\T^k}f\pa{\theta_1,\dots,\theta_k}d\Pi_k\pa{\mu_N}\pa{\theta_1,\dots,\theta_k} = \int_{\T^N}f\pa{\theta_1,\dots,\theta_k}d\mu_N\pa{\theta_1,\dots,\theta_N}.$$
Much like the proof of Lemma \ref{lem:properties_of_even_and_symmetry_in_fourier}, the density of trigonometric polynomials in $\T^k$ and the uniqueness of the Fourier coefficients imply that the $\Pi_k\pa{\mu_N}$ is the unique probability measure such that 
$$\wh{\Pi_k\pa{\mu_N}}\pa{n_1,\dots,n_k} = \int_{\T^N}e^{i\sum_{l=1}^k n_l \theta_l}d\mu_N\pa{\theta_1,\dots,\theta_N}=\wh{\mu_N}\pa{n_1,\dots,n_k,0,\dots,0}$$
for any $\pa{n_1,\dots,n_k}\in\Z^k$.

Using lemmas \ref{lem:weak_convergence_and_fourier} and \ref{lem:partial_order_in_fourier} we find that for $N\geq k \geq 2$
$$\wh{\Pi_k\pa{\mu_N}}\pa{n_1,\dots,n_k} =
	\frac{\wh{\eta_N}\pa{\sum_{j=1}^k n_j}}{N}\Big(\sum_{l=1}^k \wh{\nu_{N-1}}\pa{n_1,\dots,\wt{n_l},\dots,n_{k},0,\dots,0}  + (N-k)\wh{\nu_{N-1}}\pa{n_1,\dots,n_{k},0,\dots,0} \Big)$$
	$$=\wh{\eta_N}\pa{\sum_{j=1}^k n_j}\pa{\frac{\sum_{l=1}^k \wh{\Pi_{k-1}\pa{\nu_{N-1}}}\pa{n_1,\dots,\wt{n_l},\dots,n_{k}}}{N} +\frac{N-k}{N}\wh{\Pi_{k}\pa{\nu_{N-1}}}\pa{n_1,\dots,n_{k}}}$$
	$$\underset{N\to\infty}{\longrightarrow} \wh{\eta_{\infty}}\pa{\sum_{j=1}^k n_j}\;\frac{\wh{\zeta_{k,\infty}}\pa{\sum_{j=1}^k n_j}}{k}\sum_{l=1}^k \wh{\varrho_{k-1,\infty}}\pa{n_1,\dots,\wt{n_l},\dots,n_k}$$
	$$=\frac{\wh{\eta_{\infty}\ast \zeta_{k,\infty}}\pa{\sum_{j=1}^k n_j}}{k}\sum_{l=1}^k \wh{\varrho_{k-1,\infty}}\pa{n_1,\dots,\wt{n_l},\dots,n_k}.$$
	
	When $k=1$ we find that 
	$$\wh{\Pi_1\pa{\mu_N}}\pa{n_1} =
   \wh{\eta_N}\pa{n_1}\pa{\frac{1}{N} +\frac{N-1}{N}\wh{\Pi_{1}\pa{\nu_{N-1}}}\pa{n_1}}\underset{N\to\infty}{\longrightarrow} \wh{\eta_\infty}\pa{n_1}\wh{\zeta_{1,\infty}}(n_1) = \wh{\eta_\infty\ast \zeta_{1,\infty}}(n_1).$$
	Combining the above, and using Lemma \ref{lem:partial_order_in_fourier} and the uniqueness of Fourier coefficients, we conclude that
	$$\Pi_k\pa{\mu_N}\pa{\theta_1,\dots,\theta_k}\underset{N\to\infty}{\overset{\text{weak}}{\longrightarrow}} \begin{cases}
		\eta_{\infty}\ast \zeta_{1,\infty} (\theta_1),& k=1,\\
		\frac{1}{k}\sum_{l=1}^k \eta_{\infty}\ast \zeta_{k,\infty}(\theta_l)\rho_{k-1,\infty}\pa{\theta_1-\theta_l,\dots,\widetilde{\theta_l-\theta_l},\dots, \theta_k-\theta_l}, & k\geq 2,
	\end{cases}$$
	which gives the general result. The fact that 
	$$\wh{\frac{d\theta}{2\pi}\ast \frac{d\theta}{2\pi}}(n)  = \wh{\frac{d\theta}{2\pi}}^2(n)=\delta_0(n)^2 =\delta_0(n)=\wh{\frac{d\theta}{2\pi}}(n),$$
	and
	$$\wh{\delta}(n) \wh{\zeta_{k,\infty}}(n)=\wh{\zeta_{k,\infty}}(n),$$
	show the last two statements.
\end{proof}

Following on the above, and motivated by the fact that ordered states are automatically partially-ordered, it is natural to enquire if we can have a partially ordered state in the setting of the CL model which ``approximate'' ordered states, i.e. a state whose marginals converge to
$$\frac{1}{k}\sum_{i=1}^k \eta\pa{\theta_i}\prod_{j\not=i}^{k} \nu\pa{\theta_j-\theta_i}$$
for some measures $\eta,\nu\in \PP\pa{\T}$, where $\nu$ is even and somewhat concentrated around $\theta=0$. The next lemma sheds light on this possibility.

\begin{lemma}\label{lem:no_decoupled_partial_order}
	Let $\br{F_N}_{N\in\N}$ be a family of symmetric probability measure on $\T^N$ such that its marginals, $\br{F_{N,k}}_{N\in\N}$ converges weakly to $f_k\in\PP\pa{\T^k}$ as $N$ goes to infinity for any $k\in\N$. Assume in addition that for any $k\in\N\setminus\br{1}$ there exist $\eta_k\in\PP\pa{\T}$ and an even probability measure $\nu_k\in \PP\pa{\T}$ such that 
	\begin{equation}\label{eq:strong_partial_order}
		f_k\pa{\theta_1,\dots,\theta_k} = \frac{1}{k}\sum_{i=1}^k \eta_k\pa{\theta_i} \prod_{j\not=i}\nu_k\pa{\theta_j-\theta_i},\qquad k\geq 2.
	\end{equation}
	Then for any $j,k\in\N\setminus\br{1}$ with $j\leq k$ we have that\footnote{Recall that we use the convention that the product over an empty set is $1$.} 
		\begin{equation}\label{eq:connection_between_marginals_special_choice_sum}
		\begin{aligned}
			k&\pa{\prod_{m=1}^{j-1}\wh{\nu_j}\pa{n_m}+\wh{\nu_j}\pa{\sum_{m=1}^{j-1}n_m}\sum_{l=1}^{j-1}\prod_{m=1,\;m\not=l}^{j-1}\wh{\nu_j}\pa{n_m}}	\\
			=&j\pa{\prod_{m=1}^{j-1}\wh{\nu_k}\pa{n_m}+\wh{\nu_k}\pa{\sum_{m=1}^{j-1}n_m}\sum_{l=1}^{j-1}\prod_{m=1,\;m\not=l}^{j-1}\wh{\nu_k}\pa{n_m}+\pa{k-j}\wh{\nu_k}\pa{\sum_{m=1}^{j-1}n_m}\prod_{m=1}^{j-1}\wh{\nu_k}\pa{n_m}}\\
		\end{aligned}
	\end{equation}
	for any $n_1,\dots,n_{j-1}\in\Z$. 
	In particular
	\begin{equation}\label{eq:strong_partial_condition_j=2}
		k\wh{\nu_2}(n) = 2\wh{\nu_k}(n)+\pa{k-2}\wh{\nu_k}^2(n)
	\end{equation}
	for any $k\geq 2$ and $n\in\Z$, and 
		\begin{equation}\label{eq:strong_partial_condition_j=3}
			\begin{split}
				k&\pa{\wh{\nu_3}(n_1)\wh{\nu_3}(n_2)+\wh{\nu_3}(n_1+n_2)\pa{\wh{\nu_3}(n_1)+\wh{\nu_3}(n_2)}}\\
				=& 3\pa{\wh{\nu_k}(n_1)\wh{\nu_k}(n_2)+\wh{\nu_k}(n_1+n_2)\pa{\wh{\nu_k}(n_1)+\wh{\nu_k}(n_2)}+(k-3)\wh{\nu_k}(n_1+ n_2)\wh{\nu_k}(n_1)\wh{\nu_k}(n_2)}
			\end{split}
	\end{equation}
	for any $k\geq 3$ and $n_1,n_2\in\Z$.

	Consequently, if $~\nu_k=\nu$ for all $k\in\N\setminus\br{1}$ then $\wh{\nu}(n)=0\text{ or }1$ for all $n\in\N$. Moreover, if there exists $n_0\in\Z\setminus \br{0}$ such that $\wh{\nu}\pa{n_0}=1$ then 
	$$\wh{\nu}\pa{ln_0}=1,\qquad \forall l\in\Z.$$	
	We conclude that if $\nu\ne \frac{d\theta}{2\pi}$ then $\nu$ can't be absolutely continuous with respect to $\frac{d\theta}{2\pi}$.
\end{lemma}

\begin{proof}
	We start by noticing that $\br{f_k}_{k\in\N}$ are symmetric by their definition\footnote{The fact that $\br{F_N}_{N\in\N}$ are symmetric will also imply that the weak limit of their marginals must be symmetric.}. For any $j\leq k$ and any $\pa{n_1,\dots,n_j}$ we have that 
	\begin{equation}\label{eq:mariginal_connection}
		\begin{split}
			\wh{f_j}&\pa{n_1,\dots,n_j} = \lim_{N\to\infty} \int_{\T^N}e^{-i\sum_{l=1}^j n_l \theta_l}dF_N\pa{\theta_1,\dots,\theta_N}\\
			&= \lim_{N\to\infty} \int_{\T^N}e^{-i\sum_{l=1}^j n_l \theta_l}d\Pi_k\pa{F_N}\pa{\theta_1,\dots,\theta_k}= \wh{f_k}\pa{n_1,\dots,n_j,0,\dots,0}.
		\end{split}
	\end{equation}
	Under the assumption that \eqref{eq:strong_partial_order} holds, and using Lemma \ref{lem:fourier_of_partial_order} and the fact that 
	$$\wh{\otimes_{r=1}^l \xi}\pa{n_1,\dots,n_l} = \prod_{r=1}^l \wh{\xi}\pa{n_r}$$
	for any $\xi\in \PP\pa{\T}$ we find that \eqref{eq:mariginal_connection} implies that
	\begin{equation}\nonumber 
		\frac{\wh{\eta_j}\pa{\sum_{i=1}^j n_i}}{j}\sum_{l=1}^j  \prod_{m=1,\;m\not=l}^j\wh{\nu_j}\pa{n_{m}}=\frac{\wh{\eta_k}\pa{\sum_{i=1}^j n_i}}{k}\pa{\sum_{l=1}^j  \prod_{m=1,\;m\not=l}^{j}\wh{\nu_k}\pa{n_{m}} + \pa{k-j}\prod_{m=1}^{j}\wh{\nu_k}\pa{n_{m}}},
	\end{equation}
	where we have used the fact that $\wh{\nu}(0)=1$ for any $\nu\in\PP\pa{\T}$. 
	
	Choosing $n_j=-\sum_{i=1}^{j-1}n_i$ we find that 
	\begin{equation}\nonumber 
		\begin{aligned}
			\frac{1}{j}&\pa{\prod_{m=1}^{j-1}\wh{\nu_j}\pa{n_m}+\wh{\nu_j}\pa{-\sum_{m=1}^{j-1}n_m}\sum_{l=1}^{j-1}\prod_{m=1,\;m\not=l}^{j-1}\wh{\nu_j}\pa{n_m}}	\\
			=&\frac{1}{k}\pa{\prod_{m=1}^{j-1}\wh{\nu_k}\pa{n_m}+\wh{\nu_k}\pa{-\sum_{m=1}^{j-1}n_m}\sum_{l=1}^{j-1}\prod_{m=1,\;m\not=l}^{j-1}\wh{\nu_k}\pa{n_m}+\pa{k-j}\wh{\nu_k}\pa{-\sum_{m=1}^{j-1}n_m}\prod_{m=1}^{j-1}\wh{\nu_k}\pa{n_m}}.
		\end{aligned}
	\end{equation}
	In particular, using the fact that any even probability measure $\xi$ satisfies $\wh{\xi}(n)=\wh{\xi}(-n)$ gives us \eqref{eq:connection_between_marginals_special_choice_sum} which implies \eqref{eq:strong_partial_condition_j=2} and \eqref{eq:strong_partial_condition_j=3} by plugging $j=2$ and $j=3$ respectively.
	
	Assuming that $\nu_k=\nu$ for all $k\in\N\setminus\br{1}$ and utilising \eqref{eq:strong_partial_condition_j=2} we find that 
	$$k\wh{\nu}(n) = 2\wh{\nu}(n)+\pa{k-2}\wh{\nu}^2(n). $$
	for any $k\geq 3$. This implies that
	$$\wh{\nu}(n)=\wh{\nu}^2(n)$$
	for all $n\in\Z$ from which we conclude that $\wh{\nu}(n)=0$ or $\wh{\nu}(n)=1$ for any $n\in\Z$.
	
	Next we notice that in this case \eqref{eq:strong_partial_condition_j=3} for any $k\geq 4$ can be simplified to
	$$\wh{\nu}(n_1)\wh{\nu}(n_2) + \wh{\nu}\pa{n_1+n_2}\pa{\wh{\nu}(n_1)+\wh{\nu}(n_2) }=3\wh{\nu}(n_1+n_2) \wh{\nu}(n_1)\wh{\nu}(n_2). $$
	Consequently:
	$$\wh{\nu}(n+m) = \begin{cases}
		1,& \wh{\nu}(n)=\wh{\nu}(m)=1,\\
		0,& \wh{\nu}(n)\wh{\nu}(m)=0\ \ \ \text{ and } \ \ \ \wh{\nu}(n)+\wh{\nu}(m)\ne 0.\\
	\end{cases}$$
	This implies that if there exists $n_0\in\Z\setminus\br{0}$ such that $\wh{\nu}\pa{n_0}=1$ then $\wh{\nu}\pa{ln_0}=1$ for all $l\in\N$. Since $\nu$ is even and $\wh{\nu}(0)=1$ we find that in this case
	$$\wh{\nu}\pa{ln_0}=1,\qquad \forall l\in\Z.$$
	Consequently 
	$$\lim_{l\to\pm\infty}\wh{\nu}\pa{ln_0}\not=0,$$
	showing that $\nu$ can't be absolutely continuous with respect to $\frac{d\theta}{2\pi}$ in that case. As $\wh{\nu}(n) = \delta_0(n)$ implies that $\nu=\frac{d\theta}{2\pi}$, we conclude the proof.
\end{proof}

\begin{remark}\label{rem:restricting_conditions}
	We would like to point out two observations about Lemma \ref{lem:no_decoupled_partial_order} and its proof:
	\begin{itemize}
		\item We have not truly needed \eqref{eq:connection_between_marginals_special_choice_sum} for all $j$ and $k$ to conclude the second part of the lemma. We can achieve the same result by using \eqref{eq:strong_partial_condition_j=2} and $k=3$, and \eqref{eq:strong_partial_condition_j=3} and $k=4$.
		\item 	The identity \eqref{eq:connection_between_marginals_special_choice_sum} is extremely restrictive -- and it doesn't even include $\mu_k$! It seems unlikely that a partially ordered state of the form \eqref{eq:strong_partial_order} exists in the setting of the CL model, but the authors have elected to leave further investigation of this to a later time.
	\end{itemize}
\end{remark}

\begin{remark}\label{rem:f_infty_not_tensorised}
	Theorem \ref{thm:main_partial_order} states that $\br{f_{N,\infty}}_{N\in\N}$, the family of generated limiting partially ordered states of the CL model, is $\pa{\frac{d\theta}{2\pi},\nu_{k-1}}_{k\in\N}-$partially ordered for a family of absolutely continuous measures $\br{\nu_{k}}_{k\in\N}$. Lemma \ref{lem:no_decoupled_partial_order} implies that $\nu_{k}$ is not of the form $\otimes_{i=1}^k \nu$ for some $\nu\in\PP\pa{\T}$. 
\end{remark}

\subsection{On the Fourier coefficients of $g_{\epsilon_N}$}
 
 The study of the rescaled CL master equation \eqref{eq:master_CL_rescaled} relies heavily on $g_{\epsilon_N}$ and in particular on the asymptotic behaviour of its Fourier coefficient. Our study of  the ``low'' frequencies of $\wh{g_{\epsilon_N}}$ will be invaluable in understanding the limit equations of $\br{F_{N,k}(t)}_{N\in\N}$, expressed via \eqref{eq:limit_of_F_N_K_partial_order}, while the of the ``high'' frequencies of $\wh{g_{\epsilon_N}}$ will be imperative for the proof of Theorem \ref{thm:quantitative_convergence}.
 
 \begin{lemma}\label{lem:on_g_low_and_high_frequencies}
 	Let $g$ be a even probability density on $\pa{\R,dx}$ such that its $l$-th moment, defined as
 	$$m_l = \int_{\R}\abs{x}^l g(x)dx,$$
 	is finite for some $l\in \N \setminus \br{1,2}$. 
 	\begin{enumerate}[(i)]
 		\item\label{item:g_coefficient_real_and_bounded_by_1} For any $n\in\Z$ we have that $\wh{g_{\epsilon_N}}(n)\in\R$ and $-1\leq \wh{g_{\epsilon_N}}(n)\leq 1$.
 		\item\label{item:low_frequencies} For any $\epsilon_N<\frac{\pi}{\sqrt[l]{m_l}}$ and any $n\in\Z$
 		\begin{equation}\label{eq:estimation_for_fourier_of_g_eps_low}
 			\abs{\widehat{g_{\epsilon_N}}(n) - 1 + \frac{m_2 \epsilon_N^2 n^2}{2}} \leq \uptau_N+\begin{cases}
 				\frac{m_3}{3}\epsilon_N^3 \abs{n}^3, & l=3,\\
 				\frac{m_4}{12}\epsilon_N^4 n^4, & l>3,
 			\end{cases}
 		\end{equation}
 		where
 		\begin{equation}\label{eq:uptau_N}
 			\uptau_N = 	\frac{2\epsilon_N^l m_l}{\pi^l-\epsilon_N^l m_l} 
 		\end{equation}
 		\item\label{item:high_frequencies} If in addition we have that $g\in L^p\pa{\R}$ for some $p>1$ then for any positive sequence $\br{\alpha_N}_{N\in\N}$ such that 
 		\begin{equation}\nonumber 
 			\alpha_N \leq 4^{q+1}\norm{g}_{L^p}^q,
 		\end{equation}
 		with $q$ being the H\"older conjugate of $p$, we have that when $\abs{n} \geq \frac{\alpha_N}{\epsilon_N}$
 		\begin{equation}\label{eq:upper_bound_away_from_zero}
 			\wh{g_{\epsilon_N}(n)}-1 \leq \uptau_N  - \frac{\alpha_N^2 \pi^2}{ 2\cdot 4^{2\pa{q+1}+1}\norm{g}_{L^p}^{2q}\pa{\sqrt[l]{4m_l}\alpha_N+2\pi}^2},
 		\end{equation}
 		as long as $\epsilon_N  < \frac{\pi}{\sqrt[l]{m_l}}$.
 	\end{enumerate}
 \end{lemma}
 
\begin{proof}
	\leavevmode
	\begin{enumerate}[(i)]
		\item The claim follows form the fact that $g_{\epsilon_N}$ is a real and even probability density. 
		\item The proof can be found in the discussion about $g_{\epsilon_N}$ in \cite[Lemma 18, Lemma 23, and Lemma 24]{E2024}.
		\item The proof is inspired form similar considerations in \cite{CCLLV2010} and is left to Appendix \S\ref{app:additional_proofs}. 
	\end{enumerate}
\end{proof}
 
 With the tools we've developed in this section we are now ready to start exploring the notion of partial order in the CL model.
\section{Generation of partial order}\label{sec:partial_order}

In this section we will explore the generation of partial order in our rescaled CL master equation \eqref{eq:master_CL_rescaled} in the case when $N\epsilon_N^2=1$. We start by finding a recursive formula for the weak limits of the marginals of the family of solutions to our equation. 

\begin{theorem}\label{thm:recursive_formula_for_sol}
Let $\br{F_N(t)}_{N\in\N}$ be the family of symmetric solutions to \eqref{eq:master_CL_rescaled} with initial data $\br{F_{N}(0)}_{N\in\N}$. Assume in addition that $N\epsilon_N^2=1$, the interaction generating function $g$ has a finite moment of order $l \in \N\setminus \br{1,2}$, and $\br{F_{N,k}\pa{0}}_{N\in\N}$ converges weakly as $N$ goes to infinity to a family $f_{k,0} \in \PP\pa{\mathcal{T}^k}$ for any $k\in\N$. 
Then for any $t>0$ and any $k\in\N$, $\br{F_{N,k}(t)}_{N\in\N}$ converges weakly as $N$ goes to infinity to a symmetric family $f_k(t)\in \PP\pa{\T^k}$ whose Fourier coefficients are given by \eqref{eq:limit_of_F_N_K_partial_order}. In other words,
\begin{equation}\nonumber 
	\begin{aligned}
		&\wh{f_k}\pa{n_1,\dots,n_k,t}= e^{-\frac{\lambda \pa{2k\pa{k-1}+m_2\pa{\sum_{r=1}^k n_r^2}}t}{2}}\wh{f_{k,0}}\pa{n_1,\dots,n_k}\\
		&+2\lambda\sum_{i<j\leq k}\int_{0}^t e^{-\frac{\lambda \pa{2k\pa{k-1}+m_2\pa{\sum_{r=1}^k n_r^2}}\pa{t-s}}{2}}\wh{f_{k-1}}\pa{n_1,\dots,\wt{n_i},\dots,n_i+n_j,\dots, \wt{n_j},\dots, n_k,s}ds.
	\end{aligned}
\end{equation}
\end{theorem}

The proof of the above relies on the following lemma, whose proof can be found in \cite[Corollary 21]{E2024}:

\begin{lemma}
	A recursive formula for the $k$-th marginals of a family of solution to \eqref{eq:master_CL_rescaled}, $\br{F_{N,k}(t)}_{N\in\N}$ is given by
	\begin{equation}\label{eq:recursive}
		\begin{aligned}
			&\wh{F_{N,k}}\pa{n_1,\dots,n_k,t} = e^{-\frac{\lambda N}{N-1}\pa{\pa{N-k}\sum_{r=1}^{k}\pa{1-\wh{g_{\epsilon_N}}\pa{n_r}}+k\pa{k-1}}t}\wh{F_{N,k}}\pa{n_1,\dots,n_k,0}\\ 		
			&+\frac{\lambda N}{N-1}\sum_{i<j\leq k}\pa{\widehat{g_{\epsilon_N}}\pa{n_i}+\widehat{g_{\epsilon_N}}\pa{n_j}}\int_{0}^t e^{-\frac{\lambda N}{N-1}\pa{\pa{N-k}\sum_{r=1}^{k}\pa{1-\wh{g_{\epsilon_N}}\pa{n_r}}+k\pa{k-1}}\pa{t-s}} \\
			&\qquad\qquad\qquad\qquad\qquad\qquad\qquad\wh{F_{N,k-1}}\pa{n_1,\dots,\wt{n_i},\dots n_i+n_j,\dots,\wt{n_j},\dots,n_k,s}ds.
		\end{aligned}
	\end{equation}
\end{lemma}

\begin{proof}[Proof of Theorem \ref{thm:recursive_formula_for_sol}]
	\sloppy We start by noting the following general principle on $\T^k$: if $\br{\mu_N}_{N\in\N}$ is a family of probability measures on $\T^k$ such that $\lim_{N\to\infty}\wh{\mu_N}\pa{n_1,\dots,n_k}$ exists for all $\pa{n_1,\dots,n_k}\in\Z^k$, then there exists $\mu\in \PP\pa{\T^k}$ such that $\br{\mu_N}_{N\in\N}$ converges weakly to $\mu$ and 
	$$\wh{\mu}\pa{n_1,\dots,n_k}=\lim_{N\to\infty}\wh{\mu_N}\pa{n_1,\dots,n_k}.$$
	Indeed, the functional
	$$T(p) = \lim_{N\to\infty}\int_{\T^k}p\pa{\theta_1,\dots,\theta_k}d\mu_N\pa{\theta_1,\dots,\theta_k} $$
	is a well defined linear functional on the space of trigonometric polynomials on $\T^k$ due to our assumption on the convergence of the Fourier coefficients of $\br{\mu_N}_{N\in\N}$. Moreover, $T$ is a positive bounded functional with respect to the supremum norm. Consequently, $T$ can be extended to a positive bounded linear functional on $\pa{C\pa{\T^k},\norm{\cdot}_{\infty}}$\footnote{To show the positivity, assume that $f\geq 0$ and fix $\epsilon>0$. We can find a sequence of polynomials $\br{p_{k,\epsilon}}_{k\in\N}$ such that $\norm{p_{k,\epsilon}-\pa{f+\epsilon}}_{\infty}\underset{k\to\infty}{\longrightarrow}0$. Consequently, $p_{k,\epsilon}$ will be a positive polynomial for $k$ large enough and
	$$T(f)+\epsilon=T\pa{f+\epsilon} = \lim_{k\to\infty} T\pa{p_{k,\epsilon}}\geq 0.$$
	As $\epsilon>0$ was arbitrary, we conclude the desired positivity.}. Using the Riesz-Markov theorem and the fact that $T(1)=1$, we conclude that there exists $\mu\in \PP\pa{\T^k}$ such that
	$$T(f) = \int_{\T^k}f\pa{\theta_1,\dots,\theta_k}d\mu\pa{\theta_1,\dots,\theta_k}.$$
	As this implies that for any trigonometric polynomial $p$
	$$\lim_{N\to\infty}\int_{\T^k}p\pa{\theta_1,\dots,\theta_k}d\mu_N\pa{\theta_1,\dots,\theta_k} =\int_{\T^k}p\pa{\theta_1,\dots,\theta_k}d\mu\pa{\theta_1,\dots,\theta_k}$$
	we conclude, just like in the proof of Lemma \ref{lem:properties_of_even_and_symmetry_in_fourier}, that $\br{\mu_N}_{N\in\N}$ converges weakly to $\mu$. This, in turn, implies that  
	$$\wh{\mu}\pa{n_1,\dots,n_k}=\lim_{N\to\infty}\wh{\mu_N}\pa{n_1,\dots,n_k},$$
	as claimed.
	
	Following the above principle, we see that in order to show the existence of the family of probability measures $\br{f_k(t)}_{k\in\N}$ we only need to show the convergence of the Fourier coefficients of $\br{F_{N,k}(t)}_{N\in\N}$ for any $k\in\N$ and any $t>0$. We will show this by induction.  
	
	For $k=1$ \eqref{eq:recursive} reads as
	\begin{equation}\nonumber
		\begin{aligned}
			&\wh{F_{N,1}}\pa{n_1,t} = e^{-\lambda N\pa{1-\wh{g_{\epsilon_N}}\pa{n_1}}t}\wh{F_{N,1}}\pa{n_1,0}.
		\end{aligned}
	\end{equation}
	Using \eqref{eq:estimation_for_fourier_of_g_eps_low} from Lemma \ref{lem:on_g_low_and_high_frequencies} and the fact that $N\epsilon_N^2=1$ we find that 
	\begin{equation}\label{eq:limit_of_fourier_g}
		\lim_{N\to\infty}N\pa{1-\wh{g_{\epsilon_N}}\pa{n}} =\frac{m_2 n^2}{2}. 
	\end{equation}
	Consequently, since $\br{F_{N,1}}_{N\in\N}$ converges weakly to $f_{1,0}$, we find that 
	$$\lim_{N\to\infty}\wh{F_{N,1}}\pa{n_1,t} = e^{-\frac{\lambda m_2 n_1^2 t}{2}}\wh{f_{1,0}}(n_1). $$
	Assume that the claim holds for some $k\in\N$ and consider \eqref{eq:recursive} again. Using the induction assumption, together with the facts that $\br{F_{N,k}(0)}_{N\in\N}$ converges weakly to $f_{k,0}$,
	$$\abs{e^{-\frac{\lambda N}{N-1}\pa{\pa{N-k}\sum_{r=1}^{k}\pa{1-\wh{g_{\epsilon_N}}\pa{n_r}}+k\pa{k-1}}\pa{t-s}} \wh{F_{N,k-1}}\pa{n_1,\dots, n_i+n_j,\dots,n_k,s}}$$
	$$\leq  \abs{\wh{F_{N,k-1}}\pa{n_1,\dots, n_i+n_j,\dots,n_k,s}} \leq 1,\qquad\qquad 0\leq s\leq t,$$
	and the dominated convergence theorem we find that
	\begin{equation}\nonumber 
		\begin{aligned}
			&\lim_{N\to\infty}\wh{F_{N,k}}\pa{n_1,\dots,n_k,t}= e^{-\frac{\lambda\pa{2k\pa{k-1}+m_2\sum_{r=1}^k n_r^2}t}{2}}\wh{f_{k,0}}\pa{n_1,\dots,n_k}\\
			&+2\lambda\sum_{i<j\leq k}\int_{0}^t e^{-\frac{\lambda\pa{2k\pa{k-1}+m_2\sum_{r=1}^k n_r^2}\pa{t-s}}{2}}\wh{f_{k-1}}\pa{n_1,\dots,\wt{n_i},\dots,n_i+n_j,\dots, \wt{n_j},\dots, n_k,s}ds.
		\end{aligned}
	\end{equation}
	showing the desired recursive formula. 
	
	The fact that $\br{F_{N,k}(t)}_{N\in\N}$ is a family of symmetric probability measures that converges weakly to $f_k(t)$ implies that $f_k(t)$ is also a symmetric probability measure, which concludes the proof.
\end{proof}	
	
\begin{remark}\label{rem:L_infty_bound_on_coefficient_of_probability_measure}
	Note that in the above proof we have used the fact that for any probability measure $\mu\in \PP\pa{\T^k}$
	$$\abs{\wh{\mu}\pa{n_1,\dots,n_k}} \leq \mu\pa{\T^k}=1,$$
	for any $\pa{n_1,\dots,n_k}\in\Z^k$. We will use this estimate throughout the paper many times. 
\end{remark}	

With \eqref{eq:limit_of_F_N_K_partial_order} at hand we are now able to explore the behaviour of our marginal solutions as $t$ goes to infinity. We start by considering $\wh{f_1}(t)$ and $\wh{f_2}(t)$ to get the general intuition. 

We have seen in the proof of Theorem \ref{thm:recursive_formula_for_sol} that
\begin{equation}\label{eq:about_f_1}
	\wh{f_1}(n_1,t) = e^{-\frac{\lambda m_2 n_1^2 t}{2}}\wh{f_{1,0}}(n_1).
\end{equation}
Using \eqref{eq:limit_of_F_N_K_partial_order} with the fact that for any $\alpha,\beta\in\R$
\begin{equation}\label{eq:conv_of_exp}
	\int_{0}^{t}e^{-\alpha\pa{t-s}}e^{-\beta s}ds = \begin{cases}
		te^{-\alpha t}, & \alpha=\beta,\\
		\frac{e^{-\beta t}-e^{-\alpha t}}{\alpha-\beta},& \alpha\ne \beta,
	\end{cases}
\end{equation}
We find that 
	\begin{equation}\label{eq:about_f_2}
	\begin{aligned}
		&\wh{f_2}\pa{n_1,n_2,t}= e^{-\frac{\lambda \pa{4+m_2\pa{n_1^2+n_2^2}}t}{2}}\wh{f_{2,0}}\pa{n_1,n_2}\\
		&\qquad+2\lambda \int_{0}^t e^{-\frac{\lambda \pa{4+m_2\pa{n_1^2+n_2^2}}\pa{t-s}}{2}}e^{-\frac{\lambda m_2 \pa{n_1+n_2}^2}{2} s}\wh{f_{1,0}}(n_1+n_2)ds = e^{-\frac{\lambda \pa{4+m_2\pa{n_1^2+n_2^2}}t}{2}}\wh{f_{2,0}}\pa{n_1,n_2}\\
		&+\wh{f_{1,0}}\pa{n_1+n_2}\;\begin{cases}
			2\lambda te^{-\frac{\lambda \pa{4+m_2 \pa{n_1^2+n_2^2}} t}{2} }, & 4+m_2\pa{n_1^2+n_2^2}= m_2\pa{n_1+n_2}^2,\\
			\frac{4\pa{e^{-\frac{\lambda m_2 \pa{n_1+n_2}^2 t}{2} }-e^{-\frac{\lambda \pa{4+m_2\pa{n_1^2+n_2^2}}t}{2}}}}{4+m_2\pa{n_1^2+n_2^2-\pa{n_1+n_2}^2}}, & 4+m_2\pa{n_1^2+n_2^2}\ne m_2\pa{n_1+n_2}^2.
		\end{cases}
	\end{aligned}
\end{equation}
A closer look at the Fourier coefficients of $f_{2}(t)$ reveals two distinct behaviours as $t$ goes to infinity:
\begin{itemize}
	\item Most of the terms decay exponentially to zero uniformly in $\pa{n_1,n_2}\in\Z^2$ with rate of at most $C_2\pa{1+t}e^{-\frac{\lambda\pa{4+m_2}t}{2}}$.
	\item The term $e^{-\frac{\lambda m_2 \pa{n_1+n_2}^2 t}{2}}$ decays exponentially to zero \textit{when $n_1\ne -n_2$}  with rate of at least $e^{-\frac{\lambda m_2 t}{2}}$. When $ n_1+n_2=0$, however, the term becomes the constant $1$. 
\end{itemize}
Plugging $\wh{f_2}(t)$ in \eqref{eq:limit_of_F_N_K_partial_order} will produce a similar behaviour -- most terms will decay exponentially to zero uniformly in $\pa{n_1,\dots,n_k}\in\Z^k$ but we will always get a term of the form $e^{-\frac{\lambda m_2 \pa{\sum_{r=1}^k n_r}^2 t}{2}}$ which will give us the constant $1$ when $\sum_{r=1}^kn_r=0$. 

An explicit and quantitative version of this observation is given by the following lemma:

\begin{lemma}\label{lem:explicit_behaviour_for_recursive}
	Let $\br{f_k(t)}_{k\in\N}$ be a family of symmetric probability densities on $\br{\PP\pa{\T^k}}_{k\in\N}$ given by \eqref{eq:limit_of_F_N_K_partial_order}. 
	\begin{enumerate}[(i)]
		\item\label{item:recursive_a_b} For any $k\in \N$ we have that 
		\begin{equation}\label{eq:explicit_behaviour_for_f_k}
			\begin{split}
				&\wh{f_k}\pa{n_1,\dots,n_k,t} = \sum_{h=2}^k e^{-\lambda h\pa{h-1}t}b_{h,k}\pa{n_1,\dots,n_k,t} 	\\
				&\qquad\qquad + a_k\pa{n_1,\dots,n_k}\wh{f_{1,0}}\pa{\sum_{r=1}^k n_r}e^{-\frac{\lambda m_2 \pa{\sum_{r=1}^k n_r}^2 t}{2}}.	
			\end{split}
		\end{equation} 
		 where 
			\begin{equation}\label{eq:a11}
			\begin{split}
				&a_1(n_1)=1,\\
			\end{split}
		\end{equation}
		and $a_k$ and $\br{b_{h,k}(t)}_{h=2, \dots, k}$ are defined recursively for $k\geq 2$ by
		\begin{equation}\label{eq:recursive_a_zero}
			\begin{split}
				& a_{k}\pa{n_1,\dots,n_{k}} =0
			\end{split}
		\end{equation}
		when $~\displaystyle 2k(k-1) + m_2\sum_{r=1}^{k}n_r^2= m_2\pa{\sum_{r=1}^{k}n_r}^2$ and
		\begin{equation}\label{eq:recursive_a}
			\begin{split}
				& a_{k}\pa{n_1,\dots,n_{k}} =
				\frac{4}{2k(k-1)+m_2\pa{\sum_{r=1}^{k}n_r^2 - \pa{\sum_{r=1}^{k}n_r}^2}}\\
				&\qquad\qquad\sum_{i<j\leq k}	a_{k-1}\pa{n_1,\dots,\wt{n_i},\dots,n_i+n_j,\dots, \wt{n_j},\dots, n_{k}}
			\end{split}
		\end{equation}
		when $~\displaystyle 2k(k-1) + m_2\sum_{r=1}^{k+1}n_r^2\ne m_2\pa{\sum_{r=1}^{k+1}n_r}^2$ and where the position of $n_i+n_j$ is the same as in the recursive formula \eqref{eq:limit_of_F_N_K_partial_order},
		\begin{equation}\label{eq:b_last_term}
			\begin{split}
				&b_{k,k}\pa{n_1,\dots,n_{k},t}=e^{-\frac{\lambda m_2\pa{\sum_{r=1}^{k} n_r^2} t} {2}}\wh{f_{k,0}}\pa{n_1,\dots,n_{k}}+e^{-\frac{\lambda m_2\pa{\sum_{r=1}^{k}n_r^2}t}{2}}\wh{f_{1,0}}\pa{\sum_{r=1}^{k}n_r}\\
				&\qquad\pa{\sum_{i<j\leq k}	a_{k-1}\pa{n_1,\dots,\wt{n_i},\dots,n_i+n_j,\dots, \wt{n_j},\dots, n_{k}}}\\
				&\qquad\qquad\begin{cases}
					2\lambda t, & 2k(k-1) + m_2\sum_{r=1}^{k}n_r^2= m_2\pa{\sum_{r=1}^{k}n_r}^2,\\
					-\frac{4}{2k(k-1)+m_2\pa{\sum_{r=1}^{k}n_r^2-\pa{\sum_{r=1}^{k}n_r}^2}}, & 2k(k-1) + m_2\sum_{r=1}^{k}n_r^2\ne m_2\pa{\sum_{r=1}^{k}n_r}^2,
				\end{cases}\\
			\end{split}
		\end{equation}
			\begin{equation}\label{eq:recursive_b}
			\begin{split}
				&b_{h,k}\pa{n_1,\dots,n_{k},t}=2\lambda\sum_{i<j\leq k}\int_{0}^t e^{-\frac{\lambda \pa{2\rpa{k\pa{k-1}-h(h-1)}+m_2\pa{\sum_{r=1}^{k} n_r^2}}\pa{t-s}}{2}}\\
				&\qquad\qquad\qquad\qquad\qquad  b_{h,k-1}\pa{n_1,\dots,\wt{n_i},\dots,n_i+n_j,\dots, \wt{n_j},\dots, n_{k},s} ds,\quad 2\leq h\leq k-1
			\end{split}
		\end{equation}
		\item\label{item:some_properties_of_a} Defining
		\begin{equation}\label{eq:def_ell_for_upper_bound_for_a}
			\ell_1=1,\quad \ell_{k}=\begin{cases}
				\frac{2k(k-1)}{m_2},& \frac{2k(k-1)}{m_2}\in\Z,\\
				\frac{2k(k-1)}{\min\pa{2k(k-1)-m_2\lfloor \frac{2k(k-1)}{m_2}\rfloor,m_2\pa{\lfloor \frac{2k(k-1)}{m_2}\rfloor+1}-2k(k-1)}},& \frac{2k(k-1)}{m_2}\not\in\Z,
			\end{cases}
		\end{equation}
	 when $k\geq 2$ and where $\lfloor x \rfloor$ is the closest integer to $x$ such that $\lfloor x \rfloor \leq x$, we have that 
		\begin{equation}\label{eq:uniform_bound_on_a}
			\sup_{\pa{n_1,\dots,n_{k}}\in\Z^{k}}\abs{a_{k}\pa{n_1,\dots,n_{k}}} \leq \prod_{j=1}^{k}\ell_{j}.
		\end{equation}
		In addition
		\begin{equation}\label{eq:value_of_a_at_zero}
			a_k\pa{0,\dots,0}=1,
		\end{equation}
		for any $k\in\N$.
			\item\label{item:some_properties_of_b} We have that for any $k\geq 2$
				\begin{equation}\label{eq:uniform_bound_on_bkk}
				\begin{split}
					&\abs{b_{k,k}\pa{n_1,\dots,n_{k},t}}  \leq e^{-\frac{\lambda m_2\pa{\sum_{r=1}^{k} n_r^2} t}{2}}\\
					&\qquad \pa{1+\begin{cases}
							\lambda k(k-1)\pa{\prod_{j=1}^{k-1} \ell_j}t, & 2k(k-1) + m_2\sum_{r=1}^{k}n_r^2= m_2\pa{\sum_{r=1}^{k}n_r}^2,\\
							\prod_{j=1}^{k}\ell_j, & 2k(k-1) + m_2\sum_{r=1}^{k}n_r^2 \ne m_2\pa{\sum_{r=1}^{k}n_r}^2,
					\end{cases}}
				\end{split}
			\end{equation}
			and
				\begin{equation}\label{eq:uniform_bound_on_b}
				\sup_{\pa{n_1,\dots,n_{k}}\in \Z^k}\abs{b_{h,k}\pa{n_1,\dots,n_{k},t}}  \leq C_{k},\qquad 2\leq h \leq k
			\end{equation}
			where
			\begin{equation}\label{eq:recursive_for_upper_bound_b}
				C_2 =1+\max\pa{\ell_2,\frac{4}{ m_2 e}} ,\quad C_{k}=\max\pa{1+ \pa{\prod_{j=1}^{k-1}\ell_j}\max \pa{\ell_{k}, \frac{2k(k-1)}{ m_2 e}} ,\frac{kC_{k-1}}{2}},
			\end{equation}
			when $k\geq 3$.
	\end{enumerate}
\end{lemma}
	\begin{proof}
		The proofs of \eqref{item:recursive_a_b} and \eqref{item:some_properties_of_a} are done by induction on $k\in\N$. Starting with  \eqref{item:recursive_a_b} we notice that the base cases $a_1$ and $b_{2,2}$ follow from \eqref{eq:about_f_1} and \eqref{eq:about_f_2}.
		
	We continue by assuming that the statement holds for some $k\in\N$. Using \eqref{eq:limit_of_F_N_K_partial_order} and the induction assumption we find that 
	\allowdisplaybreaks
	\begin{align}
   &		\wh{f_{k+1}}\pa{n_1,\dots,n_{k+1},t}=e^{-\frac{\lambda \pa{2k\pa{k+1}+m_2\pa{\sum_{r=1}^{k+1} n_r^2}}t}{2}}\wh{f_{k+1,0}}\pa{n_1,\dots,n_{k+1}}\NNN
   & +2\lambda\sum_{i<j\leq k+1}\int_{0}^t e^{-\frac{\lambda \pa{2k\pa{k+1}+m_2\pa{\sum_{r=1}^{k+1} n_r^2}}\pa{t-s}}{2}}\wh{f_{k}}\pa{n_1,\dots,\wt{n_i},\dots,n_i+n_j,\dots, \wt{n_j},\dots, n_{k+1},s}ds\NNN
   & \qquad\qquad=e^{-\frac{\lambda \pa{2k\pa{k+1}+m_2\pa{\sum_{r=1}^{k+1} n_r^2}}t}{2}}\wh{f_{k+1,0}}\pa{n_1,\dots,n_{k+1}}+2\lambda\sum_{i<j\leq k+1}\int_{0}^t e^{-\frac{\lambda \pa{2k\pa{k+1}+m_2\pa{\sum_{r=1}^{k+1} n_r^2}}\pa{t-s}}{2}} \NNN
   &\qquad\qquad\qquad\Bigg[\sum_{h=2}^k e^{-\lambda h\pa{h-1}s}b_{h,k}\pa{n_1,\dots,\wt{n_i},\dots,n_i+n_j,\dots, \wt{n_j},\dots, n_{k+1},s} \NNN
   &\qquad\qquad+ a_k\pa{n_1,\dots,\wt{n_i},\dots,n_i+n_j,\dots, \wt{n_j},\dots, n_{k+1}}\wh{f_{1,0}}\pa{\sum_{r=1}^{k+1} n_l}e^{-\frac{\lambda m_2 \pa{\sum_{r=1}^{k+1} n_r}^2 s}{2}}	\Bigg]ds \NNN
   &=e^{-\lambda k\pa{k+1}t}e^{-\frac{\lambda m_2\pa{\sum_{r=1}^{k+1} n_r^2} t} {2}}\wh{f_{k+1,0}}\pa{n_1,\dots,n_{k+1}}+2\lambda\sum_{h=2}^{k}e^{-\lambda h(h-1)t}\NNN
   &\pa{\sum_{i<j\leq k+1}\int_{0}^t e^{-\frac{\lambda \pa{2\rpa{k\pa{k+1}-h(h-1)}+m_2\pa{\sum_{r=1}^{k+1} n_r^2}}\pa{t-s}}{2}}b_{h,k}\pa{n_1,\dots,\wt{n_i},\dots,n_i+n_j,\dots, \wt{n_j},\dots, n_{k+1},s} ds}\NNN
   &+2\lambda\sum_{i<j\leq k+1}a_k\pa{n_1,\dots,\wt{n_i},\dots,n_i+n_j,\dots, \wt{n_j},\dots, n_{k+1}}\wh{f_{1,0}}\pa{\sum_{l=1}^{k+1} n_l} \NNN
   &\qquad\qquad\qquad\pa{\int_{0}^t e^{-\frac{\lambda \pa{2k\pa{k+1}+m_2\pa{\sum_{r=1}^{k+1} n_r^2}}\pa{t-s}}{2}}e^{-\frac{\lambda m_2 \pa{\sum_{r=1}^{k+1} n_r}^2 s}{2}}ds}.\nonumber
	\end{align}
		Since
\begin{align}
	& 2\lambda \int_{0}^t e^{-\frac{\lambda \pa{2k\pa{k+1}+m_2\pa{\sum_{r=1}^{k+1} n_r^2}}\pa{t-s}}{2}}e^{-\frac{\lambda m_2 \pa{\sum_{r=1}^{k+1} n_r}^2 s}{2}}ds\NNN
	& =\begin{cases}
		2\lambda te^{-\frac{\lambda \pa{2k\pa{k+1}+m_2\pa{\sum_{r=1}^{k+1} n_r^2}}t}{2}},& 2k\pa{k+1}+m_2\pa{\sum_{r=1}^{k+1} n_r^2}=m_2 \pa{\sum_{r=1}^{k+1} n_r}^2\\
		\frac{4\pa{e^{-\frac{\lambda m_2 \pa{\sum_{r=1}^{k+1} n_r}^2 t}{2}}-e^{-\frac{\lambda \pa{2k\pa{k+1}+m_2\pa{\sum_{r=1}^{k+1} n_r^2}}t}{2}}}}{2k(k+1)+m_2\pa{\sum_{r=1}^{k+1} n_r^2-\pa{\sum_{r=1}^{k+1} n_r}^2}},& 2k\pa{k+1}+m_2\pa{\sum_{r=1}^{k+1} n_r^2}\ne m_2 \pa{\sum_{r=1}^{k+1} n_r}^2\\
	\end{cases},\nonumber
\end{align}
		we conclude the proof by gathering the appropriate terms. \\
		Next we turn our attention to the inductive proof of  \eqref{item:some_properties_of_a}, starting with \eqref{eq:uniform_bound_on_a}. The base case $k=1$ is immediate from \eqref{eq:a11}. We proceed by assuming that the claim holds for some $k\in\N$. 
		
		We notice that if $\frac{2k(k+1)}{m_2} \not\in\Z$ then
		$$\inf_{n\in\Z}\abs{\frac{2k(k+1)}{m_2}-n} = \min\pa{\frac{2k(k+1)}{m_2}-\lfloor \frac{2k(k+1)}{m_2} \rfloor, \lfloor \frac{2k(k+1)}{m_2} \rfloor+1-\frac{2k(k+1)}{m_2}}$$
		and
		$$2k\pa{k+1}+m_2\sum_{r=1}^{k+1}n_r^2 \ne m_2\pa{\sum_{r=1}^{k+1}n_r}^2$$
		for any $\pa{n_1,\dots,n_{k+1}}\in\Z^{k+1}$. Consequently, in this case
		$$\abs{a_{k+1}\pa{n_1,\dots,n_{k+1}}} \leq \frac{4}{m_2\inf_{n\in\Z}\abs{\frac{2k(k+1)}{m_2}-n} }\frac{k(k+1)}{2}\sup_{\pa{p_1,\dots,p_k}\in\Z^k}\abs{a_k\pa{p_1,\dots,p_k}}$$
		$$\leq \frac{2k(k+1)}{m_2\min\pa{\frac{2k(k+1)}{m_2}-\lfloor \frac{2k(k+1)}{m_2} \rfloor, \lfloor \frac{2k(k+1)}{m_2} \rfloor+1-\frac{2k(k+1)}{m_2}}}\prod_{j=1}^k \ell_j=\prod_{j=1}^{k+1} \ell_j,$$
		where we have used our induction assumption. \\
		If there exists $n_0\in\Z$ such that $2k(k+1)=m_2n_0$ then:\\
		\begin{itemize}
			\item If $ \pa{\sum_{r=1}^{k+1}n_r}^2-\sum_{r=1}^{k+1}n_r^2=n_0  $
			$$\abs{a_{k+1}\pa{n_1,\dots,n_{k+1}}}=0 \leq \prod_{j=1}^{k+1} \ell_j.$$
			\item If $\pa{\sum_{r=1}^{k+1}n_r}^2-\sum_{r=1}^{k+1}n_r^2\ne n_0  $  
			$$\abs{a_{k+1}\pa{n_1,\dots,n_{k+1}}}\leq \frac{4}{m_2\inf_{n\in\Z\setminus\br{n_0}}\abs{n_0-n} }\frac{k(k+1)}{2}\sup_{\pa{p_1,\dots,p_k}\in\Z^k}\abs{a_k\pa{p_1,\dots,p_k}}$$
			$$\leq \frac{2k(k+1)}{m_2}\prod_{j=1}^k \ell_j=\prod_{j=1}^{k+1} \ell_j.$$
		\end{itemize}
		As we've covered all possible cases for $\pa{n_1,\dots,n_{k+1}}\in\Z^{k+1}$ we conclude the validity of \eqref{eq:uniform_bound_on_a}.\\
		The proof of \eqref{eq:value_of_a_at_zero} follows from the fact that $k\pa{k-1}\not=0$ for any $k\geq 2$ and as such
		$$a_{k}\pa{0,\dots,0} = \begin{cases}
			1,& k=1,\\
			\frac{4}{2k(k-1)}\sum_{i<j\leq k}a_{k-1}\pa{0,\dots,0},& k\geq 2,
		\end{cases}=\begin{cases}
		1,& k=1,\\
		a_{k-1}\pa{0,\dots,0},& k\geq 2.
		\end{cases}$$
		Lastly, we will prove \eqref{item:some_properties_of_b}. To show \eqref{eq:uniform_bound_on_bkk} we notice that by definition
		\begin{equation}\nonumber
		\begin{split}
			&b_{k,k}\pa{n_1,\dots,n_{k},t}= e^{-\frac{\lambda m_2\pa{\sum_{r=1}^{k} n_r^2} t}{2}}\Bigg(\wh{f_{k,0}}\pa{n_1,\dots,n_{k}}\\
				+&2\lambda \wh{f_{1,0}}\pa{\sum_{r=1}^{k}n_r}\pa{\sum_{i<j\leq k}	a_{k-1}\pa{n_1,\dots,\wt{n_i},\dots,n_i+n_j,\dots, \wt{n_j},\dots, n_{k}}}t\Bigg),
		\end{split}	
		\end{equation}
		when $~2k\pa{k-1}+m_2\sum_{r=1}^{k}n_r^2 = m_2\pa{\sum_{r=1}^{k}n_r}^2$ and
				\begin{equation}\nonumber
			\begin{split}
				&b_{k,k}\pa{n_1,\dots,n_{k},t}= e^{-\frac{\lambda m_2\pa{\sum_{r=1}^{k} n_r^2} t}{2}}\pa{\wh{f_{k,0}}\pa{n_1,\dots,n_{k}} -a_{k}\pa{n_1,\dots,n_{k}}}
				\end{split}	
			\end{equation}
			when $~2k\pa{k-1}+m_2\sum_{r=1}^{k}n_r^2 \ne m_2\pa{\sum_{r=1}^{k}n_r}^2$.
		Consequently,
		\begin{equation}\nonumber
		\begin{split}
			&\abs{b_{k,k}\pa{n_1,\dots,n_{k},t}} \leq e^{-\frac{\lambda m_2\pa{\sum_{r=1}^{k} n_r^2} t}{2}}\pa{1+
				\lambda k(k-1)\pa{\prod_{j=1}^{k-1} \ell_j}t},
		\end{split}			
		\end{equation}
		when $~2k\pa{k-1}+m_2\sum_{r=1}^{k}n_r^2 = m_2\pa{\sum_{r=1}^{k}n_r}^2$ and
			\begin{equation}\nonumber
			\begin{split}
				&\abs{b_{k,k}\pa{n_1,\dots,n_{k},t}} \leq e^{-\frac{\lambda m_2\pa{\sum_{r=1}^{k} n_r^2} t}{2}}\pa{1+
						\prod_{j=1}^{k}\ell_j},
			\end{split}			
		\end{equation}
		when $~2k\pa{k-1}+m_2\sum_{r=1}^{k}n_r^2 \ne m_2\pa{\sum_{r=1}^{k}n_r}^2$. \eqref{eq:uniform_bound_on_bkk} follows.\\
		We continue to prove \eqref{item:some_properties_of_b} and show \eqref{eq:uniform_bound_on_b} and \eqref{eq:recursive_for_upper_bound_b} by induction. From \eqref{eq:uniform_bound_on_bkk} for $k=2$ we find that  
		 	\begin{equation}\nonumber
		 	\begin{split}
		 		& \abs{b_{2,2}\pa{n_1,n_{2},t}} \leq e^{-\frac{\lambda m_2\pa{n_1^2+n_2^2}t}{2}} \pa{1+\begin{cases}
		 			2\lambda t, & 4+m_2\pa{n_1^2+n_2^2} = m_2\pa{n_1+n_2}^2,\\
		 			\ell_2, & 4+m_2\pa{n_1^2+n_2^2} \ne m_2\pa{n_1+n_2}^2,
		 		\end{cases} }\\
	 		& \leq 1+\begin{cases}
	 			2\lambda e^{-\frac{\lambda m_2 t}{2}}t, & 4+m_2\pa{n_1^2+n_2^2} = m_2\pa{n_1+n_2}^2,\\
	 			\ell_2, & 2 \ne 4+m_2\pa{n_1^2+n_2^2} \ne m_2\pa{n_1+n_2}^2,
	 			\end{cases} \leq 1+\max\pa{\ell_2,\frac{4}{ m_2 e}}=C_2,
		 	\end{split}
		 \end{equation}
		 where we have used the fact that $n_1^2+n_2^2 \geq 1$ when $4+m_2\pa{n_1^2+n_2^2}= m_2\pa{n_1+n_2}^2$ and the fact that 
		 \begin{equation}\label{eq:exponential_bound_power1}
		 	\sup_{x\geq 0} x e^{-\beta x} = \frac{1}{\beta e},
		 \end{equation}
		 when $\beta>0$. We now assume that the claims holds for all $2\leq h \leq k$ for some $k\geq 2$. Using \eqref{eq:uniform_bound_on_bkk} and the same arguments we find that 
		 \begin{equation}\nonumber 
		 	\begin{split}
		 		&\abs{b_{k+1,k+1}\pa{n_1,\dots,n_{k+1},t}} \leq 1+ \begin{cases}
		 			\lambda k(k+1)\pa{\prod_{j=1}^k \ell_j}e^{-\frac{\lambda m_2 t}{2}}t, & 		 2k\pa{k+1}+m_2\sum_{r=1}^{k+1}n_r^2 = m_2\pa{\sum_{r=1}^{k+1}n_r}^2,\\
		 			\prod_{j=1}^{k+1}\ell_j, & 2k\pa{k+1}+m_2\sum_{r=1}^{k+1}n_r^2 \ne m_2\pa{\sum_{r=1}^{k+1}n_r}^2,
		 		\end{cases}\\
		 		& \qquad\qquad \qquad \leq 1+ \pa{\prod_{j=1}^k\ell_j}\max \pa{\ell_{k+1}, \frac{2k(k+1)}{m_2 e}} \leq C_{k+1}.
		 	\end{split}
		 \end{equation}
		 For any $2\leq h \leq k$ we use \eqref{eq:recursive_b} and the induction assumption to find that  
		 \begin{equation}\nonumber
		 	\begin{split}
		 		& \abs{b_{h,k+1}\pa{n_1,\dots,n_{k+1},t}} \leq 2\lambda C_k\sum_{i<j\leq k+1}\int_{0}^t e^{-\frac{\lambda \pa{2\rpa{k\pa{k+1}-h(h-1)}+m_2\pa{\sum_{r=1}^{k+1} n_r^2}}\pa{t-s}}{2}}ds\\
		 		&\qquad\qquad \leq k(k+1)\lambda C_k \int_0^t e^{-2k\lambda s}ds \leq \frac{\pa{k+1}C_k}{2} \leq C_{k+1}.
		 	\end{split}
		 \end{equation}
		 The proof is thus complete.
	\end{proof}
	An immediate corollary to Lemma \ref{lem:explicit_behaviour_for_recursive} is the following:
	\begin{corollary}\label{cor:t_limiting_behaviour_of_sol}
			Let $\br{f_k(t)}_{k\in\N}$ be a family of symmetric probability measures on $\br{\PP\pa{\T^k}}_{k\in\N}$ given by \eqref{eq:limit_of_F_N_K_partial_order}.  Then for any $k\in\N$ there exists a symmetric $f_{k,\infty}\in\PP\pa{\T^k}$ such that 
			\begin{equation}\label{eq:t_limit_weak_convergence}
				f_k\pa{\theta_1,\dots,\theta_k,t} \underset{t\to\infty}{\overset{\text{weak}}{\longrightarrow}}f_{k,\infty}\pa{\theta_1.\dots,\theta_k}.
			\end{equation}
			Moreover, the Fourier coefficients of $f_{k,\infty}$ are given by
			\begin{equation}\label{eq:fourier_of_t_limit}
				\wh{f_{k,\infty}}\pa{n_1,\dots,n_k} = \delta_0 \pa{\sum_{r=1}^k n_r} a_k\pa{n_1,\dots,n_k}.
			\end{equation} 
	\end{corollary}
	\begin{remark}\label{rem:weak_convergece_on_interval}
		It is worth to mention at this point that while we normally consider weak convergence of sequences of probability measures, one can easily extend this definition to a continuous family. This is strengthened further by the fact that weak convergence is metrisable on any Polish space and as such we can say that $\mu(t) \underset{t\to\infty}{\overset{\text{weak}}{\longrightarrow}}\mu_\infty$ if $\mu\pa{t_n} \underset{n\to\infty}{\overset{\text{weak}}{\longrightarrow}}\mu_\infty$ for all sequences $\br{t_n}_{n\in\N}$ that go to infinity. 
	\end{remark}
	\begin{proof}[Proof of Corollary \ref{cor:t_limiting_behaviour_of_sol}]
		Similarly to the general principle described in the proof of Theorem \ref{thm:recursive_formula_for_sol}, we have that $\br{f_k(t)}_{t>0}$ will converge weakly as $t$ goes to infinity if and only if its Fourier coefficients converge. In that case, denoting the weak limit of $\br{f_k(t)}_{t>0}$ by $f_{k,\infty}\in\PP\pa{\T^k}$ we must have that
					\begin{equation}\nonumber
			\wh{f_{k,\infty}}\pa{n_1,\dots,n_k} = \lim_{t\to\infty}\wh{f_k}\pa{n_1,\dots,n_k,t}.
		\end{equation} 
		Furthermore, the above together with Lemma \ref{lem:properties_of_even_and_symmetry_in_fourier} will imply that as $f_k$ are symmetric, so is $f_{k,\infty}$. 
		
		Following on Lemma \ref{lem:explicit_behaviour_for_recursive} -- in particular \eqref{eq:explicit_behaviour_for_f_k} and \eqref{eq:uniform_bound_on_b}, and the fact that 
		$$\lim_{t\to\infty} e^{-\alpha t} = \begin{cases}
			0,& \alpha>0,\\
			1, & \alpha=0,
		\end{cases}$$
		we find that 
		$$\lim_{t\to\infty}\wh{f_k}\pa{n_1,\dots,n_k,t} = a_k\pa{n_1,\dots,n_k}\wh{f_{1,0}}\pa{\sum_{r=1}^k n_r} \delta_0\pa{\frac{\lambda m_2 \pa{\sum_{r=1}^k n_r}^2}{2}}=\delta_0\pa{\sum_{r=1}^k n_r}a_k\pa{n_1,\dots,n_k}$$
		and conclude the proof.
	\end{proof}
	With the long time behaviour at hand, we turn our attention to the measure $f_{k,\infty}$. As $f_{k,\infty}$ is intimately connected to $a_k$, we start by investigating properties of these coefficients.
	
	\begin{lemma}\label{lem:additional_properties_of_ak}
		Let $a_k\pa{n_1,\dots,n_k}$ be defined as in \eqref{eq:a11}--\eqref{eq:recursive_a}. Then
		\begin{enumerate}[(i)]
			\item\label{item:a_is_symmetric} $a_k\pa{n_1,\dots,n_k}$ is symmetric for any $k\in\N$, i.e. for any $\sigma\in S^k$
			$$a_k\pa{n_1,\dots,n_k}=a_k\pa{n_{\sigma(1)},\dots,n_{\sigma(k)}}.$$ 
			\item\label{item:a_is_even} $a_k\pa{n_1,\dots,n_k}$ is even for any $k\in\N$, i.e. for any $\pa{n_1,\dots,n_k}\in \Z^k$
			$$a_k\pa{n_1,\dots,n_k}=a_k\pa{-n_1,\dots,-n_k}.$$
			\item\label{item:ak_in_l1} For any $t>0$
			\begin{equation}\nonumber 
				\sum_{\pa{n_1,\dots,n_k}\in\Z^k} e^{-\frac{\lambda m_2 \pa{\sum_{r=1}^k n_r}^2 t}{2}}\abs{a_k\pa{n_1,\dots,n_k}}<\infty.
			\end{equation}
			In particular
			\begin{equation}\label{eq:a_in_l1}
				\sum_{\pa{n_1,\dots,n_k}\in\Z^k} \delta_0\pa{\sum_{r=1}^k n_r}\abs{a_k\pa{n_1,\dots,n_k}}<\infty.
			\end{equation}
		\end{enumerate}
	\end{lemma}
	
	The proof of \eqref{item:ak_in_l1} requires the following technical lemma whose proof is left to Appendix \ref{app:additional_proofs}:
	
	\begin{lemma}\label{lem:estimating_quadratic_sum}
		For any $K\in\Z$ and $m>0$ we have that
		\begin{equation}\label{eq:estimating_quadratic_sum}
			\begin{split}
				&\sup_{A,B\in \Z} \sum_{\fontsize{5}{4}\selectfont{\begin{matrix}
							n\in\Z \\
							m\pa{n^2+An+B}+K\ne 0	
				\end{matrix}}}\frac{1}{\abs{m\pa{n^2+An+B}+K}} \\
			&\leq \begin{cases}
					\frac{6}{m},& \frac{K}{m}\in\Z,\\
					\frac{6}{\min\pa{K-m\lfloor \frac{K}{m} \rfloor, m\pa{\lfloor \frac{K}{m} \rfloor+1}-K}}, & \frac{K}{m}\not\in\Z
				\end{cases} \;\;+ \frac{16}{m}\sum_{N=1}^\infty \frac{1}{N^{\frac{3}{2}}}.
			\end{split}
		\end{equation} 
	\end{lemma}
	
	\begin{proof}[Proof of Lemma \ref{lem:additional_properties_of_ak}]
		We will prove \eqref{item:a_is_symmetric}--\eqref{item:ak_in_l1} by induction, starting with \eqref{item:a_is_symmetric} and \eqref{item:a_is_even}.\\
		The base case $a_1(n_1)=1$ is both even and symmetric. We continue by assuming that the statements hold for some $k\in\N$.
		Since $\pa{\sum_{r=1}^{k+1}n_r}^2$ and $\sum_{r=1}^{k+1}n_r^2$ are invariant under permutation and under the map $\pa{n_1,\dots,n_k}\to \pa{-n_1,\dots,-n_k}$ we find that if $2k(k+1)+m_2\pa{\sum_{r=1}^{k+1}n_r}^2 = m_2\sum_{r=1}^{k+1}n_r^2$ then
		$$a_k\pa{n_{\sigma(1)},\dots,a_{\sigma(k)}} = 0 =a_k\pa{n_1,\dots,n_k},$$
		for any $\sigma\in S^k$, and
		$$a_k\pa{-n_1,\dots,-n_k} = 0 =a_k\pa{n_1,\dots,n_k},$$
		where we have used \eqref{eq:recursive_a_zero}.\\
		We assume, thus, that $2k(k+1)+m_2\pa{\sum_{r=1}^{k+1}n_r}^2 \ne m_2\sum_{r=1}^{k+1}n_r^2$ (and as such the same holds by replacing $n_r$ with $n_{\sigma(r)}$ or $\pa{n_1,\dots,n_k}$ with $\pa{-n_1,\dots,-n_k}$). Noticing that 
		$$\sum_{i<j\leq k+1}a_k\pa{n_1,\dots,\wt{n_i},\dots,n_i+n_j,\dots, \wt{n_j},\dots, n_{k+1}} = \sum_{j<i\leq k+1}a_k\pa{n_1,\dots,\wt{n_j},\dots,n_i+n_j,\dots, \wt{n_i},\dots, n_{k+1}}$$
		$$=\frac{1}{2}\sum_{i,j=1,\;i\not=j}^{k+1}a_k\pa{n_1,\dots,\wt{n_{\min\pa{i,j}}},\dots,n_i+n_j,\dots, \wt{n_{\max\pa{i,j}}},\dots, n_{k+1}}$$
		and utilising \eqref{eq:recursive_a} we see that for any $\sigma\in S^{k+1}$
		$$a_{k+1}\pa{n_{\sigma(1)},\dots, n_{\sigma(k+1)}} = \frac{2}{2k(k+1)+m_2\pa{\sum_{r=1}^{k+1}n_{\sigma(r)}^2 - \pa{\sum_{r=1}^{k+1}n_{\sigma(r)}}^2}}$$
		$$\sum_{i,j=1,\;\sigma(i)\not=\sigma(j)}^{k+1}a_k\pa{n_{\sigma(1)},\dots,\wt{n_{\min\pa{\sigma(i),\sigma(j)}}},\dots,n_{\sigma(i)}+n_{\sigma(j)},\dots, \wt{n_{\max\pa{\sigma(i),\sigma(j)}}},\dots, n_{\sigma(k+1)}}$$
		$$=\frac{2}{2k(k+1)+m_2\pa{\sum_{r=1}^{k+1}n_r^2 - \pa{\sum_{r=1}^{k+1}n_r}^2}}$$
		$$\sum_{p,q=1,\;p\not=q}^{k+1}a_k\pa{n_{1},\dots,\wt{n_{\min\pa{p,q}}},\dots,n_{p}+n_{q},\dots, \wt{n_{\max\pa{p,q}}},\dots, n_{k+1}} = a_{k+1}\pa{n_1,\dots,n_{k+1}},$$
		where we have used our induction assumption of symmetry.\\
		Similarly, using the induction assumption of evenness we see that
		\allowdisplaybreaks
		\begin{align}
			&a_{k+1}\pa{-n_1,\dots,-n_k}=\frac{4 \sum_{i<j\leq k+1}	a_k\pa{-n_1,\dots,\wt{-n_i},\dots,-n_i-n_j,\dots, \wt{-n_j},\dots, -n_{k+1}}}{2k(k+1)+m_2\pa{\sum_{r=1}^{k+1}\pa{-n_r}^2 - \pa{\sum_{r=1}^{k+1}\pa{-n_r}}^2}}\NNN
			& =\frac{4\sum_{i<j\leq k+1}	a_k\pa{n_1,\dots,\wt{n_i},\dots,n_i+n_j,\dots, \wt{n_j},\dots, n_{k+1}} }{2k(k+1)+m_2\pa{\sum_{r=1}^{k+1}n_r^2 - \pa{\sum_{r=1}^{k+1}n_r}^2}}=a_{k+1}\pa{n_1,\dots,n_k}, \nonumber
		\end{align}
		concluding the proof of \eqref{item:a_is_symmetric} and \eqref{item:a_is_even}.\\
		We now turn our attention to \eqref{item:ak_in_l1}. The base case for our induction is
		$$\sum_{n_1\in\Z} e^{-\frac{\lambda m_2 n_1^2 t}{2}}\abs{a_1\pa{n_1}} = \sum_{n_1\in\Z} e^{-\frac{\lambda m_2 n_1^2 t}{2}}=2\sum_{n_1=0}^{\infty}e^{-\frac{\lambda m_2 n_1^2 t}{2}}-1= 2\sum_{n_1=0}^{\infty}\pa{e^{-\frac{\lambda m_2 t}{2}}}^{n_1^2}-1$$
		$$\leq 2\sum_{n_1=0}^{\infty}\pa{e^{-\frac{\lambda m_2 t}{2}}}^{n_1}-1=\frac{2}{1-e^{-\frac{\lambda m_2 t}{2}}}-1=\frac{1+e^{-\frac{\lambda m_2 t}{2}}}{1-e^{-\frac{\lambda m_2 t}{2}}}<\infty,$$
		for $t>0$. We continue and assume that the claim holds for some $k\in\N$. We have that 
		$$\sum_{\pa{n_1,\dots,n_{k+1}}\in \Z^{k+1}}e^{-\frac{\lambda m_2 \pa{\sum_{r=1}^{k+1}n_r}^2t}{2}}\abs{a_{k+1}\pa{n_1,\dots,n_{k+1}}}$$
		$$\leq  \sum_{\fontsize{5}{4}\selectfont{\begin{matrix}
			\pa{n_1,\dots,n_{k+1}}\in \Z^{k+1}\\
			2k(k+1)+m_2\sum_{r=1}^{k+1}n_r^2 \ne m_2\pa{\sum_{r=1}^{k+1}n_r}^2
			\end{matrix}}
			}e^{-\frac{\lambda m_2 \pa{\sum_{r=1}^{k+1}n_r}^2t}{2}}\frac{4}{\abs{2k(k+1)+m_2\pa{\sum_{r=1}^{k+1}n_r^2 - \pa{\sum_{r=1}^{k+1}n_r}^2}}}$$
			$$\pa{\sum_{i<j\leq k+1}	\abs{a_k\pa{n_1,\dots,\wt{n_i},\dots,n_i+n_j,\dots, \wt{n_j},\dots, n_{k+1}}}}$$
			$$=\sum_{i<j\leq k+1}\sum_{\fontsize{5}{4}\selectfont{\begin{matrix}
						\pa{n_1,\dots,n_{k+1}}\in \Z^{k+1}\\
						2k(k+1)+m_2\sum_{r=1}^{k+1}n_r^2 \ne m_2\pa{\sum_{r=1}^{k+1}n_r}^2
				\end{matrix}}
			}e^{-\frac{\lambda m_2 \pa{\sum_{r=1}^{k+1}n_r}^2t}{2}}\frac{4\abs{a_k\pa{n_i+n_j, n_1,\dots,\wt{n_i},\dots, \wt{n_j},\dots, n_{k+1}}}}{\abs{2k(k+1)+m_2\pa{\sum_{r=1}^{k+1}n_r^2 - \pa{\sum_{r=1}^{k+1}n_r}^2}}}$$
			$$=\frac{k(k+1)}{2}\sum_{\fontsize{5}{4}\selectfont{\begin{matrix}
						\pa{n_1,\dots,n_{k+1}}\in \Z^{k+1}\\
						2k(k+1)+m_2\sum_{r=1}^{k+1}n_r^2 \ne m_2\pa{\sum_{r=1}^{k+1}n_r}^2
				\end{matrix}}
			}e^{-\frac{\lambda m_2 \pa{\sum_{r=1}^{k+1}n_r}^2t}{2}}\frac{4\abs{a_k\pa{n_1+n_2, n_3,\dots,n_{k+1}}}}{\abs{2k(k+1)+m_2\pa{\sum_{r=1}^{k+1}n_r^2 - \pa{\sum_{r=1}^{k+1}n_r}^2}}}$$
		where we have used \eqref{eq:recursive_a} and the symmetry of $a_k$, $\sum_{r=1}^{k+1}n_r^2 $, and $\sum_{r=1}^{k+1}n_r $.
		
		Defining 
		$$p_l = \sum_{r=1}^{l}n_r$$
		which is invertible map from $\Z^{k+1}$ to itself with an inverse
		$$n_r=\begin{cases}
			p_1,& r=1,\\
			p_{r}-p_{r-1}, & 2\leq r\leq k+1,
		\end{cases}$$
		 we see that 
		$$ \sum_{\fontsize{5}{4}\selectfont{\begin{matrix}
			\pa{n_1,\dots,n_{k+1}}\in \Z^{k+1}\\
			2k(k+1)+m_2\sum_{r=1}^{k+1}n_r^2 \ne m_2\pa{\sum_{r=1}^{k+1}n_r}^2
			\end{matrix}}
			}e^{-\frac{\lambda m_2 \pa{\sum_{r=1}^{k+1}n_r}^2t}{2}}\frac{4	\abs{a_k\pa{n_1+n_2,n_3,\dots,n_{k+1}}}}{\abs{2k(k+1)+m_2\pa{\sum_{r=1}^{k+1}n_r^2 - \pa{\sum_{r=1}^{k+1}n_r}^2}}}$$
			$$ =\sum_{\fontsize{5}{4}\selectfont{\begin{matrix}
						\pa{p_1,\dots,p_{k+1}}\in \Z^{k+1}\\
						2k(k+1)+m_2\pa{p_1^2+ \sum_{r=2}^{k+1}\pa{p_r-p_{r-1}}^2}\ne m_2p_{k+1}^2
				\end{matrix}}
			}e^{-\frac{\lambda m_2 p_{k+1}^2t}{2}}\frac{4	\abs{a_k\pa{p_2,p_3-p_2,\dots,p_{k+1}-p_k}}}{\abs{2k(k+1)+m_2\pa{p_1^2+ \sum_{r=2}^{k+1}\pa{p_r-p_{r-1}}^2}-m_2p_{k+1}^2}}$$
			$$ =\sum_{\fontsize{5}{4}\selectfont{\begin{matrix}
						\pa{p_1,\dots,p_{k+1}}\in \Z^{k+1}\\
						2k(k+1)+2m_2\Bigg(p_1^2-p_2p_1\\
							+ \pa{\sum_{r=2}^{k}p_r^2 - \sum_{r=3}^{k+1}p_rp_{r-1}}\Bigg)\ne 0
				\end{matrix}}
			}e^{-\frac{\lambda m_2 p_{k+1}^2t}{2}}\frac{4	\abs{a_k\pa{p_2,p_3-p_2,\dots,p_{k+1}-p_k}}}{\abs{2k(k+1)+2m_2\pa{p_1^2-p_2p_1 +\pa{\sum_{r=2}^{k}p_r^2 - \sum_{r=3}^{k+1}p_rp_{r-1}}}}}$$
			$$\leq \sum_{\pa{p_2,\dots,p_{k+1}}\in \Z^k}4 e^{-\frac{\lambda m_2 p_{k+1}^2t}{2}}\abs{a_k\pa{p_2,p_3-p_2,\dots,p_{k+1}-p_k}}$$
			$$\sum_{\fontsize{5}{4}\selectfont{\begin{matrix}
						p_1\in \Z^{k+1},\\
						2k(k+1)+2m_2\Bigg(p_1^2-p_2p_1
						+ \pa{\sum_{r=2}^{k}p_r^2 - \sum_{r=3}^{k+1}p_rp_{r-1}}\Bigg)\ne 0
				\end{matrix}}}\frac{1}{\abs{2k(k+1)+2m_2\pa{p_1^2-p_2p_1 +\pa{\sum_{r=2}^{k}p_r^2 - \sum_{r=3}^{k+1}p_rp_{r-1}}}}}$$
			$$ \leq 4C_{m_2,k}\sum_{\pa{p_2,\dots,p_{k+1}}\in \Z^k} e^{-\frac{\lambda m_2 p_{k+1}^2t}{2}}\abs{a_k\pa{p_2,p_3-p_2,\dots,p_{k+1}-p_k}}$$
			where we have used Lemma \ref{lem:estimating_quadratic_sum} and denoted by
			$$C_{m_2,k}= \begin{cases}
				\frac{3}{m_2},& \frac{k(k+1)}{m_2}\in\Z,\\
				\frac{3}{\min\pa{k(k+1)-m_2\lfloor \frac{k(k+1)}{m_2} \rfloor, m_2\pa{\lfloor \frac{k(k+1)}{m_2} \rfloor+1}-k(k+1)}}, & \frac{k(k+1)}{m_2}\not\in\Z
			\end{cases} \;\;+ \frac{8}{m_2}\sum_{N=1}^\infty \frac{1}{N^{\frac{3}{2}}}.$$
			Defining $l_r = p_{r+1}-p_{r}$ for $r\geq 2$ and $l_1=p_2 $ we conclude that 
			$$\sum_{\pa{n_1,\dots,n_{k+1}}\in \Z^{k+1}}e^{-\frac{\lambda m_2 \pa{\sum_{r=1}^{k+1}n_r}^2t}{2}}\abs{a_{k+1}\pa{n_1,\dots,n_{k+1}}} $$
			$$\leq 2k(k+1)C_{m_2,k}\sum_{\pa{l_1,\dots,l_{k}}\in \Z^k} e^{-\frac{\lambda m_2 \pa{\sum_{r=1}^k l_r}^2 t}{2}}\abs{a_k\pa{l_1,\dots,l_k}}<\infty,$$
			according to our induction assumption. To complete the proof we notice that 
			$$\sum_{\pa{n_1,\dots,n_k}\in\Z^k} \delta_0\pa{\sum_{r=1}^k n_r}\abs{a_k\pa{n_1,\dots,n_k}} \leq \sum_{\pa{n_1,\dots,n_k}\in\Z^k}e^{-\frac{\lambda m_2 \pa{\sum_{r=1}^{k}n_r}^2t}{2}}\abs{a_k\pa{n_1,\dots,n_k}}<\infty.$$
 	\end{proof}
	We only require one more ingredient to be able to prove Theorem \ref{thm:main_partial_order}.
	\begin{lemma}\label{lem:about_partial_decomp}
		Let $\br{f_{k,\infty}}_{k\in\N}$ be the family of probability measures described in Corollary \ref{cor:t_limiting_behaviour_of_sol}.
		For any $k\in\N$ there exists an even and symmetric $\nu_k\in \PP\pa{\T^k}$ which is absolutely continuous with respect to $\frac{d\theta_1\dots d\theta_k}{\pa{2\pi}^k}$ such that \eqref{eq:formula_for_f_infty} holds for $k\geq 2$, i.e.
		\begin{equation}\nonumber
			\wh{f_{k,\infty}}\pa{n_1,\dots,n_k}= \frac{\delta_0\pa{\sum_{r=1}^{k}n_r}}{k}\sum_{l=1}^{k}\wh{\nu_{k-1}}\pa{n_1,\dots,\wt{n_l},\dots,n_{k}}.
		\end{equation}
		Moreover, the probability density of $\nu_k$, $\upnu_k$, is continuous and \eqref{eq:formula_for_f_infty} can be rewritten by \eqref{eq:formula_for_f_infty_spatial}, i.e.
		\begin{equation}\nonumber
			f_{k,\infty}\pa{\theta_1,\dots,\theta_k} = \pa{\frac{1}{k}\sum_{l=1}^k \upnu_{k-1}\pa{\theta_1-\theta_l,\dots \wt{\theta_l-\theta_l},\dots,\theta_k-\theta_l}}\frac{d\theta_1\dots d\theta_k}{\pa{2\pi}^k}.
		\end{equation}
		 $\nu_k$ is explicitly given by \eqref{eq:explicit_nu_k}, i.e.
				\begin{equation}\nonumber
			\begin{split}
				& \wh{\nu_{k}}\pa{n_1,\dots,n_{k}} =\sum_{i=1}^k\frac{2(k+1)\xi_{k}\pa{n_1,\dots,\wt{n_i}, - \sum_{r\not= i}n_r, \dots, n_{k}}}{2k(k+1)+m_2\sum_{r=1}^{k}n_r^2 + m_2\pa{\sum_{r=1}^{k}n_r}^2 }
			\end{split}
		\end{equation}
	where $\xi_k$ satisfies the recursive relation \eqref{eq:explicit_xi_k}, i.e.
			\begin{equation}\nonumber
			\begin{split}
				&\wh{\xi_1}(n_1)=1,\\
				& \wh{\xi_{k}}\pa{n_1,\dots,n_{k}} =
				\frac{4\sum_{i<j\leq k}\xi_{k-1}\pa{n_1,\dots,\wt{n_i},\dots,n_i+n_j,\dots, \wt{n_j},\dots, n_{k}}}{2k(k-1)+m_2\sum_{r=1}^{k}n_r^2 },\quad k \geq 2.
			\end{split}
		\end{equation}
	\end{lemma}
	\begin{proof}
	We start by noticing that when $k\geq 2$ and $\sum_{r=1}^{k}n_r=0$ we must have that 
	$$2k\pa{k-1}+m_2\pa{\sum_{r=1}^{k}n_r^2} \not= m_2\pa{\sum_{r=1}^{k}n_r}^2$$
	from which, together with  \eqref{eq:recursive_a}, we conclude that
	\allowdisplaybreaks
	\begin{align}
			&	\wh{f_{k,\infty}}\pa{n_1,\dots,n_k}=\delta_0\pa{\sum_{r=1}^{k}n_r}a_{k}\pa{n_1,\dots,n_{k}} \NNN
			&=\frac{4	\delta_0\pa{\sum_{r=1}^{k}n_r}}{2k(k-1)+m_2\sum_{r=1}^{k}n_r^2 }\sum_{i<j\leq k}	a_{k-1}\pa{n_1,\dots,\wt{n_i},\dots,n_i+n_j,\dots, \wt{n_j},\dots, n_{k}}\NNN
				&=\frac{4	\delta_0\pa{\sum_{r=1}^{k}n_r}}{2k(k-1)+m_2\sum_{r=1}^{k}n_r^2 }\sum_{i<j\leq k}	\delta_0\pa{\sum_{r=1}^{k}n_r }a_{k-1}\pa{n_1,\dots,\wt{n_i},\dots,n_i+n_j,\dots, \wt{n_j},\dots, n_{k}}\NNN
			&=\frac{4	\delta_0\pa{\sum_{r=1}^{k}n_r}}{2k(k-1)+m_2\sum_{r=1}^{k}n_r^2 }\sum_{i<j\leq k}	\wh{f_{k-1,\infty}}\pa{n_1,\dots,\wt{n_i},\dots,n_i+n_j,\dots, \wt{n_j},\dots, n_{k}}\NNN
		&\qquad=\frac{2\delta_0\pa{\sum_{r=1}^{k}n_r}}{2k(k-1)+m_2\sum_{r=1}^{k}n_r^2 }\sum_{i\not=j}^{k}	\wh{f_{k-1,\infty}}\pa{n_1,\dots,\wt{n_{\min\pa{i,j}}},\dots,n_i+n_j,\dots, \wt{n_{\max\pa{i,j}}},\dots, n_{k}}\NNN
		&\qquad  = \frac{\delta_0\pa{\sum_{r=1}^{k}n_r}}{k}\sum_{j=1}^{k} \sum_{i=1,\;i\not=j}^k\frac{2k\wh{f_{k-1,\infty}}\pa{n_1,\dots,\wt{n_{\min\pa{i,j}}},\dots,n_i+n_j,\dots, \wt{n_{\max\pa{i,j}}},\dots, n_{k}}}{2k(k-1)+m_2\sum_{r=1}^{k}n_r^2 }\NNN
		&=\qquad \frac{\delta_0\pa{\sum_{r=1}^{k}n_r}}{k}\sum_{j=1}^{k} \sum_{i=1,\;i\not=j}^k\frac{2k\wh{f_{k-1,\infty}}\pa{n_1,\dots,\wt{n_{\min\pa{i,j}}},\dots,-\sum_{l=1,\;l\ne i,j}^{k}n_l,\dots, \wt{n_{\max\pa{i,j}}},\dots, n_{k}}}{2k(k-1)+m_2\sum_{r=1,\;r\ne j}^{k}n_r^2 + m_2\pa{\sum_{r=1,\;r\not=j}^{k}n_r}^2 }\NNN
		&\qquad= \frac{\delta_0\pa{\sum_{r=1}^{k}n_r}}{k}\sum_{j=1}^{k} \wh{\nu_{k-1}}\pa{n_1,\dots,\wt{n_j},\dots,n_{k}}\nonumber
	\end{align}
		where for any $k\in\N$
	\begin{equation}\label{eq:temp_def_for_nu}
		\begin{split}
			&	\wh{\nu_k}\pa{p_1,\dots,p_k}: = \sum_{i=1}^k\frac{2(k+1)\wh{f_{k,\infty}}\pa{p_1,\dots,\wt{p_i}, - \sum_{r=1,\;r\not= i}^k p_r, \dots, p_{k}}}{2k(k+1)+m_2\sum_{r=1}^{k}p_r^2 + m_2\pa{\sum_{r=1}^{k}p_r}^2 },
		\end{split}
	\end{equation}
	where we have used the symmetry of $\wh{f_{k,\infty}}$. Note that when $k=1$
	$$\wh{f_{1,\infty}}\pa{p_1,\dots,\wt{p_1}, - \sum_{r=1,\;r\not= 1}^l p_r, \dots, p_{1}}=\wh{f_{1,\infty}}(0)=1,$$
	so the formulae above still hold and make sense. 
	
		To show that $\wh{\nu_k}$ indeed represents the Fourier coefficients of some probability measure $\nu_k$ we recall that Bochner's theorem for $\T^k$ or $\R^k$ states that $\br{\wh{\mu}\pa{p_1,\dots,p_k}}_{\pa{p_1,\dots,p_k}\in\Z^k}$ or $\br{\wh{\mu}\pa{p_1,\dots,p_k}}_{\pa{p_1,\dots,p_k}\in\R^k}$ are the Fourier coefficients of a probability measure on $\T^k$ or $\R^k$ respectively if and only if 
	$\wh{\mu}\pa{0,\dots,0}=1$ and
	\begin{equation}\label{eq:bochner_crietrion}
		\mathbf{F}(\wh{\mu})_{\bm{p},N} =\pa{\wh{\mu}\pa{p_{1,n}-p_{1,m},\dots,p_{k,n}-p_{k,m}}}_{n,m\in\br{1,\dots,N}}
	\end{equation}
	is a positive semi-definite matrix for any $N\in\N$ and $\br{p_{1,n},\dots,p_{k,n}}_{n\in\br{1,\dots,N}}$ (see, for instance, \cite[Theorem 1.4.3]{Rudin91}).
	
	We notice that if we denote by
	$$P^{(i)}_{r,n} = \begin{cases}
		p_{r,n},& r\not=i,\\
		-\sum_{r=1,\;r\ne i}^{k} p_{r,n},& r=i,
	\end{cases}$$
	for a given $i\in \br{1,\dots, k}$ and $\br{\pa{p_{1,n},\dots,p_{k,n}}}_{n\in\br{1,\dots,N}}\in \Z^{k}$ then
	\begin{equation}\nonumber
		\begin{split}
			&\mathbf{A}^{(i)}_{\bm{p},N} =\pa{\wh{f_{k,\infty}}\pa{p_{1,n}-p_{1,m},\dots,\wt{p_{i,n}-p_{i,m}}, -\sum_{r=1,\;r\ne i}^k\pa{p_{r,n}-p_{r,m}},\dots, p_{k,n}-p_{k,m}}}_{n,m\in\br{1,\dots,N}}\\
			&\qquad\qquad\qquad=\pa{\wh{f_{k,\infty}}\pa{P^{(i)}_{1,n}-P^{(i)}_{1,m},\dots,P^{(i)}_{k,n}-P^{(i)}_{k,m}}}_{n,m\in\br{1,\dots,N}}
		\end{split}
	\end{equation}
	is a positive semi-definite matrix due to the fact that $f_{k,\infty}$ is a probability measure on $\T^k$.
	
	In addition, it is well know that for any $a>0$ the Fundamental solution for the operator $a^2-\Delta$ on $\R^k$ is proportional to
	$$	Y_a\pa{\bm{x}} = \begin{cases}
		\frac{e^{-a\abs{x}}}{a},& k=1,\\
		\frac{a^{\frac{k-2}{2}}K_{\frac{k-2}{2}}\pa{a\abs{\bm{x}}}}{|x|^{\frac{k-2}{2}}}, & k\geq 2,		
	\end{cases}$$
	where $K_\nu$ is the modified Bessel function of second kind of order $\nu$.
	 $Y_a$ is a non-negative function that belongs to $ L^1\pa{\R^k}$ and its Fourier transform is given by
	$$\wh{Y_a}\pa{\bm{\xi}} = \frac{C_k}{a^2+\bm{\xi}^2},$$
	where $C_k$ is a fixed positive constant (see Appendix \ref{app:additional_proofs} for more details). Utilising Bochner's theorem again (with $Y_a$ acting as a density function of a positive and finite Borel measure on $\R^k$) we see that
	\begin{equation}\nonumber
		\begin{split}
			&\mathbf{B}_{\bm{p},N} =\pa{\frac{2(k+1)}{2k(k+1)+m_2\sum_{r=1}^{k}\pa{p_{r,n}-p_{r,m}}^2 + m_2\pa{\sum_{r=1}^{k}p_{r,n}-\sum_{r=1}^k p_{r,m}}^2 }}_{n,m\in\br{1,\dots,N}}\\
			&\qquad\qquad\qquad=\pa{\frac{2(k+1)\wh{Y_{\sqrt{\frac{2k(k+1)}{m_2}}}}\pa{\mathcal{P}_{1,n}-\mathcal{P}_{1,m},\dots,\mathcal{P}_{k+1,n}-\mathcal{P}_{k+1,m}}}{m_2C_k}}_{n,m\in\br{1,\dots,N}},
		\end{split}
	\end{equation}
	where 
	$$\mathcal{P}_{r,n} = \begin{cases}
		p_{r,n},& r\not=k+1,\\
		\sum_{r=1}^{k} p_{r,n},& r=k+1,
	\end{cases}$$
	is a positive semi-definite matrix. 
	
	Consequently, for any $N\in\N$ and $\br{\pa{p_{1,n},\dots,p_{k,n}}}_{n\in\br{1,\dots,N}}\in \Z^{k}$ the matrix $\mathbf{N}$ defined by
	\begin{equation}\nonumber
		\begin{split}
			&\mathbf{N}_{n,m} = \wh{\nu_k}\pa{p_{1,n}-p_{1,m},\dots,p_{k,n}-p_{k,m}}=\\
			&=\frac{\sum_{i=1}^k2(k+1)\wh{f_{k,\infty}}\pa{p_{1,n}-p_{1,m},\dots,\wt{p_{i,n}-p_{i,m}}, -\sum_{r=1,\;r\ne i}^k\pa{p_{r,n}-p_{r,m}},\dots, p_{k,n}-p_{k,m}}}{2k(k+1)+m_2\sum_{r=1}^{k}\pa{p_{r,n}-p_{r,m}}^2 + m_2\pa{\sum_{r=1}^{k}p_{r,n}-\sum_{r=1}^k p_{r,m}}^2 }\\
			&\qquad\qquad\qquad=\sum_{i=1}^k\pa{\mathbf{A}^{(i)}\odot\mathbf{B}} _{n,m},
		\end{split}
	\end{equation}
	where $\odot$ represents the Hadamard product, is positive semi-definite according to Schur's product theorem and the fact that a sum of positive semi-definite matrices is a positive semi-definite matrix. 
	
	Since
	\begin{equation}\label{eq:one_condition_bochner}
		\wh{\nu_k}\pa{0,\dots,0} = \sum_{i=1}^k\frac{2(k+1)}{2k(k+1)} = 1.
	\end{equation}
	we use Bochner's theorem again to conclude that $\wh{\nu_k}$ is indeed the Fourier coefficients of a probability measure which we will denote by $\nu_k$.

	We now turn our attention to showing that $\nu_k$ is symmetric and even, which we will do by using Lemma \ref{lem:properties_of_even_and_symmetry_in_fourier}.
	
	The symmetry of $\wh{f_{k,\infty}}$ and $\sum_{r=1}^kn_r^2$  and $\sum_{r=1}^kn_r$ imply that 
		\begin{equation}\label{eq:symmetry_of_wh_nu}
		\wh{\nu_k}\pa{n_1,\dots,n_k}=\wh{\nu}\pa{n_{\sigma\pa{1}},\dots,n_{\sigma\pa{k}}},
	\end{equation}
	for any $\pa{n_1,\dots,n_k}\in \Z^k$ and $\sigma\in S^k$ which implies the symmetry of $\nu_k$. In addition, as
	$$\wh{f_{k,\infty}}\pa{-n_1,\dots,-n_k}=\delta_0\pa{-\sum_{r=1}^k n_r}a_k\pa{-n_1,\dots,-n_k}$$
	$$=\delta_0\pa{\sum_{r=1}^k n_r}a_k\pa{n_1,\dots,n_k}=\wh{f_{k,\infty}}\pa{n_1,\dots,n_k},$$
	for any $\pa{n_1,\dots,n_k}\in \Z^k$ according to part \eqref{item:a_is_even} of Lemma \ref{lem:additional_properties_of_ak}, we find that 
	\begin{equation}\label{eq:evenness_of_wh_nu}
		\wh{\nu_k}\pa{n_1,\dots,n_k}=\wh{\nu}\pa{-n_1,\dots,-n_k},
	\end{equation}
	which implies the evenness of $\nu_k$.
	
	The next step in our proof will be to show that the measure $\nu_k$ is absolutely continuous with respect to $\frac{d\theta_1\dots d\theta_k}{\pa{2\pi}^k}$ with a continuous probability density function. To achieve this it will be enough to show that $\br{\wh{\nu_k}\pa{n_1,\dots,n_k}}_{\pa{n_1,\dots,n_k}}\in \ell_1\pa{\Z^k}$ as in that case the function
	
	$$\upnu_k\pa{\theta_1,\dots,\theta_k} = \sum_{\pa{n_1,\dots,n_k}\in\Z^k} \wh{\nu}\pa{n_1,\dots,n_k}e^{i\sum_{r=1}^k n_r\theta_r}$$
	is a continuous function on $\T^k$ whose Fourier coefficients are the same as those of $\nu_k$, implying that 
	$$\nu_k\pa{\theta_1,\dots.\theta_k}= \upnu_k\pa{\theta_1,\dots,\theta_k} \frac{d\theta_1\dots d\theta_k}{\pa{2\pi}^k}.$$
	By definition 
	\begin{align}
		&		\sum_{\pa{n_1,\dots,n_k}\in \Z^k}\abs{\wh{\nu}_k\pa{n_1,\dots,n_k}}\NNN
		&\leq \sum_{\pa{n_1,\dots,n_k}\in \Z^k}\frac{2(k+1)}{2k(k+1)+m_2\sum_{r=1}^{k}n_r^2 + m_2\pa{\sum_{r=1}^{k}n_r}^2 } \sum_{i=1}^k \abs{\wh{f_{k,\infty}}\pa{n_1,\dots,\wt{n_i}, - \sum_{r\not= i}n_r, \dots, n_{k}}}\NNN
		&=\sum_{\pa{n_1,\dots,n_k}\in \Z^k}\frac{2k(k+1)}{2k(k+1)+m_2\sum_{r=1}^{k}n_r^2 + m_2\pa{\sum_{r=1}^{k}n_r}^2 } \ \abs{\wh{f_{k,\infty}}\pa{ - \sum_{r\not= 1}n_r,n_2 \dots, n_{k}}}\NNN 
		& \leq \pa{\sum_{n_1\in\Z} \frac{2k(k+1)}{2k(k+1)+m_2n_1^2}} \sum_{\pa{n_2,\dots,n_k}\in \Z^{k-1}}\abs{\wh{f_{k,\infty}}\pa{ - \sum_{r\not= 1}n_r,n_2 \dots, n_{k}}}\NNN
		& \leq  \pa{\sum_{n_1\in\Z} \frac{2k(k+1)}{2k(k+1)+m_2n_1^2}} \sum_{\pa{n_1,n_2,\dots,n_k}\in \Z^{k}}\delta_0\pa{\sum_{r=1}^k n_r}\abs{a_k\pa{ n_1,n_2 \dots, n_{k}}}<\infty,
		\nonumber
	\end{align}
	according to \eqref{item:ak_in_l1} from Lemma \ref{lem:additional_properties_of_ak} and where we have used the symmetry of $\wh{f_{k,\infty}}$. 
	
	With that at hand we are able to show us \eqref{eq:formula_for_f_infty_spatial} when $k\geq 2$. Indeed, the
	calculation done in Lemma \ref{lem:fourier_of_partial_order} with $\eta\pa{\theta}=\frac{d\theta}{2\pi}$ and the uniqueness of the Fourier coefficients imply that 
	$$f_{k,\infty}\pa{\theta_1,\dots,\theta_k} = \frac{1}{k}\sum_{l=1}^k \nu_{k-1}\pa{\theta_1-\theta_l,\dots \wt{\theta_l-\theta_l},\dots,\theta_k-\theta_l}d\theta_l$$
		$$= \frac{1}{k}\sum_{l=1}^k \pa{\upnu_{k-1}\pa{\theta_1-\theta_l,\dots \wt{\theta_l-\theta_l},\dots,\theta_k-\theta_l}d\pa{\theta_1-\theta_l}\dots \wt{d\pa{\theta_l-\theta_l}}\dots d\pa{\theta_k-\theta_l}}d\theta_l$$
	$$=\pa{ \frac{1}{k}\sum_{l=1}^k \upnu_{k-1}\pa{\theta_1-\theta_l,\dots \wt{\theta_l-\theta_l},\dots,\theta_k-\theta_l}}\frac{d\theta_1\dots d\theta_k}{\pa{2\pi}^k}.$$

	Lastly, we will show that $\wh{\nu_k}$ satisfies \eqref{eq:explicit_nu_k} by showing that for any $k\in\N$
	\begin{equation}\label{eq:f_infty=xi}
		\wh{f_{k,\infty}}\pa{n_1,\dots,n_{k}} = \xi_{k}\pa{n_1,\dots,n_{k}},
	\end{equation}
	when $\sum_{r=1}^{k}n_r=0$ and where $\xi_k$ is defined via \eqref{eq:explicit_xi_k}. This will be shown by induction:\\
	The base case $k=1$ follows immediately from the fact that 
	$$\wh{f_{1,\infty}}\pa{0}=\xi_1(0)=1.$$
	We continue by assuming that \eqref{eq:f_infty=xi} holds for some $k\in\N$ and notice that when $\sum_{r=1}^{k+1}n_r=0$  \eqref{eq:recursive_a} implies that
	$$\wh{f_{k+1,\infty}}\pa{n_1,\dots,n_{k+1}} = \delta_0\pa{\sum_{r=1}^{k+1}n_r}a_{k+1}\pa{n_1,\dots,n_{k+1}}$$
	$$=\frac{4\sum_{i<j\leq k+1}	\delta_0\pa{\sum_{r=1}^{k+1}n_r}a_k\pa{n_1,\dots,\wt{n_i},\dots,n_i+n_j,\dots, \wt{n_j},\dots, n_{k+1}}}{2k(k+1)+m_2\sum_{r=1}^{k+1}n_r^2}$$
	$$= \frac{4\sum_{i<j\leq k+1}	\wh{f_{k,\infty}}\pa{n_1,\dots,\wt{n_i},\dots,n_i+n_j,\dots, \wt{n_j},\dots, n_{k+1}}}{2k(k+1)+m_2\sum_{r=1}^{k+1}n_r^2}$$
	$$=\frac{4\sum_{i<j\leq k+1}	\xi_k\pa{n_1,\dots,\wt{n_i},\dots,n_i+n_j,\dots, \wt{n_j},\dots, n_{k+1}}}{2k(k+1)+m_2\sum_{r=1}^{k+1}n_r^2}=\xi_{k+1}\pa{n_1,\dots,n_{k+1}},$$
	where we have used our induction assumption. We conclude that \eqref{eq:f_infty=xi} holds, and as such the proof is complete.
	\end{proof}
	
	\begin{remark}\label{rem:more_information_on_nu_k}
		
		Looking at the recursive definition of $\wh{\xi_k}$, \eqref{eq:explicit_xi_k}, we notice that a simple inductive argument shows that $\wh{\xi_k}\pa{n_1,\dots,n_k} >0$ for any $\pa{n_1,\dots,n_k}\in\Z^k$ and consequently so is $\wh{\nu_k}$. This implies that the probability density function of $\nu_k$, $\upnu_k$, which was found in the proof of Lemma \ref{lem:about_partial_decomp} satisfies
		$$\upnu_k\pa{\theta_1,\dots,\theta_k} = \abs{\upnu_k\pa{\theta_1,\dots,\theta_k} }\leq \sum_{\pa{n_1,\dots,n_k}\in \Z^k}\abs{\wh{\nu_k}\pa{n_1,\dots,n_k}}$$
		$$=\sum_{\pa{n_1,\dots,n_k}\in \Z^k}\wh{\nu_k}\pa{n_1,\dots,n_k}=\upnu_k\pa{0,\dots,0},$$
		validating the intuition that $\nu_k$ is ``concentrated'' at $\pa{0,\dots,0}$ (i.e. that the variables try to align).
	\end{remark}
	
	We have now gathered all the tools we need to show te proof of Theorem \ref{thm:main_partial_order}:
	
	\begin{proof}[Proof of Theorem \ref{thm:main_partial_order}]
		Theorem \ref{thm:recursive_formula_for_sol}, Corollary \ref{cor:t_limiting_behaviour_of_sol}, Lemma \ref{lem:about_partial_decomp}, and the fact that for any $\mu\in \PP\pa{\T}$ we have that $\wh{\mu}(n)=\delta_0(n)$ if and only if $\mu(\theta)=\frac{d\theta}{2\pi}$ prove all the required statements of the theorem besides that $\br{f_{N,\infty}}_{N\in\N}$ is $\br{\pa{\frac{d\theta}{2\pi},\nu_{k-1}}}_{k\in\N}-$partially ordered. To show this claim we notice that since
		$$\Pi_k\pa{F_{N,j}(t)} = F_{N,k}(t)$$
		for any $j\geq k$, and since $F_{N,l}(t)$ converges weakly to $f_l(t)$, we must have that 
		$$\Pi_k\pa{f_j(t)} = f_k(t),$$
		for any $j\geq k$. Taking the weak limit as $t$ goes to infinity implies that 
		$$\Pi_k\pa{f_{j,\infty}} = f_{k,\infty},$$
		for any $j\geq k$ and consequently
		$$\lim_{N\to\infty}\Pi_k\pa{f_{N,\infty}}\pa{\theta_1,\dots,\theta_k} = f_{k,\infty}\pa{\theta_1,\dots,\theta_k}$$
		and as
		$$\wh{f_{k,\infty}}\pa{n_1,\dots,n_k} =\begin{cases}
			\delta_0(n_1),& k=1,\\
			\frac{\delta_0\pa{\sum_{r=1}^{k}n_r}}{k}\sum_{l=1}^{k}\wh{\nu_{k-1}}\pa{n_1,\dots,\wt{n_l},\dots,n_{k}}, & k\geq 2,
		\end{cases} $$
		where $\nu_k$ is an even and symmetric probability measure on $\T^k$ we conclude the desired result from Lemma \ref{lem:partial_order_in_fourier}.
	\end{proof}
	
	Now that we've established the generation of partial order as $t$ goes to infinity we turn our attention to exploring the propagation of partial order.
	
\section{Propagation (or lack of) of partial order}\label{sec:propagation}

This short section is dedicated to the question of the propagation of partial order in the CL model and will consist solely of the proof of Theorem \ref{thm:lack_of_propagation}.

\begin{proof}[Proof of Theorem \ref{thm:lack_of_propagation}]
	We assume that $f_1(t)$ and $f_2(t)$ solve \eqref{eq:limit_of_F_N_K_partial_order} for $k=1,2$ and as such satisfy \eqref{eq:about_f_1} and \eqref{eq:about_f_2}, i.e.
	$$\wh{f_1}(n_1,t) = e^{-\frac{\lambda m_2 n_1^2 t}{2}}\wh{f_{1,0}}(n_1)$$
	and 
		\begin{equation}\nonumber
		\begin{aligned}
			&\wh{f_2}\pa{n_1,n_2,t}= e^{-\frac{\lambda \pa{4+m_2\pa{n_1^2+n_2^2}}t}{2}}\wh{f_{2,0}}\pa{n_1,n_2}\\
			&\qquad\qquad+\wh{f_{1,0}}\pa{n_1+n_2}\;\begin{cases}
				2\lambda te^{-\frac{\lambda \pa{4+m_2 \pa{n_1^2+n_2^2}} t}{2} }, & 4+m_2\pa{n_1^2+n_2^2}= m_2\pa{n_1+n_2}^2,\\
				\frac{4\pa{e^{-\frac{\lambda m_2 \pa{n_1+n_2}^2 t}{2} }-e^{-\frac{\lambda \pa{4+m_2\pa{n_1^2+n_2^2}}t}{2}}}}{4+m_2\pa{n_1^2+n_2^2-\pa{n_1+n_2}^2}}, & 4+m_2\pa{n_1^2+n_2^2}\ne m_2\pa{n_1+n_2}^2.
			\end{cases}
		\end{aligned}
	\end{equation}
	Using Lemma \ref{lem:fourier_of_partial_order} we see that under the assumption that 
	$$f_2\pa{\theta_1,\theta_2,t} = \frac{1}{2}\pa{\mu\pa{\theta_1,t}\nu\pa{\theta_2-\theta_1,t}+\mu\pa{\theta_2,t}\nu\pa{\theta_1-\theta_2,t}}$$
	for some $\mu\in\PP\pa{\T}$ and an even $\nu\in\PP\pa{\T}$ we have that
	$$\wh{f_2}(n_1,n_2,t)=\frac{\wh{\mu}\pa{n_1+n_2,t}}{2}\pa{\wh{\nu}\pa{n_1,t}+\wh{\nu}\pa{n_2,t}}$$
	from which we conclude that 
	\begin{equation}\nonumber
		\begin{aligned}
			&\wh{\nu}(n,t)=\wh{f_2}\pa{n,-n,t}= e^{-\lambda\pa{2+m_2n^2}t}\wh{f_{2,0}}\pa{n,-n}+
				\frac{2\pa{1-e^{-\lambda\pa{2+m_2n^2}t}}}{2+m_2n^2}.
		\end{aligned}
	\end{equation}
	Moreover
	$$\wh{f_1}(n,t) = \lim_{N\to\infty}\wh{F_{N,1}}(n,t) = \lim_{N\to\infty}\wh{F_{N,2}}(n,0,t)=\wh{f_2}\pa{n,0,t} = \frac{\wh{\mu}\pa{n,t}}{2}\pa{\wh{\nu}\pa{n,t}+1} $$
	and consequently
		$$\wh{f_1}(n_1+n_2,t)\pa{\wh{f_2}\pa{n_1,-n_1,t}+\wh{f_2}\pa{n_2,-n_2,t}} =  \frac{\wh{\mu}(n_1+n_2,t)}{2}\pa{\wh{\nu}\pa{n_1+n_2,t}+1}\pa{\wh{\nu}\pa{n_1,t}+\wh{\nu}\pa{n_2,t}}$$
	$$=\wh{f_2}\pa{n_1,n_2,t}\pa{\wh{f_2}\pa{n_1+n_2,-n_1-n_2,t}+1} .$$
	In particular, for any $n\in\Z$
	$$2\wh{f_1}\pa{2n,t}\wh{f_2}\pa{n,-n,t} = \wh{f_2}\pa{n,n,t}\pa{\wh{f_2}\pa{2n,-2n,t}+1}. $$
	Using the expressions for $\wh{f_1}\pa{n_1,t}$ and $\wh{f_2}\pa{n_1,n_2,t}$ we find that when $n^2 \not= \frac{2}{m_2}$, i.e. $4+m_2\pa{n^2+n^2}\ne m_2 \pa{2n}^2$, 
	\begin{equation}\nonumber
		\begin{split}
			&2e^{-2\lambda m_2n^2 t}\pa{e^{-\lambda\pa{2+m_2n^2}t}\wh{f_{2,0}}\pa{n,-n}+\frac{2\pa{1-e^{-\lambda\pa{2+m_2n^2}t}}}{2+m_2n^2}}\wh{f_{1,0}}(2n)\\
			&=\pa{e^{-\lambda\pa{2+m_2n^2}t}\wh{f_{2,0}}\pa{n,n}+\frac{2\pa{e^{-2\lambda m_2n^2t}-e^{-\lambda\pa{2+m_2n^2}t}}}{2-m_2n^2}\wh{f_{1,0}}\pa{2n}}\\
			&\pa{e^{-2\lambda\pa{1+2m_2n^2}t}\wh{f_{2,0}}\pa{2n,-2n}+\frac{1-e^{-2\lambda\pa{1+2m_2n^2}t}}{1+2m_2n^2}+1}.
		\end{split}
	\end{equation}
	Rearranging the above we find that 
	\begin{equation}\label{eq:main_ineq_for_lack_of_propagation}
	\begin{split}
		&2e^{-\lambda\pa{m_2n^2-2}t}\pa{\frac{2}{2+m_2n^2} + e^{-\lambda\pa{2+m_2n^2}t}\pa{\wh{f_{2,0}}\pa{n,-n}-\frac{2}{2+m_2n^2}}}\wh{f_{1,0}}\pa{2n}\\
		&=\pa{\wh{f_{2,0}}\pa{n,n} + \frac{2\wh{f_{1,0}}\pa{2n}}{m_2n^2-2}-\frac{2e^{-\lambda\pa{m_2n^2-2}t}\wh{f_{1,0}}\pa{2n}}{m_2n^2-2}}\\
		&\pa{1+\frac{1}{1+2m_2n^2}+e^{-2\lambda\pa{1+2m_2n^2}t}\pa{\wh{f_{2,0}}\pa{2n,-2n}-\frac{1}{1+2m_2n^2}}}.
	\end{split}	
	\end{equation}
	To show \eqref{item:no_partial_order_for_some_time} we notice that taking $t$ to infinity in the above implies that when $\abs{n}> \sqrt{\frac{2}{m_2}}$
	$$0 = \pa{\wh{f_{2,0}}\pa{n,n} + \frac{2\wh{f_{1,0}}\pa{2n}}{m_2n^2-2}}\pa{1+\frac{1}{1+2m_2n^2}}$$
	which implies that 
	$$0 = \wh{f_{2,0}}\pa{n,n} + \frac{2\wh{f_{1,0}}\pa{2n}}{m_2n^2-2}$$
	for all $\abs{n}> \sqrt{\frac{2}{m_2}}$, giving us the desired contradiction when we assume that \eqref{eq:f_2_for_lack} holds for all $t\in (0,\infty)$.
	
	To show \eqref{item:no_partial_order_for_all_time} we notice that \eqref{eq:main_ineq_for_lack_of_propagation} implies that 
	\begin{equation}\nonumber
	\begin{split}
		&\abs{p(n)\pa{\wh{f_{2,0}}\pa{n,n} + \frac{2\wh{f_{1,0}}\pa{2n}}{m_2n^2-2}}} \\
		&\leq \pa{\frac{\pa{\frac{2}{2+m_2n^2} + \abs{\wh{f_{2,0}\pa{n,-n}}}+\frac{2}{2+m_2n^2}}\abs{\wh{f_{1,0}(2n)}}}{\abs{1+\frac{1}{1+2m_2n^2}+e^{-2\lambda\pa{1+2m_2n^2}t}\pa{\wh{f_{2,0}}\pa{2n,-2n}-\frac{1}{1+2m_2n^2}}}} + \frac{2\abs{\wh{f_{1,0}}\pa{2n}}}{m_2n^2-2} }\abs{p(n)}e^{-\lambda\pa{2+m_2n^2}t}\\
		& \leq \pa{\frac{3\pa{1+2m_2n^2}}{\pa{1+2m_2n^2}\pa{1-e^{-2\lambda\pa{1+2m_2n^2}t}}}+ \frac{2}{m_2n^2-2}}\abs{p(n)}e^{-\lambda\pa{m_2n^2-2}t} 
	\end{split}	
	\end{equation}
	for any function $p(n)$, any $t\in \pa{0,\infty}$, and any $n\in\Z$ such that $\abs{n} > \sqrt {\frac{2}{m_2}}$, where we have used the fact that 
	$$\abs{e^{-2\lambda\pa{1+2m_2n^2}t}\pa{\wh{f_{2,0}}\pa{2n,-2n}-\frac{1}{1+2m_2n^2}}} \leq e^{-2\lambda\pa{1+2m_2n^2}t}\pa{1+\frac{1}{1+2m_2n^2}}<1+\frac{1}{1+2m_2n^2},$$
	for any $t>0$.  We conclude that for if there exists a $t\in\pa{0,\infty}$ such that \eqref{eq:f_2_for_lack} holds, then for any polynomial $p(n)$ we must have that 
	$$\lim_{n\to\pm\infty}\abs{p(n)\pa{\wh{f_{2,0}}\pa{n,n} + \frac{2\wh{f_{1,0}}\pa{2n}}{m_2n^2-2}}}=0,$$
	giving us the desired contradiction. 
	
	Lastly, the fact that $\br{F_{N}(0)}_{N\in\N}$ is $\br{\pa{\mu_0,\updelta_{k-1}}}_{k\in\N}-$partially ordered, where 
	$$\updelta_{k}\pa{\theta_1,\dots,\theta_k} = \prod_{i=1}^k \delta \pa{\theta_i}$$ 
	was shown in \S\ref{sec:preliminaries}. Moreover
	$$f_{1,0}(\theta_1) = \lim_{N\to\infty}\Pi_1\pa{F_{N,1}(0)}\pa{\theta_1} = \mu_0(\theta_1)$$
	and
	$$f_{2,0}(\theta_1,\theta_2) = \lim_{N\to\infty}\Pi_2\pa{F_{N,2}(0)}\pa{\theta_1,\theta_2} = \mu_0(\theta_1)\delta\pa{\theta_2-\theta_1}$$
	from which we find that 
	$$\wh{f_{1,0}}(n_1)=\wh{\mu_0}(n_1),\qquad \wh{f_{2,0}}\pa{n_1,n_2} =\wh{\mu_0}\pa{n_1+n_2}.$$
	For our choice of $\mu_0$ we have that
	$$\wh{f_{1,0}}(n) = \frac{2}{2+n^2},\qquad \wh{f_{2,0}}\pa{n,n} = \frac{1}{1+2n^2}$$
	and as such
	
	$$\lim_{n\to\infty}n^2\abs{\wh{f_{2,0}}\pa{n,n} + \frac{2\wh{f_{1,0}}\pa{2n}}{m_2n^2-2}}=\frac{1}{2}>0,$$
	showing that the conditions of \eqref{item:no_partial_order_for_all_time} hold.
\end{proof}

\section{The convergence to the limiting partially ordered state}\label{sec:quantitative}

In the penultimate section of our work we will show that the marginal solutions to the rescaled CL model, $\br{F_{N,k}(t)}_{N\in\N}$, converge uniformly in $N$ and $t$ to the limiting generated partially ordered state $f_{k,\infty}$. Moreover, we will be able to provide a quantitative estimate for this convergence by considering the Fourier coefficients of the appropriate measures and using  $f_k(t)$ as an intermediate point.

 The main technical steps we will need to achieve this are expressed in Lemma \ref{lem:distance_F_N_k_and_f_k} and Corollary \ref{cor:not_final_distance_written_simpler}. Due to their extremely technical nature, the reader is encouraged to skip their proof at first reading and focus on theorems \ref{thm:distance_F_N_k_and_f_k} and \ref{thm:distance_f_k_and_f_infty} which lead to the proof of Theorem \ref{thm:quantitative_convergence}.

\begin{lemma}\label{lem:distance_F_N_k_and_f_k}
	Let $\br{F_N(t)}_{N\in\N}$ be the family of symmetric solutions to \eqref{eq:master_CL_rescaled} with initial data $\br{F_{N}(0)}_{N\in\N}$. Assume in addition that $N\epsilon_N^2=1$, the interaction generating function $g$ is a strong interaction generating function with an $l-$th moment $m_l$, where $l\in\N$ and $l\geq 3$, and $\br{F_{N,k}\pa{0}}_{N\in\N}$ converges weakly as $N$ goes to infinity to a family $f_{k,0} \in \PP\pa{\mathcal{T}^k}$ for any $k\in\N$. Let  $f_k(t)$ be the weak limit as $N$ goes to infinity of the family of marginals $\br{F_{N,k}(t)}_{N\in\N}$, with $k\in\N$ as described in Theorem \ref{thm:main_partial_order}. 
	
	If $\br{\alpha_N}_{N\in\N}$ is a positive sequence and $N_0\in\N$ is such that
	\begin{equation}\label{eq:condition_on_N_0}
		N_0 \geq \max\pa{\frac{\pa{2m_l}^{\frac{l}{2}}}{\pi^2},\pa{\frac{32m_l }{\pi^l\max\pa{8,m_2}}}^{\frac{2}{l-2}}}.
	\end{equation}
	and for any $N\geq N_0$ 
	\begin{equation}\label{eq:condition_for_alpha_N_II}
			\alpha_N  \leq \begin{cases}
				\min\br{4^{q+1}\norm{g}_{L^p}^q, \frac{m_2}{8m_3},1}, & l=3,\\
				\min\br{4^{q+1}\norm{g}_{L^p}^q, \sqrt{\frac{m_2}{2m_4}},1}, & l\geq 4,
			\end{cases}
	\end{equation}
	with $q$ being the H\"older conjugate of $p$, and 
	\begin{equation}\label{eq:extra_condition_on_alpha_N}
		N^{\frac{l}{2}}\alpha_N^2	\geq \frac{ 4^{2\pa{q+3}}m_l\norm{g}_{L^p}^{2q}\pa{\sqrt[l]{4m_l}+2\pi}^2}{\pi^{l+2}}.
	\end{equation}
	Then:
	\begin{enumerate}[(i)]
		\item\label{item:approximation_for_k=1} 	
		For any $n\in\Z$ such that $\abs{n\epsilon_N} \geq  \alpha_N$ we have that 
		\begin{equation}\label{eq:distance_between_F_N1_and_f_1_large_frequency}
			\begin{split}
				&\abs{\wh{F_{N,1}}(n,t)-\wh{f_1}(n,t)} \leq e^{-\frac{\lambda m_2t}{2}}\abs{\wh{F_{N,1}}(n,0)-\wh{f_{1,0}}(n)}\\
				&+e^{-\frac{\lambda N\alpha_N^2 \pi^2 t}{ 4^{2\pa{q+2}}\norm{g}_{L^p}^{2q}\pa{\sqrt[l]{4m_l}+2\pi}^2}}+	e^{-\frac{\lambda m_2N\alpha_N^2t}{2}},
			\end{split}		
		\end{equation}
		and for any $n\in\Z$ such that $\abs{n\epsilon_N} \leq \alpha_N$ we have that 
		\begin{equation}\label{eq:distance_between_F_N1_and_f_1_small_frequency}
			\begin{split}
				&\abs{\wh{F_{N,1}}(n,t)-\wh{f_1}(n,t)} \leq e^{-\frac{\lambda m_2t}{2}}\abs{\wh{F_{N,1}}(n,0)-\wh{f_{1,0}}(n)}
				+\mathfrak{e}_1\pa{N}e^{-\frac{\lambda m_2 t}{16} },
			\end{split}		
		\end{equation}
		where 
		\begin{equation}\label{eq:epsilon_1_N}
			\mathfrak{e}_1\pa{N} = \begin{cases}
				\frac{4}{3\pi^lm_2e}\pa{12 m_l N^{\frac{2-l}{2}}+\pi^lm_3 \alpha_N}, & l=3,\\
				\frac{1}{3\pi^l m_2 e}\pa{48m_l N^{\frac{2-l}{2}} + \pi^l m_4\alpha_N^2}, & l\geq 4.
			\end{cases}
		\end{equation}
		\item\label{item:approximation_for_k>1}  
		For a given $k\geq 2$ we define
		\begin{equation}\label{eq:high_low_frequency_sets}
			\begin{gathered}
				\A^{(k)} = \br{\pa{n_1,\dots,n_{k}}\in \Z^{k}\;|\;\exists l_0\in \br{1,\dots,{k}}\;\text{such that } \abs{n_{l_0} \epsilon_N} \geq \alpha_N},\\
				\B^{(k)}=\pa{\A^{(k)}}^c = \br{\pa{n_1,\dots,n_{k}}\in \Z^{k}\;|\;\forall l\leq 1,\; \abs{n_l \epsilon_N} \leq \alpha_N},
			\end{gathered}
		\end{equation}
		and $N_1\in\N$ by 
		\begin{equation}\label{eq:condition_on_N_1s}
			N_1 = \max\pa{N_0,2k, \pa{\frac{96m_l k}{\pi^l m_2}}^{\frac{2}{l-2}}}
		\end{equation}
		Then for any $\pa{n_1,\dots,n_{k}}\in \A^{(k)}$ and $N\geq N_1$ we have that 
			\allowdisplaybreaks
		\begin{align}
				\nonumber&\abs{\wh{F_{N,k}}\pa{n_1,\dots,n_{k},t}-\wh{f_{k}}\pa{n_1,\dots,n_{k},t}}\\
				\nonumber& \leq 
				e^{-\frac{\lambda\pa{ 2k\pa{k-1}+m_2}t}{2}}\abs{\wh{F_{N,k}}\pa{n_1,\dots,n_{k},0}-\wh{f_{k,0}}\pa{n_1,\dots,n_{k}}} \\
				\nonumber& +e^{-\lambda k(k-1)t}\pa{e^{-\frac{\lambda N\alpha_N^2 \pi^2 t}{ 2\cdot4^{2\pa{q+2}}\norm{g}_{L^p}^{2q}\pa{\sqrt[l]{4m_l}+2\pi}^2}}+e^{-\frac{\lambda m_2N\alpha_N^2 t }{2}}}\\
				\label{eq:distance_between_F_Nk_and_f_k_on_A_k} 	
				 &+\frac{ \pa{2+k(k-1)}\cdot 4^{2\pa{q+2}+1} \norm{g}_{L^p}^{2q}\pa{\sqrt[l]{4m_l}+2\pi}^2}{N\alpha_N^2\pi^2 + 2\cdot 4^{2\pa{q+2}}k(k-1)\norm{g}_{L^p}^{2q}\pa{\sqrt[l]{4m_l}+2\pi}^2}+\frac{8}{m_2N\alpha_N^2+2k(k-1)}\\
				\nonumber & +\frac{2\lambda N}{N-1} \sum_{i<j\leq k}\int_{0}^t e^{-\frac{\lambda\pa{2k\pa{k-1}+m_2}\pa{t-s}}{2}} \\
				\nonumber &\qquad\qquad \abs{\wh{F_{N,k-1}}\pa{n_1,\dots, n_i+n_j,\dots,n_{k},s}-\wh{f_{k-1}}\pa{n_1,\dots, n_i+n_j,\dots,n_{k},s}}ds\\
				\nonumber&+\frac{2k(k-1)}{(N-1)\pa{m_2+2k(k-1)}},
		\end{align}
		and for any $\pa{n_1,\dots,n_{k}}\in \B^{(k)}$ and $N\geq N_1$ we have that 
			\allowdisplaybreaks
		\begin{align}
				\nonumber&\abs{\wh{F_{N,{k}}}\pa{n_1,\dots,n_{k},t}-\wh{f_{k}}\pa{n_1,\dots,n_{k},t}} \\
				\nonumber&\leq 
				e^{-\frac{\lambda\pa{ 2k\pa{k-1}+m_2}t}{2}}\abs{\wh{F_{N,k}}\pa{n_1,\dots,n_{k},0}-\wh{f_{k,0}}\pa{n_1,\dots,n_{k}}}\\
				\nonumber&+\mathfrak{e}_2(k,N)e^{-\frac{\lambda m_2t}{24}}e^{-\lambda k(k-1)t}+\frac{ e^{-\frac{\lambda m_2 t}{2}}e^{-\frac{\lambda k(k-1)t}{2}}\pa{2+m_2N\alpha_N^2}
				}{(N-1)e}\\
			\nonumber &+ \frac{
					4e^{-\frac{\lambda m_2 t}{2}}e^{-\frac{\lambda k(k-1)t}{4}}\pa{4+m_2^2N^2 \alpha_N^4}}{(N-1)^2e^2}+\mathfrak{e}_3(k,N)\alpha_N^2 +\frac{96\mathfrak{e}_2(k,N)}{m_2+24k(k-1)}\\
					\label{eq:distance_between_F_Nk_and_f_k_on_B_k}  &
				+\frac{8 \pa{2+m_2N\alpha_N^2}
				}{(N-1)e\pa{m_2+k(k-1)}} +\frac{
				64\pa{4+m_2^2N^2 \alpha_N^4}}{(N-1)^2e^2\pa{2m_2+k(k-1)}}\\
				\nonumber& 
				+\frac{2\lambda N}{N-1} \sum_{i<j\leq k}\int_{0}^t e^{-\frac{\lambda\pa{2k\pa{k-1}+m_2}\pa{t-s}}{2}} \\
				\nonumber &\qquad\qquad \abs{\wh{F_{N,k-1}}\pa{n_1,\dots, n_i+n_j,\dots,n_{k},s}-\wh{f_{k-1}}\pa{n_1,\dots, n_i+n_j,\dots,n_{k},s}}ds\\
				\nonumber&+\frac{2k(k-1)}{(N-1)\pa{m_2+2k(k-1)}}.
		\end{align}
		where 
		\begin{equation}\label{eq:epsilon_2_N}
			\mathfrak{e}_2(k,N) = \begin{cases}
				\frac{8\pa{12km_l N^{\frac{2-l}{2}}+\pi^lm_3\alpha_N}}{3\pi^l m_2e}, & l=3,\\
				\frac{2\pa{48km_l N^{\frac{2-l}{2}}+\pi^lm_4\alpha_N^2}}{3\pi^l m_2e}, &l\geq 4,
			\end{cases}
		\end{equation}
			\begin{equation}\label{eq:epsilon_3_N}
			\mathfrak{e}_3(k,N)= \begin{cases}
				\frac{ 24k m_lN^{\frac{2-l}{2}}+2\pi^lm_3\alpha_N+ 3\pi^lm_2 }{6\pi^l}, &l=3,\\
				\frac{96km_lN^{\frac{2-l}{2}} + 2\pi^lm_4\alpha_N^2+ 12\pi^lm_2 }{24 \pi^l}, & l\geq 4.
			\end{cases}	
		\end{equation}
	\end{enumerate}
\end{lemma}

\begin{proof} [Proof of Lemma \ref{lem:distance_F_N_k_and_f_k}]
	We start by noticing that for any $k\in\N$ and $t\in [0,\infty)$
	$$\abs{\wh{F_{N,k}}(0,\dots,0,t)-\wh{f_{k}}(0,\dots,0,t)}=\abs{1-1}=0.$$
	Consequently, we can restrict our attention to the case where $\pa{n_1,\dots,n_k}\ne \pa{0,\dots,0}$.

	Starting with \eqref{item:approximation_for_k=1} we see that as
	$$\abs{\wh{F_{N,1}}(n,t)-\wh{f}_1(n,t)}  = \abs{e^{\lambda N\pa{\wh{g_{\epsilon_N}}(n)-1}t}\wh{F_{N,1}}(n,0)- e^{-\frac{\lambda m_2n^2 t}{2}}\wh{f_{1,0}}(n)}\leq \abs{\pa{e^{\lambda N\pa{\wh{g_{\epsilon_N}}(n)-1}t}-e^{-\frac{\lambda m_2n^2t}{2}}}\wh{F_{N,1}}(n,0)} $$
	$$+ e^{-\frac{\lambda m_2n^2t}{2}}\abs{\wh{F_{N,1}}(n,0)-\wh{f_{1,0}}(n)}\underset{\abs{n}\geq 1}{\leq} \abs{e^{\lambda N\pa{\wh{g_{\epsilon_N}}(n)-1}t}-e^{-\frac{\lambda m_2n^2t}{2}}}+e^{-\frac{\lambda m_2t}{2}}\abs{\wh{F_{N,1}}(n,0)-\wh{f_{1,0}}(n)}$$
	it is enough to show that under assumptions \eqref{eq:condition_on_N_0}, \eqref{eq:condition_for_alpha_N_II} and \eqref{eq:extra_condition_on_alpha_N} we have that for any $N\geq N_0$
	$$ \abs{e^{\lambda N\pa{\wh{g_{\epsilon_N}}(n)-1}t}-e^{-\frac{\lambda m_2n^2t}{2}}}\leq  e^{- \frac{\lambda N\alpha_N^2 \pi^2 t}{ 4^{2\pa{q+2}}\norm{g}_{L^p}^{2q}\pa{\sqrt[l]{4m_l}+2\pi}^2}}+	e^{-\frac{\lambda m_2N\alpha_N^2t}{2}} $$
	when $\abs{n\epsilon_N}\geq \alpha_N$, and 
	$$ \abs{e^{\lambda N\pa{\wh{g_{\epsilon_N}}(n)-1}t}-e^{-\frac{\lambda m_2}{2}n^2t}}\leq 
\mathfrak{e}_1(N) e^{-\frac{\lambda m_2 t}{16}},$$
	when $0<\abs{n\epsilon_N}\leq \alpha_N$ to conclude the result. To achieve that we will use the fact that $g$ is a strong interaction generating function, and as such all parts of Lemma \ref{lem:on_g_low_and_high_frequencies} hold.
	
	We notice that the condition $N \geq \frac{\pa{2m_l}^{\frac{2}{l}}}{\pi^2}$ implies that $\uptau_N$, defined by \eqref{eq:uptau_N}, satisfies
	\begin{equation}\label{eq:bound_on_Ntau_N}
		0<N\tau_N = \frac{2m_lN}{\pi^l N^{\frac{l}{2}}-m_l} \leq \frac{4m_l}{\pi^l N^{\frac{l}{2}-1}}.
	\end{equation}
	The additional condition $N\geq N_0 \geq \pa{\frac{32m_l }{\pi^l\max\pa{8,m_2}}}^{\frac{2}{l-2}}$ together with the above results in the estimate
	\begin{equation}\label{eq:Ntau_N_estimate}
				N\uptau_N \leq \min\pa{1,\frac{m_2}{8}}.
	\end{equation}
	 Furthermore, conditions \eqref{eq:condition_for_alpha_N_II} and \eqref{eq:extra_condition_on_alpha_N} together with \eqref{eq:bound_on_Ntau_N}  imply that
	\begin{equation}\nonumber 
		\begin{split}
			&\frac{N\alpha_N^2 \pi^2}{2\cdot 4^{2\pa{q+1}+1}\norm{g}_{L^p}^{2q}\pa{\sqrt[l]{4m_l}\alpha_N+2\pi}^2} 
			> \frac{N\alpha_N^2 \pi^2}{ 2\cdot 4^{2\pa{q+2}}\norm{g}_{L^p}^{2q}\pa{\sqrt[l]{4m_l}+2\pi}^2}\\
			&=\frac{N^{\frac{l}{2}}\alpha_N^2 \pi^{l+2}}{ 4^{2\pa{q+3}}m_l\norm{g}_{L^p}^{2q}\pa{\sqrt[l]{4m_l}+2\pi}^2} \frac{8m_l}{\pi^l N^{\frac{l}{2}-1}}\geq 2N\tau_N.
		\end{split}	
	\end{equation}
	We conclude from the above that under our conditions
	\begin{equation}\label{eq:term_with_alpha_N_squared_minus_Ntau_N_is_ok}
		\begin{split}
			&\frac{N\alpha_N^2 \pi^2}{2\cdot 4^{2\pa{q+1}+1}\norm{g}_{L^p}^{2q}\pa{\sqrt[l]{4m_l}\alpha_N+2\pi}^2} -N\tau_N
			\geq  \frac{N\alpha_N^2 \pi^2}{4^{2\pa{q+2}}\norm{g}_{L^p}^{2q}\pa{\sqrt[l]{4m_l}\alpha_N+2\pi}^2}\\
			&\qquad\qquad \geq  \frac{N\alpha_N^2 \pi^2}{4^{2\pa{q+2}}\norm{g}_{L^p}^{2q}\pa{\sqrt[l]{4m_l}+2\pi}^2}>0.
		\end{split}	
	\end{equation}

When $\abs{n\epsilon_N} \geq \alpha_N$ we utilise part \eqref{item:high_frequencies} of Lemma \ref{lem:on_g_low_and_high_frequencies} to find that 
	\begin{equation}\nonumber 
		\begin{split}
			& \abs{e^{\lambda N\pa{\wh{g_{\epsilon_N}}(n)-1}t}-e^{-\frac{\lambda m_2n^2t}{2}}}  \leq e^{- \lambda\pa{\frac{N\alpha_N^2 \pi^2}{ 2\cdot 4^{2\pa{q+1}+1}\norm{g}_{L^p}^{2q}\pa{\sqrt[l]{4m_l}\alpha_N+2\pi}^2}-N\uptau_N}t}+	e^{-\frac{\lambda m_2}{2}\pa{\frac{\alpha_N^2}{\epsilon_N^2}}t} \\
			& \leq e^{- \frac{\lambda N\alpha_N^2 \pi^2 t}{ 4^{2\pa{q+2}}\norm{g}_{L^p}^{2q}\pa{\sqrt[k]{4m_k}+2\pi}^2}}+	e^{-\frac{\lambda m_2N\alpha_N^2t}{2}}.
		\end{split}
	\end{equation}
	
	Next we consider the case where $\abs{n\epsilon_N} \leq \alpha_N$ and $n\not=0$. When $l=3$ we utilise part \eqref{item:low_frequencies} of Lemma \ref{lem:on_g_low_and_high_frequencies} to find that
	\begin{equation}\nonumber
		\begin{split}
			&\abs{e^{\lambda N\pa{\wh{g_{\epsilon_N}}(n)-1}t}-e^{-\frac{\lambda m_2n^2t}{2}}} =e^{-\frac{\lambda m_2n^2 t}{2}}\abs{e^{\lambda N\pa{\wh{g_{\epsilon_N}}(n)-1 +\frac{m_2\epsilon_N^2 n^2}{2}}t}-1}\\
			&  \leq e^{-\frac{\lambda m_2n^2 t}{2}}\pa{e^{\lambda \pa{N\uptau_N +\frac{m_3}{3}N\epsilon_N^3\abs{n}^3} t}-1}\leq \lambda \pa{N\uptau_N +\frac{m_3}{3}N\epsilon_N^3\abs{n}^3} t e^{-\frac{\lambda m_2n^2 t}{2}}e^{\lambda \pa{N\uptau_N +\frac{m_3}{3}N\epsilon_N^3\abs{n}^3} t} \\
			&\underset{\abs{n}\geq 1}{\leq} \lambda\pa{N\tau_N + \frac{m_3}{3}\abs{\epsilon_N n}}n^2 te^{-\frac{\lambda m_2n^2 t}{2}}e^{\lambda \pa{N\uptau_N +\frac{m_3}{3}\abs{\epsilon_N n}} n^2 t}\leq \lambda\pa{N\tau_N + \frac{m_3}{3}\alpha_N}n^2 te^{-\frac{\lambda m_2n^2 t}{2}}e^{\lambda \pa{N\uptau_N +\frac{m_3}{3}\alpha_N} n^2 t}\\
			& \leq \frac{4}{m_2 e} \pa{N\tau_N + \frac{m_3}{3}\alpha_N}e^{-\frac{\lambda m_2n^2 t}{4}}e^{\lambda \pa{\frac{m_2}{8} +\frac{m_3}{3}\alpha_N} n^2 t}
			=  \frac{4}{m_2 e}\pa{N\tau_N  + \frac{m_3 \alpha_N}{3}}e^{-\frac{\lambda}{8}\pa{m_2-\frac{8m_3}{3}\alpha_N}n^2t }\\
			&\leq \frac{4}{m_2 e}\pa{N\tau_N  + \frac{m_3\alpha_N}{3}}e^{-\frac{\lambda}{8}\pa{m_2-\frac{8m_3}{3}\alpha_N}t } \leq \frac{4}{3\pi^lm_2e}\pa{12 m_l N^{\frac{2-l}{2}}+\pi^lm_3 \alpha_N}e^{-\frac{\lambda m_2 t}{16} },
		\end{split}
	\end{equation}
	where we have used \eqref{eq:bound_on_Ntau_N}, \eqref{eq:Ntau_N_estimate}, the fact that $m_2-\frac{8m_3}{3}\alpha_N\geq \frac{m_2}{2}>0$ when \eqref{eq:condition_for_alpha_N_II} holds, and the following: 
	\begin{itemize}
		\item for any $M>0$
		$$\max_{x\in [-M,M]}\abs{e^x-1} = \max \pa{e^M-1,1-e^{-M}} =e^{M}-1.$$
		\item for all $x\geq 0$ we have that $e^x-1 \leq xe^x$.\footnote{Indeed, defining $f(t)=e^t-1-te^t$ we find that $f(0)=0$ and $f'(t) = -te^t<0$ for $t>0$. }
		\item For any $\alpha>0$
			$$xe^{-\alpha x} =\pa{xe^{-\frac{\alpha}{2}x}} e^{-\frac{\alpha}{2}x} \leq \frac{2}{\alpha e}e^{-\frac{\alpha}{2}x}.$$
	\end{itemize}
	Similarly, for $l\geq 4$ we have that 
	\begin{equation}\nonumber
		\begin{split}
			&\abs{e^{\lambda N\pa{\wh{g_{\epsilon_N}}(n)-1}t}-e^{-\frac{\lambda m_2n^2t}{2}}} \leq e^{-\frac{\lambda m_2n^2 t}{2}}\pa{e^{\lambda \pa{N\uptau_N +\frac{m_4}{12}N\epsilon_N^4\abs{n}^4} t}-1}\\
			&  \leq \lambda \pa{N\uptau_N +\frac{m_4}{12}\epsilon_N^2n^2} n^2t e^{-\frac{\lambda m_2n^2 t}{2}}e^{\lambda \pa{N\uptau_N +\frac{m_4}{12}\alpha_N^2 n^2} t} \leq \frac{4}{m_2 e} \pa{N\tau_N + \frac{m_4}{12}\alpha_N^2}e^{-\frac{\lambda m_2n^2 t}{4}}e^{\lambda \pa{\frac{m_2}{8} +\frac{m_4}{12}\alpha_N} n^2 t}  \\
			& \leq \frac{1}{3\pi^l m_2 e}\pa{48m_l N^{\frac{2-l}{2}} + \pi^l m_4\alpha_N^2}e^{-\frac{\lambda}{8}\pa{m_2-\frac{2m_4}{3}\alpha_N^2}t }
			\leq \frac{1}{3\pi^l m_2 e}\pa{48m_l N^{\frac{2-l}{2}} + \pi^l m_4\alpha_N^2}e^{-\frac{\lambda m_2 t}{16} },
		\end{split}
	\end{equation}
	where we have used  fact that $m_2-\frac{2m_4}{3}\alpha_N^2\geq \frac{m_2}{2}>0$ when \eqref{eq:condition_for_alpha_N_II} holds.
	This concludes the proof of \eqref{item:approximation_for_k=1}.
	
	We turn our attention to showing \eqref{item:approximation_for_k>1}. Much like the previous part, we can assume that $\pa{n_1,\dots,n_{k}}\ne \pa{0,\dots,0}$. Using \eqref{eq:limit_of_F_N_K_partial_order} and \eqref{eq:recursive} we find that
	
	\allowdisplaybreaks
	\begin{align}
		\nonumber&\abs{\wh{F_{N,k}}\pa{n_1,\dots,n_{k},t}-\wh{f_{k}}\pa{n_1,\dots,n_{k},t}} \\
		\nonumber& \leq e^{-\frac{\lambda\pa{ 2k\pa{k-1}+m_2\sum_{r=1}^{k} n_r^2}t}{2}}\abs{\wh{F_{N,k}}\pa{n_1,\dots,n_{k},0}-\wh{f_{k,0}}\pa{n_1,\dots,n_{k}}}  \\
		\nonumber&+\abs{e^{-\frac{\lambda N}{N-1}\pa{\pa{N-k}\sum_{l=1}^{k}\pa{1-\wh{g_{\epsilon_N}}\pa{n_l}}+k\pa{k-1}}t}-e^{-\frac{\lambda\pa{ 2k\pa{k-1}+m_2\sum_{r=1}^{k} n_r^2}t}{2}}}\abs{\wh{F_{N,k}}\pa{n_1,\dots,n_{k},0}}\\
		\nonumber&+\frac{\lambda N}{N-1}\left |\sum_{i<j\leq k}\pa{\widehat{g_{\epsilon_N}}\pa{n_i}+\widehat{g_{\epsilon_N}}\pa{n_j}-2}\int_{0}^t e^{-\frac{\lambda N}{N-1}\pa{\pa{N-k}\sum_{l=1}^{k}\pa{1-\wh{g_{\epsilon_N}}\pa{n_l}}+k\pa{k-1}}\pa{t-s}} \right.\\
		\label{eq:main_breakdown_of_distance_k}&\qquad\qquad \left.\wh{F_{N,k-1}}\pa{n_1,\dots, n_i+n_j,\dots,n_{k},s}ds \right |\\
		\nonumber&+\frac{2N\lambda}{N-1}\left |\sum_{i<j\leq k}\int_{0}^t \pa{e^{-\frac{\lambda N}{N-1}\pa{\pa{N-k}\sum_{l=1}^{k}\pa{1-\wh{g_{\epsilon_N}}\pa{n_l}}+k\pa{k-1}}\pa{t-s}}-e^{-\frac{\lambda\pa{2k\pa{k-1}+m_2\sum_{r=1}^{k} n_r^2}\pa{t-s}}{2}}} \right.\\
		\nonumber&\left. \wh{F_{N,k-1}}\pa{n_1,\dots, n_i+n_j,\dots,n_{k},s}ds \right | +\frac{2N\lambda}{N-1}\left| \sum_{i<j\leq k}\int_{0}^t e^{-\frac{\lambda\pa{2k\pa{k-1}+m_2\sum_{r=1}^{k} n_r^2}\pa{t-s}}{2}} \right.\\
		\nonumber&\left.\pa{\wh{F_{N,k-1}}\pa{n_1,\dots, n_i+n_j,\dots,n_{k},s}-\wh{f_{k-1}}\pa{n_1,\dots, n_i+n_j,\dots,n_{k},s}}ds\right|\\
		\nonumber&+\frac{2\lambda }{N-1}\abs{\sum_{i<j\leq k}\int_{0}^t e^{-\frac{\lambda\pa{2k\pa{k-1}+m_2\sum_{r=1}^{k} n_r^2}\pa{t-s}}{2}}
			\wh{f_{k-1}}\pa{n_1,\dots, n_i+n_j,\dots,n_{k},s}ds}\\
		\nonumber&=\tm{I}+\tm{II}+\tm{III}+\tm{IV}+\tm{V}+\tm{VI}.
	\end{align}

	 We will bound each of these terms from above to conclude the proof.\\
	  \underline{The term $\tm{I}$:} Since $\sum_{r=1}^{k}n_r^2\geq 1$ we have that 
	  \begin{equation}\label{eq:estimate_I}
	  	\begin{split}
	  		&\tm{I} \leq e^{-\frac{\lambda\pa{ 2k\pa{k-1}+m_2}t}{2}}\abs{\wh{F_{N,k}}\pa{n_1,\dots,n_{k},0}-\wh{f_{k,0}}\pa{n_1,\dots,n_{k}}}. 
	  	\end{split}
	  \end{equation}
	 \underline{The term $\tm{II}$:} We will estimate this term in a similar way to part \eqref{item:approximation_for_k=1}.\\ 
	 For any $\pa{n_1,\dots,n_{k}}\in \A^{(k)}$ we can find $l_0\in\br{1,\dots,k}$ such that $\abs{n_{l_0}\epsilon_N} \geq \alpha_N$. Consequently,
	 	 \begin{equation}\nonumber
	 	\sum_{r=1}^{k}n_r^2 \geq n_{l_0}^2 \geq \frac{\alpha_N^2}{\epsilon_N^2}=N\alpha_N^2,
	 \end{equation}
	 from which we conclude that 
	  \begin{equation}\label{eq:bound_on_frequencies_sum_A_k}
	 	\min_{\pa{n_1,\dots,n_{k}}\in \A^{(k)}}\sum_{r=1}^{k}n_r^2 \geq N\alpha_N^2.
	 \end{equation}
	  In addition, using parts \eqref{item:g_coefficient_real_and_bounded_by_1} and \eqref{item:high_frequencies} from Lemma \ref{lem:on_g_low_and_high_frequencies} we see that on $\A^{(k)}$
	 \begin{equation}\nonumber
	 	\begin{split}
	 		&N\sum_{l=1}^{k}\pa{\wh{g_{\epsilon_N}}\pa{n_l}-1} \leq  N\pa{\wh{g_{\epsilon_N}}\pa{n_{l_0}}-1}\leq  N\uptau_N  - \frac{N\alpha_N^2 \pi^2}{ 2\cdot 4^{2\pa{q+1}+1}\norm{g}_{L^p}^{2q}\pa{\sqrt[l]{4m_l}\alpha_N+2\pi}^2}\\
	 	\end{split},
	 \end{equation}
	 and together with \eqref{eq:term_with_alpha_N_squared_minus_Ntau_N_is_ok} we conclude that 
	  \begin{equation}\label{eq:estimates_on_A_k}
	 	\begin{split}
	 		&\max_{\pa{n_1,\dots,n_{k}}\in \A^{(k)}}N\sum_{l=1}^{k}\pa{\wh{g_{\epsilon_N}}\pa{n_l}-1} \leq - \frac{N\alpha_N^2 \pi^2}{ 4^{2\pa{q+2}}\norm{g}_{L^p}^{2q}\pa{\sqrt[l]{4m_l}+2\pi}^2}
	 	\end{split},
	 \end{equation}
	 Combining the above observations we find that for any $\pa{n_1,\dots,n_k}\in \A^{(k)}$
	 \begin{equation}\nonumber 
	 	\begin{split}
	 		&\abs{e^{-\frac{\lambda N}{N-1}\pa{\pa{N-k}\sum_{l=1}^{k}\pa{1-\wh{g_{\epsilon_N}}\pa{n_l}}+k\pa{k-1}}t}-e^{-\frac{\lambda\pa{ 2k\pa{k-1}+m_2\sum_{r=1}^{k} n_r^2}t}{2}}}\\
	 		&\leq e^{\lambda\pa{1-\frac{k-1}{N-1}}N\sum_{l=1}^{k}\pa{\wh{g_{\epsilon_N}}\pa{n_l}-1}t}e^{-\frac{\lambda N}{N-1}k\pa{k-1}t}+e^{-\frac{\lambda\pa{ 2k\pa{k-1}+m_2\sum_{r=1}^{k} n_r^2}t}{2}}\\
	 		& \leq e^{-\lambda k(k-1)t}\pa{e^{-\lambda\pa{1-\frac{k-1}{N-1}}\frac{N\alpha_N^2 \pi^2 t}{ 4^{2\pa{q+2}}\norm{g}_{L^p}^{2q}\pa{\sqrt[l]{4m_l}+2\pi}^2}}+e^{-\frac{\lambda m_2N\alpha_N^2 t }{2}}}.
	 	\end{split}
	 \end{equation}
	 The fact that $N\geq N_1\geq 2k >2k-1$ implies that
	 \begin{equation}\label{eq:N_k_connection}
	 	1-\frac{k-1}{N-1} \geq \frac{1}{2}
	 \end{equation}
	 and as such we find that on $\A^{(k)}$
	 	 \begin{equation}\label{eq:estimate_II_on_A_k}
	 	\begin{split}
	 		&	\tm{II} \leq \abs{e^{-\frac{\lambda N}{N-1}\pa{\pa{N-k}\sum_{l=1}^{k}\pa{1-\wh{g_{\epsilon_N}}\pa{n_l}}+k\pa{k-1}}t}-e^{-\frac{\lambda\pa{ 2k\pa{k-1}+m_2\sum_{r=1}^{k} n_r^2}t}{2}}}\\
	 		& \leq e^{-\lambda k(k-1)t}\pa{e^{-\frac{\lambda N\alpha_N^2 \pi^2 t}{ 2\cdot4^{2\pa{q+2}}\norm{g}_{L^p}^{2q}\pa{\sqrt[l]{4m_l}+2\pi}^2}}+e^{-\frac{\lambda m_2N\alpha_N^2 t }{2}}}.
	 	\end{split}
	 \end{equation}
	 We turn our attention to the set $\B^{(k)}$. For any $\pa{n_1,\dots,n_{k}}\in \B^{(k)}$ we have that
	 \allowdisplaybreaks
	 \begin{align}
	 		\nonumber&\tm{II} \leq \abs{e^{-\frac{\lambda N}{N-1}\pa{\pa{N-k}\sum_{l=1}^{k}\pa{1-\wh{g_{\epsilon_N}}\pa{n_l}}+k\pa{k-1}}t}-e^{-\frac{\lambda\pa{ 2k\pa{k-1}+m_2\sum_{r=1}^{k} n_r^2}t}{2}}}\\
			\nonumber&=\abs{e^{\lambda\pa{1-\frac{k-1}{N-1}}\sum_{l=1}^{k}N\pa{\wh{g_{\epsilon_N}}\pa{n_l}-1}t}e^{-\lambda \pa{1+\frac{1}{N-1}}k\pa{k-1}t}-e^{-\frac{\lambda\pa{ 2k\pa{k-1}+m_2\sum_{r=1}^{k} n_r^2}t}{2}}}\\
			\nonumber&=e^{-\frac{\lambda\pa{1-\frac{k-1}{N-1}}m_2\pa{\sum_{r=1}^{k} n_r^2}t}{2}}e^{-\lambda k(k-1)t}\\
			\label{eq:starting_to_estimate_II_on_B_k}&\qquad\qquad\abs{e^{\lambda\pa{1-\frac{k-1}{N-1}}\sum_{l=1}^{k}N\pa{\wh{g_{\epsilon_N}}\pa{n_l}-1+\frac{m_2 \epsilon_N^2 n_l^2}{2}}t}e^{-\frac{\lambda k(k-1)t}{N-1}}-e^{-\frac{\lambda (k-1)m_2\pa{\sum_{r=1}^{k}n_r^2}t}{2\pa{N-1}}}}\\	 	
			\nonumber&\leq e^{-\frac{\lambda m_2\pa{\sum_{r=1}^{k} n_r^2}t}{4}}e^{-\lambda k(k-1)t}e^{-\frac{\lambda k(k-1)t}{N-1}}\abs{e^{\lambda\pa{1-\frac{k}{N-1}}\sum_{l=1}^{k}N\pa{\wh{g_{\epsilon_N}}\pa{n_l}-1+\frac{m_2 \epsilon_N^2 n_l^2}{2}}t}-1}\\
			\nonumber& + e^{-\frac{\lambda m_2\pa{\sum_{r=1}^{k} n_r^2}t}{4}}e^{-\lambda k(k-1)t}\abs{e^{-\frac{\lambda k(k-1)t}{N-1}}-e^{-\frac{\lambda (k-1)m_2\pa{\sum_{r=1}^{k}n_r^2}t}{2\pa{N-1}}}},
	 \end{align}
	 where we have used \eqref{eq:N_k_connection}.
	 We will estimate each of the above terms individually. 
	 
	 We notice that 
	 $$e^{-\frac{\lambda m_2\pa{\sum_{r=1}^{k} n_r^2}t}{4}}e^{-\lambda k(k-1)t}e^{-\frac{\lambda k(k-1)t}{N-1}}\abs{e^{\lambda\pa{1-\frac{k-1}{N-1}}\sum_{l=1}^{k}N\pa{\wh{g_{\epsilon_N}}\pa{n_l}-1+\frac{m_2 \epsilon_N^2 n_l^2}{2}}t}-1}$$
	$$\leq e^{-\frac{\lambda m_2\pa{\sum_{r=1}^{k} n_r^2}t}{4}}e^{-\lambda k(k-1)t}\abs{e^{\lambda\pa{1-\frac{k-1}{N-1}}\sum_{l=1}^{k}N\pa{\wh{g_{\epsilon_N}}\pa{n_l}-1+\frac{m_2 \epsilon_N^2 n_l^2}{2}}t}-1}$$
	$$\leq \lambda \pa{1-\frac{k-1}{N-1}}B_{\sup}\pa{\epsilon_N}\pa{\sum_{r=1}^k n_r^2}te^{-\frac{\lambda m_2\pa{\sum_{r=1}^{k} n_r^2}t}{4}}e^{-\lambda k(k-1)t}e^{\lambda \pa{1-\frac{k-1}{N-1}}B_{\sup}\pa{\epsilon_N}\pa{\sum_{r=1}^k n_r^2}t}$$
	$$\leq \lambda B_{\sup}\pa{\epsilon_N}\pa{\sum_{r=1}^k n_r^2}t e^{-\frac{\lambda m_2\pa{\sum_{r=1}^{k} n_r^2}t}{4}}e^{-\lambda k(k-1)t}e^{\lambda B_{\sup}\pa{\epsilon_N}\pa{\sum_{r=1}^k n_r^2}t},$$
	where $B_{\sup}\pa{\epsilon_N}$ is defined by
	\begin{equation}\label{eq:def_of_B_eps_N}
		B_{\sup}\pa{\epsilon_N}=\sup_{\pa{n_1,\dots,n_{k}}\in \B^{(k)}\setminus\br{0,\dots,0}}\frac{N\sum_{l=1}^{k}\abs{\wh{g_{\epsilon_N}}\pa{n_l}-1+\frac{m_2 \epsilon_N^2 n_l^2}{2}}}{\sum_{l=1}^k n_l^2}.
	\end{equation}
	Using part \eqref{item:low_frequencies} from Lemma \ref{lem:on_g_low_and_high_frequencies} we see that when $l=3$ and $\pa{n_1,\dots,n_k}\ne\pa{0,\dots,0}$
	\begin{equation}\nonumber
		\begin{split}
			&N\sum_{l=1}^{k}\abs{\wh{g_{\epsilon_N}}\pa{n_l}-1+\frac{m_2 \epsilon_N^2 n_l^2}{2}} \leq k N\uptau_N+
				\frac{m_3}{3}N\epsilon_N^3 \pa{\sum_{l=1}^{k}\abs{n_l}^3}\\
				&= kN\uptau_N+
				\frac{m_3}{3} \pa{\sum_{l=1}^{k}\abs{\epsilon_N n_l} n_l^2} \leq \pa{ kN\uptau_N+
					\frac{m_3\alpha_N}{3} }\pa{\sum_{l=1}^{k}n_l^2},
		\end{split}
	\end{equation}
	from which we conclude that 
		\begin{equation}\label{eq:estimate_B_eps_N_l=3}
		\begin{split}
			&B_{\sup}\pa{\epsilon_N} \leq kN\uptau_N+
				\frac{m_3\alpha_N}{3}.
		\end{split}
	\end{equation}
	and consequently for  any $\pa{n_1,\dots,n_k}\in \B^{(k)}$
	\begin{equation}\nonumber
		\begin{split}
		&e^{-\frac{\lambda m_2\pa{\sum_{r=1}^{k} n_r^2}t}{4}}e^{-\lambda k(k-1)t}e^{-\frac{\lambda k(k-1)t}{N-1}}\abs{e^{\lambda\pa{1-\frac{k-1}{N-1}}\sum_{l=1}^{k}N\pa{\wh{g_{\epsilon_N}}\pa{n_l}-1+\frac{m_2 \epsilon_N^2 n_l^2}{2}}t}-1}\\
		&\leq  \frac{8B_{\sup}\pa{\epsilon_N}e^{-\frac{\lambda m_2\pa{\sum_{r=1}^{k} n_r^2}t}{8}}e^{\lambda B_{\sup}\pa{\epsilon_N}\pa{\sum_{r=1}^{k} n_r^2} t}}{m_2e}e^{-\lambda k(k-1)t} \\
		&\leq \frac{8\pa{3kN\uptau_N+m_3\alpha_N}}{3m_2e}e^{-\frac{\lambda \pa{m_2-8kN\uptau_N -\frac{8m_3}{3}\alpha_N}\pa{\sum_{r=1}^{k} n_r^2}t}{8}}e^{-\lambda k(k-1)t}.
		\end{split}
	\end{equation}
		At this point we notice that due to \eqref{eq:condition_for_alpha_N_II}, \eqref{eq:bound_on_Ntau_N} and the fact that $N\geq N_1\geq \pa{\frac{96m_lk}{\pi^l m_2}}^{\frac{2}{l-2}}$ we have that 
	$$\frac{8m_3\alpha_N}{3} +8kN\tau_N\leq \frac{m_2}{3}+ \frac{32km_l}{\pi^{l}N^{\frac{l}{2}-1}} \leq \frac{2m_2}{3},$$
	and as such, together with \eqref{eq:bound_on_Ntau_N} again, the above implies that for any $\pa{n_1,\dots,n_{k}}\in \B^{(k)}$
	\begin{equation}\label{eq:first_term_in_II_on_B_k_estimate_l=3}
		\begin{split}
			&e^{-\frac{\lambda m_2\pa{\sum_{r=1}^{k} n_r^2}t}{4}}e^{-\lambda k(k-1)t}e^{-\frac{\lambda k(k-1)t}{N-1}}\abs{e^{\lambda\pa{1-\frac{k-1}{N-1}}\sum_{l=1}^{k}N\pa{\wh{g_{\epsilon_N}}\pa{n_l}-1+\frac{m_2 \epsilon_N^2 n_l^2}{2}}t}-1}\\
			&\leq  \frac{8\pa{12km_l N^{\frac{2-l}{2}}+\pi^lm_3\alpha_N}}{3\pi^l m_2e}e^{-\frac{\lambda m_2t}{24}}e^{-\lambda k(k-1)t}.
		\end{split}
	\end{equation}
	Similarly, when $l\geq 4$ Lemma \ref{lem:on_g_low_and_high_frequencies} shows that on $\B^{(k)}\setminus\br{0,\dots,0}$
		\begin{equation}\nonumber
		\begin{split}
			&N\sum_{l=1}^{k}\abs{\wh{g_{\epsilon_N}}\pa{n_l}-1+\frac{m_2 \epsilon_N^2 n_l^2}{2}} \leq kN\uptau_N+
			\frac{m_4}{12}N\epsilon_N^4 \pa{\sum_{l=1}^{k}n_l^4} \leq \pa{ kN\uptau_N+
				\frac{m_4\alpha_N^2}{12} }\pa{\sum_{l=1}^{k}n_l^2},
		\end{split}
	\end{equation}
	which implies that
	\begin{equation}\label{eq:estimate_B_eps_N_l>3}
		\begin{split}
			&B_{\sup}\pa{\epsilon_N} \leq kN\uptau_N+
				\frac{m_4\alpha_N^2}{12},
		\end{split}
	\end{equation}
	and as such
	\begin{equation}\label{eq:first_term_in_II_on_B_k_estimate_l>3}
		\begin{split}
			&e^{-\frac{\lambda m_2\pa{\sum_{r=1}^{k} n_r^2}t}{4}}e^{-\lambda k(k-1)t}e^{-\frac{\lambda k(k-1)t}{N-1}}\abs{e^{\lambda\pa{1-\frac{k-1}{N-1}}\sum_{l=1}^{k}N\pa{\wh{g_{\epsilon_N}}\pa{n_l}-1+\frac{m_2 \epsilon_N^2 n_l^2}{2}}t}-1}\\
			&\leq  \frac{8B_{\sup}\pa{\epsilon_N}e^{-\frac{\lambda m_2\pa{\sum_{r=1}^{k} n_r^2}t}{8}}e^{\lambda B_{\sup}\pa{\epsilon_N}\pa{\sum_{r=1}^{k} n_r^2} t}}{m_2e}e^{-\lambda k(k-1)t} \\
			&\leq \frac{2\pa{12kN\uptau_N+m_4\alpha_N^2}}{3m_2e}e^{-\frac{\lambda \pa{m_2-8kN\uptau_N -\frac{2m_4}{3}\alpha_N^2}\pa{\sum_{r=1}^{k} n_r^2}t}{8}}e^{-\lambda k(k-1)t}\\
			&\leq  \frac{2\pa{48km_l N^{\frac{2-l}{2}}+\pi^lm_4\alpha_N^2}}{3\pi^l m_2e}e^{-\frac{\lambda m_2t}{24}}e^{-\lambda k(k-1)t},
		\end{split}
	\end{equation}
	where, again, when $N\geq N_1\geq \pa{\frac{96m_l k}{\pi^l m_2}}^{\frac{2}{l-2}}$, \eqref{eq:condition_for_alpha_N_II}, and \eqref{eq:bound_on_Ntau_N} hold we have that 
	$$\frac{2m_4\alpha_N^2}{3} +8kN\tau_N\leq \frac{m_2}{3}+ \frac{32km_l}{\pi^{l}N^{\frac{l}{2}-1}} \leq \frac{2m_2}{3}.$$
	To estimate the second term in \eqref{eq:starting_to_estimate_II_on_B_k} we start by noticing that as for any $x,y>0$\footnote{We have that 
	$$e^{-x} = e^{-y} +\pa{-e^{-y}}\pa{x -y} + \frac{e^{-c}}{2}\pa{x-y}^2,$$
	where $c$ is between $x$ and $y$.}
	$$\abs{e^{-x}-e^{-y}} \leq \abs{y-x}\pa{1+\abs{y-x}},$$
	we have that on $\B^{(k)}$
	\begin{equation}\nonumber
	\begin{split}
		&e^{-\frac{\lambda m_2\pa{\sum_{r=1}^{k} n_r^2}t}{4}}e^{-\lambda k(k-1)t}\abs{e^{-\frac{\lambda k(k-1)t}{N-1}}-e^{-\frac{\lambda (k-1)m_2\pa{\sum_{r=1}^{k}n_r^2}t}{2\pa{N-1}}}}\\
		&\leq  \frac{\lambda (k-1) t e^{-\frac{\lambda m_2 t}{2}}e^{-\lambda k(k-1)t}\abs{2k-m_2\sum_{r=1}^{k}n_r^2}
		}{2(N-1)}\\
		&+ \frac{
			\pa{\lambda (k-1) t}^2 e^{-\frac{\lambda m_2 t}{2}}e^{-\lambda k(k-1)t}\pa{2k-m_2\sum_{r=1}^{k}n_r^2}^2}{4(N-1)^2}\\
		&\leq \frac{ e^{-\frac{\lambda m_2 t}{2}}e^{-\frac{\lambda k(k-1)t}{2}}\pa{2k+m_2\sum_{r=1}^{k}n_r^2}
		}{k(N-1)e}+ \frac{
		4e^{-\frac{\lambda m_2 t}{2}}e^{-\frac{\lambda k(k-1)t}{4}}\pa{4k^2+m_2^2\pa{\sum_{r=1}^{k}n_r^2}^2}}{k^2(N-1)^2e^2}
	\end{split}	
	\end{equation}
	where we have used  the fact that
	$$x^2 e^{-\alpha x} \leq \frac{2}{\alpha e}xe^{-\frac{\alpha x}{2}} \leq \frac{8}{\alpha^2 e^2}e^{-\frac{\alpha x}{4}}.$$ 
	Recall that on $\B^{(k)}$
	\begin{equation}\nonumber
	\sum_{r=1}^{k}n_r^2 = N	\sum_{r=1}^{k}\epsilon_N^2n_r^2 \leq kN\alpha_N^2, 
	\end{equation}
	and as such 
	\begin{equation}\label{eq:second_term_in_II_on_B_k_estimate}
		\begin{split}
			&e^{-\frac{\lambda m_2\pa{\sum_{r=1}^{k} n_r^2}t}{4}}e^{-\lambda k(k-1)t}\abs{e^{-\frac{\lambda k(k-1)t}{N-1}}-e^{-\frac{\lambda (k-1)m_2\pa{\sum_{r=1}^{k}n_r^2}t}{2\pa{N-1}}}}\\
			&\leq \frac{ e^{-\frac{\lambda m_2 t}{2}}e^{-\frac{\lambda k(k-1)t}{2}}\pa{2+m_2N\alpha_N^2}
			}{(N-1)e}+ \frac{
				4e^{-\frac{\lambda m_2 t}{2}}e^{-\frac{\lambda k(k-1)t}{4}}\pa{4+m_2^2N^2 \alpha_N^4}}{(N-1)^2e^2}.
		\end{split}	
	\end{equation}
	Combining \eqref{eq:starting_to_estimate_II_on_B_k} with \eqref{eq:first_term_in_II_on_B_k_estimate_l=3} or \eqref{eq:first_term_in_II_on_B_k_estimate_l>3}, and \eqref{eq:second_term_in_II_on_B_k_estimate} gives us that for any $\pa{n_1,\dots, n_{k}}\in \B^{(k)}$
	\begin{equation}\label{eq:estimate_II_on_B_k}
		\begin{split}
			&\tm{II} \leq \abs{e^{-\frac{\lambda N}{N-1}\pa{\pa{N-k}\sum_{l=1}^{k}\pa{1-\wh{g_{\epsilon_N}}\pa{n_l}}+k\pa{k+1}}t}-e^{-\frac{\lambda\pa{ 2k\pa{k-1}+m_2\sum_{r=1}^{k} n_r^2}t}{2}}} \\
			& \leq \mathfrak{e}_2(k,N)e^{-\frac{\lambda m_2t}{24}}e^{-\lambda k(k-1)t}
			+\frac{ e^{-\frac{\lambda m_2 t}{2}}e^{-\frac{\lambda k(k-1)t}{2}}\pa{2+m_2N\alpha_N^2}
			}{(N-1)e}\\
			&\qquad\qquad + \frac{
				4e^{-\frac{\lambda m_2 t}{2}}e^{-\frac{\lambda k(k-1)t}{4}}\pa{4+m_2^2N^2 \alpha_N^4}}{(N-1)^2e^2}.
		\end{split}
	\end{equation}
	with $\mathfrak{e}_2(k,N)$ defined in \eqref{eq:epsilon_2_N}.\\
	 \underline{The term $\tm{III}$:} Using the facts that 
	 $$\int_{0}^{t}e^{-\alpha\pa{t-s}}ds = \frac{1-e^{-\alpha t}}{\alpha}\leq \frac{1}{\alpha}$$
	 for any $\alpha>0$ and part \eqref{item:g_coefficient_real_and_bounded_by_1} of Lemma \ref{lem:on_g_low_and_high_frequencies} we find that 
	 $$\tm{III} \leq \frac{\lambda N}{N-1}\sum_{i<j\leq k}\pa{2-\widehat{g_{\epsilon_N}}\pa{n_i}-\widehat{g_{\epsilon_N}}\pa{n_j}}\int_{0}^t e^{-\frac{\lambda N}{N-1}\pa{\pa{N-k}\sum_{l=1}^{k}\pa{1-\wh{g_{\epsilon_N}}\pa{n_l}}+k\pa{k-1}}\pa{t-s}}ds$$
	 $$\leq \frac{\sum_{i<j\leq k}\pa{2-\widehat{g_{\epsilon_N}}\pa{n_i}-\widehat{g_{\epsilon_N}}(n_j)}}{\pa{1-\frac{k}{N}}\sum_{l=1}^{k}N\pa{1-\wh{g_{\epsilon_N}}\pa{n_l}}+k\pa{k-1}}.$$
	 Consequently, as $N\geq N_1\geq 2k$ implies that $1-\frac{k}{N }\geq \frac{1}{2}$ and 
	 $$\sum_{i<j\leq k}\pa{2-\widehat{g_{\epsilon_N}}\pa{n_i}-\widehat{g_{\epsilon_N}}\pa{n_j}} \leq \sum_{i<j\leq k} 4 = 2k(k-1),$$
	 we see that when $\pa{n_1,\dots,n_{k}}\in \A^{(k)}$
	 \begin{equation}\label{eq:estimate_III_on_A_k}
	 	\begin{split}
	 		&\tm{III} \leq \frac{4^{2\pa{q+2}+1}k(k-1)\norm{g}_{L^p}^{2q}\pa{\sqrt[l]{4m_l}+2\pi}^2}{N\alpha_N^2 \pi^2 + 2\cdot4^{2\pa{q+2}}k(k-1)\norm{g}_{L^p}^{2q}\pa{\sqrt[l]{4m_l}+2\pi}^2}
	 	\end{split}
	 \end{equation}
	 where we have used \eqref{eq:estimates_on_A_k}.
	
	 On the other hand, for any $\pa{n_1,\dots,n_{k}}\in\B^{(k)}$ we find that
	 $$\tm{III} \leq \frac{\sum_{i<j\leq k}\pa{2-\widehat{g_{\epsilon_N}}\pa{n_i}-\widehat{g_{\epsilon_N}}(n_j)}}{k\pa{k-1}} = \frac{\sum_{i,j=1,\;i\ne j}^{k}\pa{2-\widehat{g_{\epsilon_N}}\pa{n_i}-\widehat{g_{\epsilon_N}}(n_j)}}{2k\pa{k-1}}$$
	 $$\leq \frac{\sum_{i,j=1,\;i\ne j}^{k}\pa{1-\widehat{g_{\epsilon_N}}\pa{n_i}}}{k\pa{k-1}}=\frac{\sum_{i=1}^{k}\pa{1-\widehat{g_{\epsilon_N}}\pa{n_i}}}{k}=\frac{\sum_{i=1}^{k}\pa{1-\widehat{g_{\epsilon_N}}\pa{n_i}-\frac{m_2 \epsilon_N^2n_i^2}{2}}+\frac{m_2 \epsilon_N^2}{2}\sum_{i=1}^{k}n_i^2}{k}$$
	 $$\underset{N\epsilon_N^2=1}{\leq} \frac{\pa{2\epsilon_N^2B\pa{\epsilon_N}+ m_2 \epsilon_N^2}\sum_{i=1}^{k}n_i^2}{2k} = \frac{\pa{2B_{\sup}\pa{\epsilon_N}+ m_2 }}{2k}\pa{\epsilon_N^2\sum_{i=1}^{k}n_i^2}\leq \frac{\pa{2B_{\sup}\pa{\epsilon_N}+ m_2}\alpha_N^2 }{2}.$$
	 We conclude that on $\B^{(k)}$
	 	\begin{equation}\nonumber 
	 	\begin{split}
	 		&	\tm{III} \leq \frac{\pa{ 6kN\uptau_N+2m_3\alpha_N+ 3m_2}\alpha_N^2 }{6} \\
	 		&\qquad\qquad\leq \frac{\pa{ 24km_lN^{\frac{2-l}{2}}+2\pi^lm_3\alpha_N+ 3\pi^lm_2}\alpha_N^2 }{6\pi^l}.
	 	\end{split}
	 \end{equation}
	 when $l=3$ and
	 	\begin{equation}\nonumber 
	 	\begin{split}
	 		&	\tm{III} \leq \frac{\pa{96km_lN^{\frac{2-l}{2}} + 2\pi^lm_4\alpha_N^2+ 12\pi^lm_2}\alpha_N^2 }{24 \pi^l},
	 	\end{split}
	 \end{equation}
	 when $l\geq 4$, where we have used  \eqref{eq:bound_on_Ntau_N}, \eqref{eq:estimate_B_eps_N_l=3}, and \eqref{eq:estimate_B_eps_N_l>3}. In other words 
	 \begin{equation}\label{eq:estimate_III_on_B_k}
	 	\begin{split}
	 		&	\tm{III} \leq \mathfrak{e}_3(k,N)\alpha_N^2,
	 	\end{split}
	 	\end{equation}
	 	with $\mathfrak{e}_3(k,N)$ defined in \eqref{eq:epsilon_3_N}. \\
	\underline{The term $\tm{IV}$:} As
	$$\tm{IV}\leq \frac{2N\lambda}{N-1}\sum_{i<j\leq k}\int_{0}^t \abs{e^{-\frac{\lambda N}{N-1}\pa{\pa{N-k}\sum_{l=1}^{k}\pa{1-\wh{g_{\epsilon_N}}\pa{n_l}}+k\pa{k-1}}\pa{t-s}}-e^{-\frac{\lambda\pa{2k\pa{k-1}+m_2\sum_{r=1}^{k+1} n_r^2}\pa{t-s}}{2}}}ds $$
	we can use \eqref{eq:estimate_II_on_A_k} together a simple change of variable to see that when $\pa{n_1,\dots,n_{k}}\in \A^{(k)}$ we have that
	\begin{equation}\label{eq:estimate_IV_on_A_k}
		\begin{split}
			&\tm{IV}\leq \frac{2N\lambda }{N-1}\int_{0}^t e^{-\lambda k(k-1)s}\pa{e^{-\frac{\lambda N\alpha_N^2 \pi^2 s}{ 2\cdot4^{2\pa{q+2}}\norm{g}_{L^p}^{2q}\pa{\sqrt[l]{4m_l}+2\pi}^2}}+e^{-\frac{\lambda m_2N\alpha_N^2 s }{2}}}ds\\
			&\leq \frac{ 2\cdot 4^{2\pa{q+2}+1} \norm{g}_{L^p}^{2q}\pa{\sqrt[l]{4m_l}+2\pi}^2}{N\alpha_N^2\pi^2 + 2\cdot 4^{2\pa{q+2}}k(k-1)\norm{g}_{L^p}^{2q}\pa{\sqrt[l]{4m_l}+2\pi}^2}+\frac{8}{m_2N\alpha_N^2+2k(k-1)},
		\end{split}
	\end{equation}
	where we have used the fact that $N\leq 2(N-1)$ when $N\geq 2$. 
	
	Turning our attention to $\B^{(k)}$ we utilise \eqref{eq:estimate_II_on_B_k} to find that
		\begin{equation}\nonumber
		\begin{split}
			&\mathrm{IV} \leq \frac{2N\lambda }{N-1}\int_{0}^t \left(\mathfrak{e}_2(k,N)e^{-\frac{\lambda m_2s}{24}}e^{-\lambda k(k-1)s} \right.\\
			&\left.+\frac{ e^{-\frac{\lambda m_2 2}{2}}e^{-\frac{\lambda k(k-1)s}{2}}\pa{2+m_2N\alpha_N^2}
			}{(N-1)e}+ \frac{
				4e^{-\frac{\lambda m_2 s}{2}}e^{-\frac{\lambda k(k-1)s}{4}}\pa{4+m_2^2N^2 \alpha_N^4}}{(N-1)^2e^2}\right) ds
		\end{split}
	\end{equation}
	
	and conclude that 
		\begin{equation}\label{eq:estimate_IV_on_B_k}
		\begin{split}
			&\tm{IV} \leq \frac{96\mathfrak{e}_2(k,N)}{m_2+24k(k-1)}
			+\frac{8 \pa{2+m_2N\alpha_N^2}
			}{(N-1)e\pa{m_2+k(k-1)}} +\frac{
				64\pa{4+m_2^2N^2 \alpha_N^4}}{(N-1)^2e^2\pa{2m_2+k(k-1)}}.
		\end{split}
	\end{equation}
		\underline{The term $\tm{V}$:} Using the definition of $\tm{V}$ and that fact we are assuming that $\sum_{r=1}^{k}n_r^2 \geq 1$ we find that
		
	\begin{equation}\label{eq:estimate_V}
		\begin{split}
			&\tm{V} \leq \frac{2\lambda N}{N-1}\sum_{i<j\leq k}\int_{0}^t e^{-\frac{\lambda\pa{2k\pa{k-1}+m_2}\pa{t-s}}{2}} \\
			&\abs{\wh{F_{N,k-1}}\pa{n_1,\dots, n_i+n_j,\dots,n_{k},s}-\wh{f_{k-1}}\pa{n_1,\dots, n_i+n_j,\dots,n_{k},s}}ds.
		\end{split}
	\end{equation}
	\underline{The term $\tm{VI }$:} Like above, and using estimates that were mentioned previously, we find that 
	\begin{equation}\label{eq:estimate_VI}
		\begin{split}
			&\tm{VI} \leq \frac{\lambda k(k-1) }{N-1}\int_{0}^t e^{-\frac{\lambda\pa{2k\pa{k-1}+m_2}\pa{t-s}}{2}}ds \leq \frac{2k(k-1)}{(N-1)\pa{m_2+2k(k-1)}}.
		\end{split}
	\end{equation}
	Combining \eqref{eq:estimate_I}, \eqref{eq:estimate_II_on_A_k}, \eqref{eq:estimate_III_on_A_k}, \eqref{eq:estimate_IV_on_A_k}, \eqref{eq:estimate_V} and \eqref{eq:estimate_VI} gives us \eqref{eq:distance_between_F_Nk_and_f_k_on_A_k}, while combining \eqref{eq:estimate_I}, \eqref{eq:estimate_II_on_B_k}, \eqref{eq:estimate_III_on_B_k}, \eqref{eq:estimate_IV_on_B_k}, \eqref{eq:estimate_V} and \eqref{eq:estimate_VI} gives us \eqref{eq:distance_between_F_Nk_and_f_k_on_B_k}. The proof is thus complete. 
\end{proof}

An immediate corollary of Lemma \ref{lem:distance_F_N_k_and_f_k} is the following:

\begin{corollary}\label{cor:not_final_distance_written_simpler}
	Under the same assumptions and notations as Lemma \ref{lem:distance_F_N_k_and_f_k}, there exists an explicit $c_{1}>0$ that depend only on the interaction generating function $g$ such that if $N\geq N_0$ then
	\begin{equation}\label{eq:not_final_distance_simplified_1}
		\begin{split}
			&\abs{\wh{F_{N,1}}(n,t)-\wh{f_1}(n,t)} \leq e^{-\frac{\lambda m_2t}{2}}\abs{\wh{F_{N,1}}(n,0)-\wh{f_{1,0}}(n)}\\
			&+e^{-\frac{\lambda N\alpha_N^2 \pi^2 t}{ 4^{2\pa{q+2}}\norm{g}_{L^p}^{2q}\pa{\sqrt[l]{4m_l}+2\pi}^2}}+	e^{-\frac{\lambda m_2N\alpha_N^2t}{2}} +c_{1}\pa{N^{\frac{2-l}{2}}+\alpha_N^{\min\pa{\lfloor \frac{l}{2}\rfloor,2}}}e^{-\frac{\lambda m_2 t}{16}}.
		\end{split}
	\end{equation}
	for any $n\in\Z$.
	
	In addition, for any $k\geq 2$ there exist explicit $c_{2},c_3,c_4>0$ that depend only on the interaction generating function $g$ such that  if $N\geq N_1$ then
			\allowdisplaybreaks
	\begin{align}
		\nonumber&\abs{\wh{F_{N,k}}\pa{n_1,\dots,n_{k},t}-\wh{f_{k}}\pa{n_1,\dots,n_{k},t}}\\
		\nonumber& \leq 
		e^{-\frac{\lambda\pa{ 2k\pa{k-1}+m_2}t}{2}}\abs{\wh{F_{N,k}}\pa{n_1,\dots,n_{k},0}-\wh{f_{k,0}}\pa{n_1,\dots,n_{k}}} \\
		\nonumber& +e^{-\lambda k(k-1)t}\pa{e^{-\frac{\lambda N\alpha_N^2 \pi^2 t}{ 2\cdot4^{2\pa{q+2}}\norm{g}_{L^p}^{2q}\pa{\sqrt[l]{4m_l}+2\pi}^2}}+e^{-\frac{\lambda m_2N\alpha_N^2 t }{2}}}\\
		\label{eq:not_final_distance_simplified_k} 	
		&+\frac{c_{2}}{N\alpha_N^2} + c_{3}e^{-\frac{\lambda m_2 t}{24}}e^{-\frac{\lambda k(k-1)t}{4}}\pa{kN^{\frac{2-l}{2}}+\frac{1}{N}+\alpha_N^{\min\pa{\lfloor \frac{l}{2}\rfloor,2}}}\\
		\nonumber & +c_4\pa{\frac{1}{N} + \pa{kN^{\frac{2-l}{2}}+1}\alpha_N^2+kN^{\frac{2-l}{2}}+\alpha_N^{\min\pa{\lfloor \frac{l}{2}\rfloor,2}}}\\
		\nonumber &+\frac{2\lambda N}{N-1} \sum_{i<j\leq k}\int_{0}^t e^{-\frac{\lambda\pa{2k\pa{k-1}+m_2}\pa{t-s}}{2}} \\
		\nonumber &\qquad\qquad \abs{\wh{F_{N,k-1}}\pa{n_1,\dots, n_i+n_j,\dots,n_{k},s}-\wh{f_{k-1}}\pa{n_1,\dots, n_i+n_j,\dots,n_{k},s}}ds,
	\end{align}
	for any  $\pa{n_1,\dots,n_{k}}\in \Z^{k} $.
\end{corollary}

\begin{proof}
	From \eqref{eq:distance_between_F_N1_and_f_1_large_frequency}, \eqref{eq:distance_between_F_N1_and_f_1_small_frequency}, \eqref{eq:distance_between_F_Nk_and_f_k_on_A_k} and \eqref{eq:distance_between_F_Nk_and_f_k_on_B_k}, and using the fact that under our assumptions $\alpha_N \leq 1$ we see that we can choose
	\begin{equation}\nonumber
		c_{1} = \begin{cases}
		\frac{4\max\pa{12m_l,\pi^lm_3}}{3\pi^lm_2e}	, & l=3,\\
		\frac{\max\pa{48m_l,\pi^l m_4}}{3\pi^l m_l e}	, & l\geq 4.
		\end{cases}
	\end{equation}
	\begin{equation}\nonumber
		c_{2} = \frac{ \pa{2+k(k+1)}\cdot 4^{2\pa{q+2}+1} \norm{g}_{L^p}^{2q}\pa{\sqrt[l]{4m_l}+2\pi}^2}{\pi^2}+\frac{8}{m_2}
	\end{equation}
	\begin{equation}\nonumber
		c_{3}=\max\pa{2c_1, \frac{32}{e^2},\frac{2m_2}{e},\frac{16m_2^2}{e^2} },
	\end{equation}
	and 
	\begin{equation}\nonumber
	c_4 = \max\pa{\zeta,\frac{192c_1}{m_2},\frac{16}{e},\frac{256}{m_2e^2}, \frac{128 m_2}{e^2},2},
	\end{equation}
		with 
	\begin{equation}\nonumber
		\zeta=\begin{cases}
			\frac{ \max\pa{24 m_l,2\pi^lm_3,3\pi^lm_2 }}{6\pi^l}, &l=3,\\
			\frac{\max\pa{96m_l, 2\pi^lm_4,12\pi^lm_2} }{24 \pi^l}, & l\geq 4,
		\end{cases}	
	\end{equation}
	where we have used the facts that
	$$\frac{\mathfrak{e}_2\pa{k,N}}{kN^{\frac{2-l}{2}}+\alpha_N^{\min\pa{\lfloor \frac{l}{2} \rfloor, 2}}} \leq \begin{cases}
		\frac{8\max\pa{12m_l,\pi^lm_3}}{3\pi^lm_2e},& l=3,\\
			\frac{2\max\pa{48m_l,\pi^lm_4}}{3\pi^lm_2 e}, & l=4,
	\end{cases} =2c_1,$$
	$$\frac{1}{\pa{N-1}^2}\leq \frac{1}{N-1} \leq \frac{2}{N},\qquad \forall N\geq 2,$$
	and
			\begin{equation}\nonumber
		\frac{\mathfrak{e}_3(k,N)}{kN^{\frac{2-l}{l}}+\alpha_N^{\min\pa{\lfloor \frac{l}{2} \rfloor, 2}}+1}= \begin{cases}
			\frac{ \max\pa{24 m_l,2\pi^lm_3,3\pi^lm_2 }}{6\pi^l}, &l=3,\\
			\frac{\max\pa{96m_l, 2\pi^lm_4,12\pi^lm_2} }{24 \pi^l}, & l\geq 4,
		\end{cases}	=\zeta.
	\end{equation}
\end{proof}

With this at hand, we can build on Lemma \ref{lem:distance_F_N_k_and_f_k} and give an explicit estimate to the distance between $F_{N,k}(t)$ and $f_k(t)$. 

\begin{theorem}\label{thm:distance_F_N_k_and_f_k}
	Under the same assumptions and notations as in Lemma \ref{lem:distance_F_N_k_and_f_k} we have that if $N\geq N_1$\footnote{Recall that by definition $N_1\geq N_0$.} then
	\allowdisplaybreaks
	\begin{align}
		\nonumber&\abs{\wh{F_{N,k}}\pa{n_1,\dots,n_{k},t}-\wh{f_{k}}\pa{n_1,\dots,n_{k},t}}\\
		\nonumber& \leq 
		e^{-\frac{\lambda\pa{ 2k\pa{k-1}+m_2}t}{2}}\abs{\wh{F_{N,k}}\pa{n_1,\dots,n_{k},0}-\wh{f_{k,0}}\pa{n_1,\dots,n_{k}}} \\
			\label{eq:final_distance_simplified_k} 	& +\pa{\frac{N}{N-1}}^{k-1}\pa{k-1}\max_{\fontsize{5}{4}\selectfont{\begin{matrix}
											r\in \br{1,\dots,k-1},\\
											\bm{p}\in \mathfrak{L}_{k}\pa{r},\;\;\sigma\in S^k
									\end{matrix}}, }\abs{\wh{F_{N,r}}\pa{s_{\bm{p}}^{\sigma}\pa{n_1,\dots,n_k},0}-\wh{f_{r,0}}\pa{s_{\bm{p}}^{\sigma}\pa{n_1,\dots,n_k}}}\\
		\nonumber& +e^{-\lambda k(k-1)t}\pa{e^{-\frac{\lambda N\alpha_N^2 \pi^2 t}{ 2\cdot4^{2\pa{q+2}}\norm{g}_{L^p}^{2q}\pa{\sqrt[l]{4m_l}+2\pi}^2}}+e^{-\frac{\lambda m_2N\alpha_N^2 t }{2}}}\\
	\nonumber
		 &+\mathfrak{c}_k\pa{\frac{1}{N} + \pa{kN^{\frac{2-l}{2}}+1}\alpha_N^2+kN^{\frac{2-l}{2}}+\alpha_N^{\min\pa{\lfloor \frac{l}{2}\rfloor,2}}+\frac{k(k-1)}{N\alpha_N^2}},
	\end{align}
	for any  $\pa{n_1,\dots,n_{k}}\in \Z^{k} $ where $\mathfrak{L}_{k}(r)$ and $s_{\bm{p}}^\sigma$ are as in Definition \ref{def:k_level_set_of_Z_r}, and $\mathfrak{c}_k>0$ is given by the recursive formula
	\begin{equation}\label{eq:mathfrak_c_k}
		\mathfrak{c}_1=c_1,\qquad \mathfrak{c}_{k+1} =c_2+c_3+c_4+ 2\mathfrak{c}_k+\frac{4^{2\pa{q+2}+1}\norm{g}_{L^p}^{2q}\pa{\sqrt[l]{4m_l}+2\pi}^2}{\pi^2}+\frac{4}{m_2} 
	\end{equation}
	with $c_1,c_2,c_3,c_4$ as in Corollary \ref{cor:not_final_distance_written_simpler}.
\end{theorem}

\begin{proof}
	We prove the claim by induction. The base case $k=1$ follows immediately from \eqref{eq:not_final_distance_simplified_1}. We assume that the claim holds for some $k\in\N$ and show its validity for $k+1$. Using \eqref{eq:not_final_distance_simplified_k} we see that 
	\allowdisplaybreaks
	\begin{align}
		\nonumber&\abs{\wh{F_{N,k+1}}\pa{n_1,\dots,n_{k+1},t}-\wh{f_{k+1}}\pa{n_1,\dots,n_{k+1},t}}\\
		\nonumber& \leq 
		e^{-\frac{\lambda\pa{ 2k\pa{k+1}+m_2}t}{2}}\abs{\wh{F_{N,k+1}}\pa{n_1,\dots,n_{k+1},0}-\wh{f_{k+1,0}}\pa{n_1,\dots,n_{k+1}}} \\
		\nonumber& +e^{-\lambda k(k+1)t}\pa{e^{-\frac{\lambda N\alpha_N^2 \pi^2 t}{ 2\cdot4^{2\pa{q+2}}\norm{g}_{L^p}^{2q}\pa{\sqrt[l]{4m_l}+2\pi}^2}}+e^{-\frac{\lambda m_2N\alpha_N^2 t }{2}}}\\
		\label{eq:starting_induction} & +\pa{c_2+c_3+c_4}\pa{\frac{1}{N} + \pa{\pa{k+1}N^{\frac{2-l}{2}}+1}\alpha_N^2+\pa{k+1}N^{\frac{2-l}{2}}+\alpha_N^{\min\pa{\lfloor \frac{l}{2}\rfloor,2}}+\frac{1}{N\alpha_N^2}}\\
		\nonumber &+\frac{2\lambda N}{N-1} \sum_{i<j\leq k+1}\int_{0}^t e^{-\frac{\lambda\pa{2k\pa{k+1}+m_2}\pa{t-s}}{2}} \\
		\nonumber &\qquad\qquad \abs{\wh{F_{N,k}}\pa{n_1,\dots, n_i+n_j,\dots,n_{k+1},s}-\wh{f_{k}}\pa{n_1,\dots, n_i+n_j,\dots,n_{k+1},s}}ds.
	\end{align}
Since 
$$ \frac{2\lambda N}{N-1}\sum_{i<j\leq k+1}\int_{0}^t e^{-\frac{\lambda\pa{2k\pa{k+1}+m_2}\pa{t-s}}{2}}ds \leq \frac{4N k \pa{k+1}}{2\pa{N-1}\pa{m_2+2k(k+1)}}\leq \frac{N}{N-1}$$
we see that according to the induction assumption we have that
\begin{equation}\nonumber
\begin{split}
	&\frac{2\lambda N}{N-1}\sum_{i<j\leq k+1}\int_{0}^t e^{-\frac{\lambda\pa{2k\pa{k+1}+m_2}\pa{t-s}}{2}} \abs{\wh{F_{N,k}}\pa{n_1,\dots, n_i+n_j,\dots,n_{k+1},s}-\wh{f_{k}}\pa{n_1,\dots, n_i+n_j,\dots,n_{k+1},s}}ds\\
	& \leq \frac{2\lambda N}{N-1} \sum_{i<j\leq k+1}\int_{0}^t e^{-\frac{\lambda\pa{2k\pa{k+1}+m_2}\pa{t-s}}{2}}  e^{-\frac{\lambda\pa{ 2k\pa{k-1}+m_2}s}{2}}\\
	&\qquad\qquad \abs{\wh{F_{N,k}}\pa{n_1,\dots, n_i+n_j,\dots,n_{k+1},0}-\wh{f_{k,0}}\pa{n_1,\dots, n_i+n_j,\dots,n_{k+1}}}ds \\
	&\pa{\frac{N}{N-1}}^{k} \pa{k-1}\max_{i<j\leq k+1}\max_{\fontsize{5}{4}\selectfont{\begin{matrix}
				r\in \br{1,\dots,k-1},\\
				\bm{p}\in \mathfrak{L}_{k}\pa{r},\;\;\sigma\in S^k
		\end{matrix}} }\abs{\wh{F_{N,r}}\pa{s_{\bm{p}}^{\sigma}\pa{n_1,\dots,n_i+n_j,\dots, n_k},0}-\wh{f_{r,0}}\pa{s_{\bm{p}}^{\sigma}\pa{n_1,\dots, n_i+n_j,\dots,n_{k+1}}}}\\
&+\frac{2\lambda N}{N-1} \sum_{i<j\leq k+1}\int_{0}^t e^{-\frac{\lambda\pa{2k\pa{k+1}+m_2}\pa{t-s}}{2}} e^{-\lambda k(k-1)s}\pa{e^{-\frac{\lambda N\alpha_N^2 \pi^2 s}{ 2\cdot4^{2\pa{q+2}}\norm{g}_{L^p}^{2q}\pa{\sqrt[l]{4m_l}+2\pi}^2}}+e^{-\frac{\lambda m_2N\alpha_N^2 s }{2}}}ds\\
&+\frac{N\mathfrak{c}_{k}}{N-1}\pa{\frac{1}{N} + \pa{kN^{\frac{2-l}{2}}+1}\alpha_N^2+kN^{\frac{2-l}{2}}+\alpha_N^{\min\pa{\lfloor \frac{l}{2}\rfloor,2}}+\frac{k(k-1)}{N\alpha_N^2}}.
\end{split}	
\end{equation}
Furthermore
\begin{equation}\nonumber
	\begin{split}
		&\frac{2\lambda N}{N-1} \sum_{i<j\leq k+1}\int_{0}^t e^{-\frac{\lambda\pa{2k\pa{k+1}+m_2}\pa{t-s}}{2}}  e^{-\frac{\lambda\pa{ 2k\pa{k-1}+m_2}s}{2}}\\
		&\qquad \abs{\wh{F_{N,k}}\pa{n_1,\dots, n_i+n_j,\dots,n_{k+1}0}-\wh{f_{k,0}}\pa{n_1,\dots, n_i+n_j,\dots,n_{k+1}}}ds\\
		&\leq \frac{N}{N-1} \max_{i<j\leq k+1}\abs{\wh{F_{N,k}}\pa{n_1,\dots, n_i+n_j,\dots,n_{k+1}0}-\wh{f_{k,0}}\pa{n_1,\dots, n_i+n_j,\dots,n_{k+1}}},\\
		&\leq \pa{\frac{N}{N-1}}^k \max_{i<j\leq k+1}\abs{\wh{F_{N,k}}\pa{n_1,\dots, n_i+n_j,\dots,n_{k+1}0}-\wh{f_{k,0}}\pa{n_1,\dots, n_i+n_j,\dots,n_{k+1}}},
	\end{split}	
\end{equation}
and
$$\frac{2\lambda N}{N-1} \sum_{i<j\leq k+1}\int_{0}^t e^{-\frac{\lambda\pa{2k\pa{k+1}+m_2}\pa{t-s}}{2}} e^{-\lambda k(k-1)s}\pa{e^{-\frac{\lambda N\alpha_N^2 \pi^2 s}{ 2\cdot4^{2\pa{q+2}}\norm{g}_{L^p}^{2q}\pa{\sqrt[l]{4m_l}+2\pi}^2}}+e^{-\frac{\lambda m_2N\alpha_N^2 s }{2}}}ds$$
$$\leq \frac{Nk(k+1)}{N-1}\pa{\frac{2\cdot4^{2\pa{q+2}}\norm{g}_{L^p}^{2q}\pa{\sqrt[l]{4m_l}+2\pi}^2}{\pi^2}+\frac{2}{m_2}}\frac{1}{N\alpha_N^2 }.$$
Plugging the above into \eqref{eq:starting_induction} we find that 
	\allowdisplaybreaks
\begin{align}
	\nonumber&\abs{\wh{F_{N,k+1}}\pa{n_1,\dots,n_{k+1},t}-\wh{f_{k+1}}\pa{n_1,\dots,n_{k+1},t}}\\
	\nonumber& \leq 
	e^{-\frac{\lambda\pa{ 2k\pa{k+1}+m_2}t}{2}}\abs{\wh{F_{N,k+1}}\pa{n_1,\dots,n_{k+1},0}-\wh{f_{k+1,0}}\pa{n_1,\dots,n_{k+1}}} \\
	\nonumber& +\pa{\frac{N}{N-1}}^{k}\max_{i<j\leq k+1}\abs{\wh{F_{N,k}}\pa{n_1,\dots, n_i+n_j,\dots,n_{k+1}0}-\wh{f_{k,0}}\pa{n_1,\dots, n_i+n_j,\dots,n_{k+1}}}\\
	\nonumber& +\pa{\frac{N}{N-1}}^{k} \pa{k-1}\max_{i<j\leq k+1}\max_{\fontsize{5}{4}\selectfont{\begin{matrix}
				r\in \br{1,\dots,k-1},\\
				\bm{p}\in \mathfrak{L}_{k}\pa{r},\;\;\sigma\in S^k
	\end{matrix}} }\abs{\wh{F_{N,r}}\pa{s_{\bm{p}}^{\sigma}\pa{n_1,\dots,n_i+n_j,\dots, n_{k+1}},0}-\wh{f_{r,0}}\pa{s_{\bm{p}}^{\sigma}\pa{n_1,\dots, n_i+n_j,\dots,n_{k+1}}}}\\
	\nonumber& +e^{-\lambda k(k+1)t}\pa{e^{-\frac{\lambda N\alpha_N^2 \pi^2 t}{ 2\cdot4^{2\pa{q+2}}\norm{g}_{L^p}^{2q}\pa{\sqrt[l]{4m_l}+2\pi}^2}}+e^{-\frac{\lambda m_2N\alpha_N^2 t }{2}}}\\
	\nonumber& +\pa{c_2+c_3+c_4}\pa{\frac{1}{N} + \pa{\pa{k+1}N^{\frac{2-l}{2}}+1}\alpha_N^2+\pa{k+1}N^{\frac{2-l}{2}}+\alpha_N^{\min\pa{\lfloor \frac{l}{2}\rfloor,2}}+\frac{1}{N\alpha_N^2}}\\
	\nonumber &+2k(k+1)\pa{\frac{2\cdot4^{2\pa{q+2}}\norm{g}_{L^p}^{2q}\pa{\sqrt[l]{4m_l}+2\pi}^2}{\pi^2}+\frac{2}{m_2}}\frac{1}{N\alpha_N^2 }\\
	\nonumber&+2\mathfrak{c}_{k}\pa{\frac{1}{N} + \pa{kN^{\frac{2-l}{2}}+1}\alpha_N^2+kN^{\frac{2-l}{2}}+\alpha_N^{\min\pa{\lfloor \frac{l}{2}\rfloor,2}}+\frac{k(k-1)}{N\alpha_N^2}},
\end{align}
from which the desired result will follow once we'll prove that
\begin{align}
\nonumber	&\max_{i<j\leq k+1}\abs{\wh{F_{N,k}}\pa{n_1,\dots, n_i+n_j,\dots,n_{k+1},0}-\wh{f_{k,0}}\pa{n_1,\dots, n_i+n_j,\dots,n_{k+1}}}\\
		\nonumber& +\pa{k-1}\max_{i<j\leq k+1}\max_{\fontsize{5}{4}\selectfont{\begin{matrix}
					r\in \br{1,\dots,k-1},\\
					\bm{p}\in \mathfrak{L}_{k}\pa{r},\;\;\sigma\in S^k
		\end{matrix}} }\abs{\wh{F_{N,r}}\pa{s_{\bm{p}}^{\sigma}\pa{n_1,\dots,n_i+n_j,\dots, n_{k+1}},0}-\wh{f_{r,0}}\pa{s_{\bm{p}}^{\sigma}\pa{n_1,\dots, n_i+n_j,\dots,n_{k+1}}}}\\
	\nonumber	& \leq k\max_{\fontsize{5}{4}\selectfont{\begin{matrix}
				r\in \br{1,\dots,k},\\
				\bm{p}\in \mathfrak{L}_{k+1}\pa{r},\;\;\sigma\in S^{k+1}
	\end{matrix}} }\abs{\wh{F_{N,r}}\pa{s_{\bm{p}}^{\sigma}\pa{n_1,\dots, n_{k+1}},0}-\wh{f_{r,0}}\pa{s_{\bm{p}}^{\sigma}\pa{n_1,\dots,n_{k+1}}}}.
\end{align}
Recall that due to the symmetry of all involved probability measures, the position of $n_i+n_j$ in the recursive formula for $F_{N,k}$ and $f_k$ is arbitrary. In the proof that follows we assume we've chosen it to be in the $i-$th position (instead of $n_i$).

Given $r\in \br{1,\dots,k}$, $\bm{p}\in \mathfrak{L}_{k}\pa{r}$, $\sigma\in S^k$, $i<j\leq k+1$ and $\pa{n_1,\dots,n_{k+1}}\in \Z^{k+1}$ we define
$$\sigma_1:\br{1,\dots,j,\dots,k}\to \br{1,\dots,\wt{j},\dots,k+1}$$
by
$$\sigma_1(m) = \begin{cases}
	\sigma(m), & \sigma(m) <j,\\
	\sigma(m)+1, & \sigma(m)\geq j.
\end{cases}$$
Next we find $r_0\in \br{1,\dots,r}$ such that 
$$\sum_{l=1}^{r_0-1} p_{l} +1 \leq  \sigma_1^{-1}(i) \leq\sum_{l=1}^{r_0} p_{l}$$
and notice that by definition\footnote{this holds since $\pa{n_1,\dots,\wt{n_i},n_i+n_j,\dots \wt{n_j},n_{j+1},\dots,n_{k+1}}$
has $n_l$ as its $l-$th element when $l<j$ and $l\ne i$, has  $n_i+n_j$ as its $i-$th element, and has $n_{l+1}$ as its $l-$th element when $l\geq j$.}
\begin{equation}\nonumber
	\begin{split}
		&s_{\bm{p}}^{\sigma}\pa{n_1,\dots,n_i+n_j,\dots, n_{k+1}} 
		=\pa{\sum_{m=1}^{p_1}n_{\sigma_{1}(m)},\dots, \sum_{\fontsize{5}{4}\selectfont{\begin{matrix}
						m=\sum_{l=1}^{r_0-1}p_l+1
			\end{matrix}} }^{\sum_{l=1}^{r_0 }p_l}n_{\sigma_{1}(m)} + n_j,\dots, \sum_{m=\sum_{l=1}^{r-1}p_l+1}^{\sum_{l=1}^{r}p_l}n_{\sigma_{1}(m)}}
	\end{split}	
\end{equation}
Defining $\bm{p^{(i,j)}}\in \mathfrak{L}_{k+1}\pa{r}$ by
$$\pa{\bm{p^{(i,j)}}}_l = \begin{cases}
	p_l,& l\ne r_0,\\
	p_{r_0}+1, & l=r_0.\\
\end{cases}$$
and $\sigma_{i,j}\in S^{k+1}$ by
 $$\sigma_{i,j}(m) = \begin{cases}
 	\sigma_1(m), &  m \leq \sum_{l=1}^{r_0}p_l\\
 	j, & m=\sum_{l=1}^{r_0}p_l+1,\\
 	\sigma_1(m-1), & m>\sum_{l=1}^{r_0}p_l+1,\\ 
 \end{cases}$$
 we see that 
 \begin{equation}\nonumber
 	\begin{split}
 		&s_{\bm{p}}^{\sigma}\pa{n_1,\dots,n_i+n_j,\dots, n_{k+1}}  = s_{\bm{p^{(i,j)}}}^{\sigma_{i,j}}\pa{n_1,\dots,n_{k+1}}.
  	\end{split}	
 \end{equation}
 Consequently for $r\in \br{1,\dots,k-1}$ we find that
 \begin{equation}
 	\begin{split}
 		\nonumber&\max_{i<j\leq k+1}\max_{\fontsize{5}{4}\selectfont{\begin{matrix}
 					r\in \br{1,\dots,k-1},\\
 					\bm{p}\in \mathfrak{L}_{k}\pa{r},\;\;\sigma\in S^k
 		\end{matrix}} }\abs{\wh{F_{N,r}}\pa{s_{\bm{p}}^{\sigma}\pa{n_1,\dots,n_i+n_j,\dots, n_{k+1}},0}-\wh{f_{r,0}}\pa{s_{\bm{p}}^{\sigma}\pa{n_1,\dots, n_i+n_j,\dots,n_{k+1}}}}\\
 	\nonumber& \leq \max_{\fontsize{5}{4}\selectfont{\begin{matrix}
 					r\in \br{1,\dots,k-1},\\
 					\bm{p}\in \mathfrak{L}_{k+1}\pa{r},\;\;\sigma\in S^{k+1}
 		\end{matrix}} }\abs{\wh{F_{N,r}}\pa{s_{\bm{p}}^{\sigma}\pa{n_1,\dots, n_{k+1}},0}-\wh{f_{r,0}}\pa{s_{\bm{p}}^{\sigma}\pa{n_1,\dots,n_{k+1}}}}.
 	\end{split}
 \end{equation}
 By considering $r=k$,  $\bm{p}'=\pa{1,\dots,1}\in\mathfrak{L}_k(k)$, and $\sigma'=\mathrm{id}\in S^k$ we find that
 \begin{equation}
 	\begin{split}
 		\nonumber&\max_{i<j\leq k+1}\abs{\wh{F_{N,k}}\pa{n_1,\dots, n_i+n_j,\dots,n_{k+1},0}-\wh{f_{k,0}}\pa{n_1,\dots, n_i+n_j,\dots,n_{k+1}}}\\
 		\nonumber&=\max_{i<j\leq k+1}\abs{\wh{F_{N,k}}\pa{s_{\bm{p}'}^{\sigma'}\pa{n_1,\dots, n_i+n_j,\dots,n_{k+1},0}}-\wh{f_{k,0}}\pa{s_{\bm{p}'}^{\sigma'}\pa{n_1,\dots, n_i+n_j,\dots,n_{k+1}}}}\\
 		\nonumber& \leq \max_{\fontsize{5}{4}\selectfont{\begin{matrix}
 					\bm{p}\in \mathfrak{L}_{k+1}\pa{k},\;\;\sigma\in S^{k+1}
 		\end{matrix}} }\abs{\wh{F_{N,k}}\pa{s_{\bm{p}}^{\sigma}\pa{n_1,\dots, n_{k+1}},0}-\wh{f_{k,0}}\pa{s_{\bm{p}}^{\sigma}\pa{n_1,\dots,n_{k+1}}}}.
 	\end{split}
 \end{equation}
 We conclude that
 	\begin{align}
 		\nonumber	&\max_{i<j\leq k+1}\abs{\wh{F_{N,k}}\pa{n_1,\dots, n_i+n_j,\dots,n_{k+1},0}-\wh{f_{k,0}}\pa{n_1,\dots, n_i+n_j,\dots,n_{k+1}}}\\
 		\nonumber& +\pa{k-1}\max_{i<j\leq k+1}\max_{\fontsize{5}{4}\selectfont{\begin{matrix}
 					r\in \br{1,\dots,k-1},\\
 					\bm{p}\in \mathfrak{L}_{k}\pa{r},\;\;\sigma\in S^k
 		\end{matrix}} }\abs{\wh{F_{N,r}}\pa{s_{\bm{p}}^{\sigma}\pa{n_1,\dots,n_i+n_j,\dots, n_{k+1}},0}-\wh{f_{r,0}}\pa{s_{\bm{p}}^{\sigma}\pa{n_1,\dots, n_i+n_j,\dots,n_{k+1}}}}\\
 	 \nonumber&\leq \max_{\fontsize{5}{4}\selectfont{\begin{matrix}
 	 			\bm{p}\in \mathfrak{L}_{k+1}\pa{k},\;\;\sigma\in S^{k+1}
 	 \end{matrix}} }\abs{\wh{F_{N,k}}\pa{s_{\bm{p}}^{\sigma}\pa{n_1,\dots, n_{k+1}},0}-\wh{f_{k,0}}\pa{s_{\bm{p}}^{\sigma}\pa{n_1,\dots,n_{k+1}}}}\\
  \nonumber& + \pa{k-1}\max_{\fontsize{5}{4}\selectfont{\begin{matrix}
  			r\in \br{1,\dots,k-1},\\
  			\bm{p}\in \mathfrak{L}_{k+1}\pa{r},\;\;\sigma\in S^{k+1}
  \end{matrix}} }\abs{\wh{F_{N,r}}\pa{s_{\bm{p}}^{\sigma}\pa{n_1,\dots, n_{k+1}},0}-\wh{f_{r,0}}\pa{s_{\bm{p}}^{\sigma}\pa{n_1,\dots,n_{k+1}}}}\\
 		\nonumber	& \leq  k\max_{\fontsize{5}{4}\selectfont{\begin{matrix}
 					r\in \br{1,\dots,k},\\
 					\bm{p}\in \mathfrak{L}_{k+1}\pa{r},\;\;\sigma\in S^{k+1}
 		\end{matrix}} }\abs{\wh{F_{N,r}}\pa{s_{\bm{p}}^{\sigma}\pa{n_1,\dots, n_{k+1}},0}-\wh{f_{r,0}}\pa{s_{\bm{p}}^{\sigma}\pa{n_1,\dots,n_{k+1}}}}.
 	\end{align}
 
 The proof is thus complete.
\end{proof}

With the convergence of $F_{N,k}(t)$ to $f_k(t)$ explicitly given for any $k\in\N$, we state and prove the last ingredient we need to show Theorem \ref{thm:quantitative_convergence}.

\begin{theorem}\label{thm:distance_f_k_and_f_infty}
	Let $\br{f_k(t)}_{k\in\infty}$ and $\br{f_{k,\infty}}_{k\in\N}$ be as in Theorem \ref{thm:main_partial_order}. Then, there exists an explicit constant $\mathcal{D}_k$ such that 
	\begin{equation}\label{eq:distance_f_k_and_f_infty}
		\abs{\wh{f_k}\pa{n_1,\dots,n_k,t} - \wh{f_{k,\infty}}\pa{n_1,\dots,n_k}} \leq \mathcal{D}_k\pa{\frac{ e^{-2\lambda t}}{1-e^{-2\lambda t}}+e^{-\frac{\lambda m_2 t}{2}}}
	\end{equation}
\end{theorem}

\begin{proof}
	Recall that with the notations of Lemma \ref{lem:explicit_behaviour_for_recursive} we have that 
	\begin{equation}\nonumber
		\begin{split}
			&\wh{f_k}\pa{n_1,\dots,n_k,t} = \sum_{h=2}^k e^{-\lambda h\pa{h-1}t}b_{h,k}\pa{n_1,\dots,n_k,t} + a_k\pa{n_1,\dots,n_k}\wh{f_{1,0}}\pa{\sum_{r=1}^k n_r}e^{-\frac{\lambda m_2 \pa{\sum_{r=1}^k n_r}^2 t}{2}}	
		\end{split}
	\end{equation} 
	and that Corollary \ref{cor:t_limiting_behaviour_of_sol} states that
	$$\wh{f_{k,\infty}}\pa{n_1,\dots,n_k} = \delta_0 \pa{\sum_{r=1}^k n_r} a_k\pa{n_1,\dots,n_k}.$$ 
	Consequently, using \eqref{eq:uniform_bound_on_a} and \eqref{eq:uniform_bound_on_b} we find that 
	$$	\abs{\wh{f_k}\pa{n_1,\dots,n_k,t} - \wh{f_{k,\infty}}\pa{n_1,\dots,n_k}} \leq \sum_{h=2}^k e^{-\lambda h\pa{h-1}t}\abs{b_{h,k}\pa{n_1,\dots,n_k,t}} $$
	$$+\abs{a_k\pa{n_1,\dots,n_k}}\abs{\wh{f_{1,0}}\pa{\sum_{r=1}^k n_r}}\abs{e^{-\frac{\lambda m_2 \pa{\sum_{r=1}^k n_r}^2 t}{2}}-\delta_0 \pa{\sum_{r=1}^k n_r} }$$
	$$\leq \frac{C_k e^{-2\lambda t}}{1-e^{-2\lambda t}}\delta_{ \N\setminus\br{1}}(k)+\pa{1-\delta_{ 0}\pa{\sum_{r=1}^k n_r}}\pa{\prod_{j=1}^k\ell_j}e^{-\frac{\lambda m_2 \pa{\sum_{r=1}^k n_r}^2 t}{2}}$$
	$$\leq \frac{C_k e^{-2\lambda t}}{1-e^{-2\lambda t}}\delta_{ \N\setminus\br{1}}(k)+\pa{\prod_{j=1}^k\ell_j}e^{-\frac{\lambda m_2 t}{2}},$$
	where 
	$$\delta_{\N\setminus\br{1}}(k) = \begin{cases}
		0,& k=1,\\
		1,& k\geq 2,
	\end{cases}$$
	and $\ell_k$ and $C_k$ are defined by \eqref{eq:def_ell_for_upper_bound_for_a} and \eqref{eq:recursive_for_upper_bound_b} respectively. The result follows with the choice
	\begin{equation}\nonumber
		\mathcal{D}_k = \max\pa{C_k\delta_{\N\setminus\br{1}}(k) , \prod_{j=1}^k\ell_j}.
	\end{equation}
\end{proof}

We are finally ready to conclude this section with the proof of Theorem \ref{thm:quantitative_convergence}.
\begin{proof}[Proof of Theorem \ref{thm:quantitative_convergence}]
	The idea of the proof is to utilise both theorems \ref{thm:distance_F_N_k_and_f_k} and \ref{thm:distance_f_k_and_f_infty} by choosing a suitable sequence $\br{\alpha_N}_{N\in\N}$ that goes to zero in a way that $N\alpha_N^2$ goes to infinity. Optimising \eqref{eq:final_distance_simplified_k} shows us that we need to choose $\alpha_N$ such that 
	$$\alpha_N^{\min\pa{\lfloor \frac{l}{2} \rfloor, 2}} = \frac{1}{N\alpha_N^2} $$
	or
	$$\alpha_N = N^{-\frac{1}{2+\min\pa{\lfloor \frac{l}{2}\rfloor,2}}}.$$
	With $\alpha_N$ as above, we see that to satisfy \eqref{eq:condition_for_alpha_N_II} we must have that 
	\begin{equation}\label{eq:condition_for_N_0_final_I}
		N \geq \begin{cases}
			\min\br{4^{q+1}\norm{g}_{L^p}^q, \frac{m_2}{8m_3},1}^{-3}, & l=3,\\
			\min\br{4^{q+1}\norm{g}_{L^p}^q, \sqrt{\frac{m_2}{2m_4}},1}^{-4}, & l\geq 4,
		\end{cases}
	\end{equation}
	and in order to satisfy \eqref{eq:extra_condition_on_alpha_N} we must have that 
	\begin{equation}\label{eq:condition_for_N_0_final_II}
		N \geq \pa{\frac{ 4^{2\pa{q+3}}m_l\norm{g}_{L^p}^{2q}\pa{\sqrt[l]{4m_l}+2\pi}^2}{\pi^{l+2}}}^{\frac{1}{\frac{l}{2}-\frac{2}{2+\min\pa{\lfloor \frac{l}{2}\rfloor,2}}}}.
	\end{equation}
	Consequently, denoting by 
	\begin{equation}\label{eq:special_n_0}
		\mathfrak{N}_0=\begin{cases}
			\max\pa{\pa{\frac{ 4^{2\pa{q+3}}m_l\norm{g}_{L^p}^{2q}\pa{\sqrt[l]{4m_l}+2\pi}^2}{\pi^{l+2}}}^{\frac{4}{5}}, \min\br{4^{q+1}\norm{g}_{L^p}^q, \frac{m_2}{8m_3},1}^{-3}}, & l=3,\\
			\max\pa{\pa{\frac{ 4^{2\pa{q+3}}m_l\norm{g}_{L^p}^{2q}\pa{\sqrt[l]{4m_l}+2\pi}^2}{\pi^{l+2}}}^{\frac{1}{\frac{l}{2}-\frac{2}{2+\min\pa{\lfloor \frac{l}{2}\rfloor,2}}}},\min\br{4^{q+1}\norm{g}_{L^p}^q, \sqrt{\frac{m_2}{2m_4}},1}^{-4}}, & l\geq 4,
		\end{cases}
	\end{equation}
	we have that as long as 
	$$N\geq \max\pa{\mathfrak{N}_0, \frac{\pa{2m_l}^{\frac{l}{2}}}{\pi^2},\pa{\frac{32m_l }{\pi^l\max\pa{8,m_2}}}^{\frac{2}{l-2}},2k, \pa{\frac{96m_l k}{\pi^l}}^{\frac{2}{l-2}}}$$
	we can use Theorem \ref{thm:distance_F_N_k_and_f_k} and Theorem \ref{thm:distance_f_k_and_f_infty} to find that 
		\allowdisplaybreaks
	\begin{align}
			\nonumber&\abs{\wh{F_{N,k}}\pa{n_1,\dots,n_{k},t}-\wh{f_{k,\infty}}\pa{n_1,\dots,n_{k}}}\\
		\nonumber&\leq \abs{\wh{F_{N,k}}\pa{n_1,\dots,n_{k},t}-\wh{f_{k}}\pa{n_1,\dots,n_{k},t}}+	\abs{\wh{f_{k}}\pa{n_1,\dots,n_{k},t}-\wh{f_{k,\infty}}\pa{n_1,\dots,n_{k}}}\\
		\nonumber& \leq 
		e^{-\frac{\lambda\pa{ 2k\pa{k-1}+m_2}t}{2}}\abs{\wh{F_{N,k}}\pa{n_1,\dots,n_{k},0}-\wh{f_{k,0}}\pa{n_1,\dots,n_{k}}} \\
		\nonumber& +\pa{\frac{N}{N-1}}^{k-1}\pa{k-1}\max_{\fontsize{5}{4}\selectfont{\begin{matrix}
					r\in \br{1,\dots,k-1},\\
					\bm{p}\in \mathfrak{L}_{k}\pa{r},\;\;\sigma\in S^k
			\end{matrix}}, }\abs{\wh{F_{N,r}}\pa{s_{\bm{p}}^{\sigma}\pa{n_1,\dots,n_k},0}-\wh{f_{r,0}}\pa{s_{\bm{p}}^{\sigma}\pa{n_1,\dots,n_k}}}\\
		\nonumber& +e^{-\lambda k(k-1)t}\pa{e^{-\frac{\lambda N\alpha_N^2 \pi^2 t}{ 2\cdot4^{2\pa{q+2}}\norm{g}_{L^p}^{2q}\pa{\sqrt[l]{4m_l}+2\pi}^2}}+e^{-\frac{\lambda m_2N\alpha_N^2 t }{2}}}\\
		\nonumber
		\nonumber&+\mathfrak{c}_k\pa{\frac{1}{N} + \pa{kN^{\frac{2-l}{2}}+1}\alpha_N^2+kN^{\frac{2-l}{2}}+\alpha_N^{\min\pa{\lfloor \frac{l}{2}\rfloor,2}}+\frac{k(k-1)}{N\alpha_N^2}}+\mathcal{D}_k\pa{\frac{ e^{-2\lambda t}}{1-e^{-2\lambda t}}+e^{-\frac{\lambda m_2 t}{2}}}
	\end{align}
	Noticing that 
	\begin{equation}\nonumber
		\alpha_N^{\min\pa{\lfloor \frac{l}{2}\rfloor,2}} = \frac{1}{N\alpha_N^2} = \frac{1}{\sqrt[\kappa\pa{l}~]{N}}
	\end{equation}
	with 
		\begin{equation}\nonumber
		\kappa(l) = \begin{cases}
			3,& l=3,\\
			2, & l\geq 4.
		\end{cases}
	\end{equation}
	we conclude that the result follows from the fact that 
	$$e^{-\frac{\lambda N\alpha_N^2 \pi^2 t}{ 2\cdot4^{2\pa{q+2}}\norm{g}_{L^p}^{2q}\pa{\sqrt[l]{4m_l}+2\pi}^2}}+e^{-\frac{\lambda m_2N\alpha_N^2 t }{2}} \leq 2 e^{-\lambda \gamma \sqrt[\kappa\pa{l}~]{N} t},$$
	where 
	\begin{equation}\nonumber
		\gamma = \min\pa{\frac{ \pi^2 }{ 2\cdot4^{2\pa{q+2}}\norm{g}_{L^p}^{2q}\pa{\sqrt[l]{4m_l}+2\pi}^2},\frac{m_2}{2}}, 
	\end{equation}
	and the fact that 
	$$\mathfrak{c}_k\pa{\frac{1}{N} + \pa{kN^{\frac{2-l}{2}}+1}\alpha_N^2+kN^{\frac{2-l}{2}}+\alpha_N^{\min\pa{\lfloor \frac{l}{2}\rfloor,2}}+\frac{k(k-1)}{N\alpha_N^2}}	$$
	$$\leq \mathfrak{c}_k \pa{\frac{1}{\sqrt[\kappa\pa{l}~]{N}} + \pa{k+1}\alpha_N^{\min\pa{\lfloor \frac{l}{2}\rfloor,2}}+\frac{k}{\sqrt[\kappa\pa{l}~]{N}} +\alpha_N^{\min\pa{\lfloor \frac{l}{2}\rfloor,2}}+\frac{k(k-1)}{N\alpha_N^2}}	
	=\frac{\pa{k^2+k+3}\mathfrak{c}_k}{\sqrt[\kappa\pa{l}~]{N}}=\frac{\mathcal{C}_k}{{\sqrt[\kappa\pa{l}~]{N}}}, $$
	where we have used the fact that 
	$$N^{\frac{2-l}{2}} \leq \begin{cases}
		\frac{1}{\sqrt{N}},& l=3,\\
		\frac{1}{N}, & l\geq 4
	\end{cases} \leq \begin{cases}
	\frac{1}{\sqrt[3]{N}},& l=3,\\
	\frac{1}{\sqrt{N}}, & l\geq 4
	\end{cases} = \frac{1}{\sqrt[\kappa(l)~]{N}}.$$
	The proof is thus complete.
\end{proof}

\section{Final Remarks}\label{sec:remarks}
For all its relative simplicity, the CL model has proven to be an exceptionally fertile ground to exploring mean filed limits and their associated asymptotic correlations. It motivated the ``birth'' of not one, but two new notions -- order and partial order. While these notions are at their infancy, their relevance and applicability encourage further exploration. 
\subsection*{The philosophy behind partial order}
The main idea that motivated our definition of partial order, expressed in Definition \ref{def:partial_order}, is that as the number of elements in our system increases a ``leader'' emerges to whom the other elements align -- potentially with non-negligible deviation. While we have used this notion in the context of the family of tori $\br{\T^N}_{N\in\N}$ and defined it in a more general metric settings, this idea can be adapted to other situations where a similar phenomenon is expected. For instance: we have assumed that the elements in our system are indistinguishable implicitly by using a uniform weighting of potential leaders in our definition. This can easily be adapted to a case where our system involves two (or more) types of elements, one of which is more prone to lead.   

\subsection*{The challenges of exploring partial order}
While the notion of partial order seems extremely relevant to (and prevalent in) many biological and societal phenomena, it is much more challenging to show its validity in a given mean field model in comparison to the notions of chaos or order. This is mainly due to the fact that both the latter phenomena are described completely by the limit of the first marginal, while the notion of partial order requires finding a \textit{family} of probability measures. There is also no apriori simple connection between these probability measures (for instance, it is straightforward to see that in general $\nu_{k-1}$ is not the $(k-1)-$th marginal of $\nu_k$ even when $\mu_k=\mu$ for some probability measure $\mu$).

Another interesting point we'd like to raise here is the fact that it is unclear to us at this point whether or not the representation of a partially ordered state, expressed by \eqref{eq:def_of_partial_order}, is unique. We intend to explore this further in future work. 

\subsection*{Beyond the CL model}
We expect that the notion of partial order (as well as that of order) is well suited to deal with many existing Kac-like models of biological, societal, and economical nature like the ones mentioned in \cite{CD(II)2022}. Much like the discussion in \cite{E2024}, we expect that the study of any such model will start with finding an appropriate scaling to distinguish between chaos, order, and partial order (or any other asymptotic correlation). We intend to explore this further in future work.

\section*{Acknowledgement}
The first author was supported by the Royal Society Research Grant (Round 1) RG\textbackslash R1\textbackslash241410. The second author was supported by the China Scholarship Council (with Scholarship Number 202308060310).

\begin{appendix}
	
	\section{Additional Proofs}\label{app:additional_proofs}
In this appendix we will consider proofs to various technical results which were mentioned throughout the presented work. 

\subsection*{On the Fourier coefficients of $g_{\epsilon_N}$}

To prove part \eqref{item:high_frequencies} of Lemma \ref{lem:on_g_low_and_high_frequencies} we will need the following lemma which is inspired by \cite{CCLLV2010}:

\begin{lemma}\label{lem:away_from_zero}
	Let $g\in L^1\pa{\R}\cap L^p\pa{\R}$ for some $p>1$ be a symmetric probability density with finite $l-$th moment, $m_l$. For any $~0<\beta<1$, $\alpha>0$,  and $R>0$ we have that
	\begin{equation}\label{eq:away_from_zero}
		\mathcal{F}\pa{g}(\xi) - 1 \leq -\beta \pa{1-\frac{m_l}{R^l}-\norm{g}_{L^p}\pa{4\pa{\frac{R}{\pi} + \frac{2}{\alpha}}}^{\frac{1}{q}}\beta^{\frac{1}{2q}}},
	\end{equation} 
	for any $\xi\in \R$ such that $\abs{\xi} \geq \alpha$ where 
	 $$\F\pa{g}\pa{\xi} = \int_{\R}g(x)e^{-i\xi x}dx.$$
	In particular, for any 
	\begin{equation}\label{eq:beta_condition}
		\beta \leq \min\br{\frac{1}{2},\frac{\alpha^2 \pi^2}{4^{2\pa{q+1}}\norm{g}_{L^p}^{2q}\pa{\sqrt[l]{4m_l}\alpha+2\pi}^2}}	
	\end{equation}
	we find that 
	\begin{equation}\label{eq:away_from_zero_simple}
		\mathcal{F}\pa{g}(\xi) - 1 \leq -\frac{\beta}{2}. 
	\end{equation} 
\end{lemma}
\begin{proof}
		We start by noticing that since $g\in \PP\pa{\R,dx}$ and is even, $\mathcal{F}\pa{g}$ real valued, even, and lies in $[-1,1]$. Moreover
	$$\F\pa{g}\pa{\xi} = \int_{\R}g(x)\cos\pa{\xi x}dx = \int_{\R}g(x)\cos\pa{\abs{\xi} x}dx $$
	and
	$$\F\pa{g}(\xi) - 1 = - \int_{\R}g(x)\pa{1-\cos\pa{\abs{\xi} x}}dx.$$
	For a given $\xi\not=0$, $0<\beta<1$, and $R>0$ we define the set 
	$$B_{\beta,R} = \br{x\in [-R,R]\;|\; 1-\cos\pa{\abs{\xi} x} \leq \beta}.$$ 
	We notice that the condition
	$$1-\cos\pa{\abs{\xi}x} \leq \beta,\qquad \abs{\xi}x\in [-\pi,\pi]$$
	implies that 
	$$-\cos^{-1}\pa{1-\beta} \leq \abs{\xi}x\leq \cos^{-1}\pa{1-\beta}$$ 
	from which we conclude that 
	$$\abs{B_{\beta,R}\cap \rpa{-\frac{\pi}{\abs{\xi}},\frac{\pi}{\abs{\xi}}}}\leq \frac{2\cos^{-1}\pa{1-\beta}}{\abs{\xi}},$$
	where $\abs{A}$ refers to the Lebesgue measure of the set $A$.
	
	Due to the fact that $\cos\pa{\abs{\xi}x}$ is $\frac{2\pi}{\abs{\xi}}$ periodic and since there are at most $\frac{2R}{\frac{2\pi}{\abs{\xi}}}+2$ segments of length $\frac{2\pi}{\abs{\xi}}$ in $[-R,R]$ we conclude that
	$$\abs{B_{\beta,R}} \leq \pa{\frac{2R\abs{\xi}}{\pi}+4}\frac{\cos^{-1}\pa{1-\beta}}{\abs{\xi}}=\pa{\frac{2R}{\pi} + \frac{4}{\abs{\xi}}}\cos^{-1}\pa{1-\beta}$$
	$$\leq 4\pa{\frac{R}{\pi} + \frac{2}{\alpha}}\sqrt{\beta}$$
	where we have used the condition $\abs{\xi} \geq \alpha$ and the inequality\footnote{Defining $f(t)=\cos^{-1}\pa{1-t}-2\sqrt{t}$, with $t\in [0,1]$, we see that for $t\in (0,1)$
		$$f^\prime(t) = \frac{1}{\sqrt{1-\pa{1-t}^2}}-\frac{1}{\sqrt{t}}=\frac{1}{\sqrt{t}}\pa{\frac{1}{\sqrt{2-t}}-1}<0$$
		which implies that $f(t)$ is decreasing. As $f(0)=0$ we conclude the desired result.}
	$$\cos^{-1}\pa{1-\beta} \leq 2\sqrt{\beta}.$$ 
	Turning our attention back to our desired quantity we find that 
	\begin{equation}\label{eq:away_from_zero_comp_I}
		\begin{split}
			\int_{\R}g(x)&\pa{1-\cos\pa{\abs{\xi} x}}dx \geq \int_{[-R,R]\setminus B_{\beta,R}}g(x)\pa{1-\cos\pa{\abs{\xi} x}}dx\\
			&\geq \beta \int_{[-R,R]\setminus B_{\beta,R}}g(x)dx = \beta\pa{1 - \int_{[-R,R]^c}g(x)dx-\int_{ B_{\beta,R}}g(x)dx}.
		\end{split}	
	\end{equation}
	Since $g$ has a finte moment of order $l$
	$$\int_{[-R,R]^c}g(x)dx\leq  \int_{[-R,R]^c}\frac{\abs{x}^l}{R^l}g(x)dx \leq \frac{m_l}{R^l},$$
	and by H\"older inequality
	$$\int_{ B_{\beta,R}}g(x)dx \leq \norm{g}_{L^p} \abs{B_{\beta,R}}^{\frac{1}{q}} \leq \norm{g}_{L^p}\pa{4\pa{\frac{R}{\pi} + \frac{2}{\alpha}}}^{\frac{1}{q}}\beta^{\frac{1}{2q}}.$$
	Combining these inequalities with \eqref{eq:away_from_zero_comp_I} gives us \eqref{eq:away_from_zero}. 
	
	To show \eqref{eq:away_from_zero_simple} we notice that by for $R=\sqrt[l]{4m_l}$ we have that
	$$\norm{g}_{L^p}\pa{4\pa{\frac{\sqrt[l]{4m_l}}{\pi} + \frac{2}{\alpha}}}^{\frac{1}{q}}\beta^{\frac{1}{2q}} \leq \frac{1}{4}$$
	if and only if 
	$$\beta \leq \frac{\alpha^2 \pi^2}{4^{2\pa{q+1}}\norm{g}_{L^p}^{2q}\pa{\sqrt[k]{4m_k}\alpha+2\pi}^2}.$$
	With this choice of $R$ and condition on $\beta$ we find that  \eqref{eq:away_from_zero} implies \eqref{eq:away_from_zero_simple} as desired.
	\end{proof}
	The second ingredient in the proof of part \eqref{item:high_frequencies} of Lemma \ref{lem:on_g_low_and_high_frequencies} is the following lemma, whose proof can be found in \cite[Lemma 23]{E2024}:
	
	\begin{lemma}\label{lem:fourier_connection}
		Let $g\in L^1\pa{\R}$ be a probability density with finite $l-$th moment, $m_l$. Then for any $n\in\Z$ we have that
		\begin{equation}\nonumber 
			\abs{\widehat{g_{\epsilon_N}}(n)-\F\pa{g}\pa{n\epsilon_N}} \leq \frac{2\epsilon_N^l m_l}{\pi^l-\epsilon_N^l m_l}
		\end{equation}
		whenever $\epsilon_N  < \frac{\pi}{\sqrt[l]{m_l}}$. 
	\end{lemma}

	\begin{proof}[Proof of part \eqref{item:high_frequencies} of Lemma \ref{lem:on_g_low_and_high_frequencies}]
	Since
	$$\wh{g}_{\epsilon_N}(n)-1 = \pa{\wh{g}_{\epsilon_N}(n)-\F\pa{g}(n\epsilon_N)} + \F\pa{g}\pa{n\epsilon_N}-1,$$
	we find that lemmas \ref{lem:away_from_zero} and \ref{lem:fourier_connection} imply that 
	\begin{equation}\label{eq:almost_final_low_frequencies_bound}
		\wh{g}_{\epsilon_N}(n)-1 \leq \frac{2\epsilon_N^l m_l}{\pi^l-\epsilon_N^l m_l} - \frac{\min\br{\frac{1}{2},\frac{\alpha_N^2 \pi^2}{4^{2\pa{q+1}}\norm{g}_{L^p}^{2q}\pa{\sqrt[l]{4m_l}\alpha_N+2\pi}^2}}}{2},	
	\end{equation}
	when $\epsilon_N \leq \frac{\pi}{\sqrt[l]{m_l}}$ and $\abs{n}\epsilon_N \geq \alpha_N$.
	
	As 
	$$\frac{\alpha_N^2 \pi^2}{4^{2\pa{q+1}}\norm{g}_{L^p}^{2q}\pa{\sqrt[l]{4m_l}\alpha_N+2\pi}^2} \leq \frac{\alpha_N^2 }{4^{2\pa{q+1}+1}\norm{g}_{L^p}^{2q}}$$
	we see that if $\alpha_N \leq 4^{q+1}\norm{g}_{L^p}^q$ then
	$$\min\br{\frac{1}{2},\frac{\alpha_N^2 \pi^2}{4^{2\pa{q+1}}\norm{g}_{L^p}^{2q}\pa{\sqrt[k]{4m_k}\alpha_N+2\pi}^2}}=\frac{\alpha_N^2 \pi^2}{4^{2\pa{q+1}}\norm{g}_{L^p}^{2q}\pa{\sqrt[k]{4m_k}\alpha_N+2\pi}^2},$$
	which, together with \eqref{eq:almost_final_low_frequencies_bound}, gives us
	\begin{equation}\nonumber
		\wh{g}_{\epsilon_N}(n)-1 \leq \uptau_N - \frac{\alpha_N^2 \pi^2}{2\cdot 4^{2\pa{q+1}}\norm{g}_{L^p}^{2q}\pa{\sqrt[l]{4m_l}\alpha_N+2\pi}^2}.	
	\end{equation}
	This concludes the proof.
\end{proof}

\subsection*{The proof of Lemma \ref{lem:estimating_quadratic_sum}}
 	\begin{proof}
 		For a given $A,B,K\in\Z$ we define the following sets
 		$$S_0\pa{A,B,K} := \br{n\in\Z\;|\; 0<\abs{n^2+An+B+\frac{K}{m}} \leq 1},$$
		$$S_N^+\pa{A,B,K} := \br{n\in\Z\;|\; N^2 \leq n^2+An+B+\frac{K}{m} \leq \pa{N+1}^2},$$
 		$$S_N^-\pa{A,B,K} := \br{n\in\Z\;|\; -\pa{N+1}^2 \leq n^2+An+B+\frac{K}{m} \leq -N^2},$$
 		where $N\in\N$. To prove our desired result we will start by estimating the size of each of these sets:\\
 		\underline{$S_0\pa{A,B}$:} Any $n\in S_0\pa{A,B,K}$ will satisfy
 		$$-1 - \frac{K}{m} \leq n^2+An+B \leq 1-\frac{K}{m}.$$
 		As $A$, $B$ and $n$ are integers, the above implies that the integer $n^2+An+B$ lies in an interval of length $2$. As such an interval can contain at most $3$ integers, and as quadratic equations have at most two solutions, we conclude that 
 		\begin{equation}\label{eq:S_0}
 			\sharp S_0\pa{A,B,K} \leq 6.
 		\end{equation}
 		\underline{$S_N^+\pa{A,B,K}$:} The inequality 
 		$$ N^2 \leq n^2+An+B+\frac{K}{m} \leq \pa{N+1}^2$$ 
 		can be rewritten as
 		$$N^2 - \pa{B -\frac{A^2}{4}+ \frac{K}{m}} \leq \pa{n+\frac{A}{2}}^2 \leq \pa{N+1}^2 -\pa{ B -\frac{A^2}{4}+ \frac{K}{m}}.$$
 		We have the following options:
 		\begin{itemize}
 			\item If $B-\frac{A^2}{4}+\frac{K}{m}\geq 0$ and $N+1 < \sqrt{B-\frac{A^2}{4}+\frac{K}{m}}$ there are no solutions to the above inequality and consequently $\sharp S_N^+\pa{A,B,K}=0$.
 			\item If $B-\frac{A^2}{4}+\frac{K}{m}\geq 0$, $N \leq \sqrt{B-\frac{A^2}{4}+\frac{K}{m}}$, and $N+1 \geq \sqrt{B-\frac{A^2}{4}+\frac{K}{m}}$ then 
 			$$-\frac{A}{2}- \sqrt{\pa{N+1}^2 - \pa{B-\frac{A^2}{4}+\frac{K}{m}}} \leq n \leq -\frac{A}{2} + \sqrt{\pa{N+1}^2 - \pa{B-\frac{A^2}{4}+\frac{K}{m}}} $$
 			Consequently, $n$ lies in an interval of length 
 			$$2\sqrt{\pa{N+1}^2 - \pa{B-\frac{A^2}{4}+\frac{K}{m}}}
 			\leq 2 \sqrt{\pa{N+1}^2-N^2} = 2\sqrt{2N+1}.$$
 			and we conclude that 
 			$$\sharp S_N^+\pa{A,B,K} \leq 2\sqrt{2N+1}+2.$$
 			\item Lastly, if $N^2 \geq B-\frac{A^2}{4}+\frac{K}{m} \geq 0$ then 
 			$$-\frac{A}{2} + \sqrt{N^2 - \pa{B-\frac{A^2}{4}+\frac{K}{m}}}\leq n \leq -\frac{A}{2} + \sqrt{\pa{N+1}^2 - \pa{B-\frac{A^2}{4}+\frac{K}{m}}},$$
 			or 
 			$$-\frac{A}{2} - \sqrt{\pa{N+1}^2 - \pa{B-\frac{A^2}{4}+\frac{K}{m}}}\leq n \leq -\frac{A}{2}-  \sqrt{N^2 - \pa{B-\frac{A^2}{4}+\frac{K}{m}}},$$
 			and consequently $n$ lies in two possible intervals of length
 			$$ \sqrt{\pa{N+1}^2 - \pa{B-\frac{A^2}{4}+\frac{K}{m}}}-\sqrt{N^2 - \pa{B-\frac{A^2}{4}+\frac{K}{m}}} = \frac{2N+1}{\sqrt{N^2 - \pa{B-\frac{A^2}{4}+\frac{K}{m}}}+\sqrt{\pa{N+1}^2 - \pa{B-\frac{A^2}{4}+\frac{K}{m}}}}$$
 			$$ \leq \frac{2N+1}{\sqrt{N^2 - \pa{B-\frac{A^2}{4}+\frac{K}{m} }+ 2N+1}} \leq \sqrt{2N+1}$$
 			and similarly to the previous case we find that 
 			$$\sharp S_N^+\pa{A,B,K} \leq 2\pa{\sqrt{2N+1}+2}.$$
 		\end{itemize}
 		After considering all the possibilities we see that for any $N\in\N$
 		\begin{equation}\label{eq:S_N+}
 			\sharp S_N^+\pa{A,B,K} \leq 2\sqrt{2N+1}+4 \leq 8\sqrt{N}.
 		\end{equation}
 		\underline{$S_N^-\pa{A,B,K}$:} The inequality 
 		$$ -\pa{N+1}^2 \leq n^2+An+B+\frac{K}{m} \leq -N^2$$ 
 		can be rewritten as
 		$$-\pa{N+1}^2 - \pa{B -\frac{A^2}{4}+ \frac{K}{m} }\leq \pa{n+\frac{A}{2}}^2 \leq -N^2 - \pa{B -\frac{A^2}{4}+ \frac{K}{m}}.$$
 		Much like our study of $S_N^+\pa{A,B,K}$, we will need to consider the following possibilities:
 		\begin{itemize}
 			\item If $N^2> -\pa{B-\frac{A^2}{4}+\frac{K}{m}}$ there are no solutions to the above inequality, i.e. $\sharp S_N^-\pa{A,B,K}=0$.
 			\item If $B-\frac{A^2}{4}+\frac{K}{m}\leq  0$, $N+1 \geq \sqrt{-\pa{B-\frac{A^2}{4}+\frac{K}{m}}}$, and $N \leq \sqrt{-\pa{B-\frac{A^2}{4}+\frac{K}{m}}}$ then 
 			$$-\frac{A}{2}- \sqrt{-N^2 - \pa{B-\frac{A^2}{4}+\frac{K}{m}}} \leq n \leq -\frac{A}{2} + \sqrt{-N^2 - \pa{B-\frac{A^2}{4}+\frac{K}{m}}}$$
 			i.e. $n$ lies in an interval of length 
 			$$2\sqrt{-N^2 - \pa{B-\frac{A^2}{4}+\frac{K}{m}}}
 			\leq 2 \sqrt{\pa{N+1}^2-N^2} = 2\sqrt{2N+1},$$
 			from which we find that 
 			$$\sharp S_N^-\pa{A,B,K} \leq 2\sqrt{2N+1}+2.$$
 			\item Lastly, if $B-\frac{A^2}{4}+\frac{K}{m}\leq  0$ and $N+1 \leq \sqrt{-\pa{B-\frac{A^2}{4}+\frac{K}{m}}}$ then 
 			$$-\frac{A}{2} + \sqrt{-\pa{N+1}^2 - \pa{B-\frac{A^2}{4}+\frac{K}{m}}}\leq n \leq -\frac{A}{2} + \sqrt{-N^2 - \pa{B-\frac{A^2}{4}+\frac{K}{m}}},$$
 			or
 			$$-\frac{A}{2} - \sqrt{-N^2 - \pa{B-\frac{A^2}{4}+\frac{K}{m}}}\leq n \leq -\frac{A}{2} - \sqrt{-\pa{N+1}^2 -\pa{ B-\frac{A^2}{4}+\frac{K}{m}}},$$
 			and consequently $n$ lies in two possible intervals of length
 			$$\frac{2N+1}{\sqrt{-\pa{N+1}^2 -\pa{ B-\frac{A^2}{4}+\frac{K}{m}}}+\sqrt{-N^2 -\pa{ B-\frac{A^2}{4}+\frac{K}{m}}}}$$
 			$$ \leq \frac{2N+1}{\sqrt{-\pa{N+1}^2 - \pa{B-\frac{A^2}{4}+\frac{K}{m}} + 2N+1}} \leq \sqrt{2N+1}.$$
 			We conclude that 
 			$$\sharp S_N^-\pa{A,B,K} \leq 2\sqrt{2N+1}+4.$$
 		\end{itemize}
 		After considering all the possibilities we see that for any $N\in\N$
 		\begin{equation}\label{eq:S_N-}
 			\sharp S_N^-\pa{A,B,K} \leq 2\sqrt{2N+1}+4\leq 8\sqrt{N}.
 		\end{equation}
 		The last ingredient we need for the proof is the fact that for any $A,B,K\in\Z$ and any $m>0$ 
 		$$\sup_{\fontsize{5}{4}\selectfont{\begin{matrix}
 					n\in\Z \\
 					m\pa{n^2+An+B}+K\ne 0	
 		\end{matrix}}}\frac{1}{\abs{m\pa{n^2+An+B}+K}}  \leq \begin{cases}
 			\frac{1}{m},& \frac{K}{m}\in\Z,\\
 			\frac{1}{\min\pa{K-m\lfloor \frac{K}{m} \rfloor, m\pa{\lfloor \frac{K}{m} \rfloor+1}-K}}, & \frac{K}{m}\not\in\Z
 		\end{cases}.$$
 		where we have used identity
 			$$\inf_{n\in\Z\setminus\br{\frac{K}{m}}}\abs{\frac{K}{m}-n} = \begin{cases}
 				1,& \frac{K}{m}\in\Z,\\
 				\min\pa{\frac{K}{m}-\lfloor \frac{K}{m} \rfloor, \lfloor \frac{K}{m} \rfloor+1-\frac{K}{m}}, & \frac{K}{m}\not\in\Z
 			\end{cases}.$$
 		 which was shown in the proof Lemma \ref{lem:explicit_behaviour_for_recursive}.\\
 		 Combining all the above
 		 we see that for a given $A,B,K\in\Z$ and $m>0$ we have that 
 		 $$\Z \setminus\br{n\in\Z\;|\; m\pa{n^2+An+B}+K=0}=S_0\pa{A,B,K}\cup \bigcup_{N\in\N}\pa{S_N^+(A,B,K)\cup S_N^-(A,B,K)},$$
 		 and
 		 $$\sum_{\fontsize{5}{4}\selectfont{\begin{matrix}
 		 			n\in\Z \\
 		 			m\pa{n^2+An+B}+K\ne 0	
 		 \end{matrix}}}\frac{1}{\abs{m\pa{n^2+An+B}+K}} \leq \sum_{n\in S_0\pa{A,B}}\frac{1}{\abs{m\pa{n^2+An+B}+K}}  $$
 	 $$+\sum_{N=1}^\infty\sum_{n\in S_N^+\pa{A,B,K}}\frac{1}{\abs{m\pa{n^2+An+B}+K}}+\sum_{N=1}^\infty\sum_{n\in S_N^{-}\pa{A,B,K}}\frac{1}{\abs{m\pa{n^2+An+B}+K}}$$
 	 $$\leq \begin{cases}
 	 	\frac{6}{m},& \frac{K}{m}\in\Z,\\
 	 	\frac{6}{\min\pa{K-m\lfloor \frac{K}{m} \rfloor, m\pa{\lfloor \frac{K}{m} \rfloor+1}-K}}, & \frac{K}{m}\not\in\Z
 	 \end{cases} + \frac{16}{m}\sum_{N=1}^\infty \frac{1}{N^{\frac{3}{2}}},$$
 	 where we have used the fact that on $S^\pm_N\pa{A,B,K}$
 	 $$\frac{1}{\abs{m\pa{n^2+An+B}+K}}= \frac{1}{m\abs{n^2+An+B+\frac{K}{m}}}\leq \frac{1}{mN^2}.$$
 	 The proof is thus complete.
 	\end{proof}
 	
\subsection*{The Fundamental solution for the operator $a^2-\Delta$ on $\R^k$} 	
 	
 	\begin{lemma}\label{lem:green_function}
 		Let $a>0$ be given and let $Y_a:\R^k\to\R$ be defined by
 		\begin{equation}\nonumber
 			Y_a\pa{\bm{x}} = \begin{cases}
 			 		\frac{e^{-a\abs{x}}}{a},& k=1,\\
 			 		\frac{a^{\frac{k-2}{2}} K_{\frac{k-2}{2}}\pa{a\abs{\bm{x}}}}{|x|^{\frac{k-2}{2}}}, & k\geq 2,		
 			 	\end{cases}
 		\end{equation}
 		where $K_\nu$ is the modified Bessel function of second kind or order $\nu$ which can be written as\footnote{See, for instance \cite[Formula 1 in 8.432]{GR2007}.}
 		$$K_{\nu}(x) = \int_{0}^{\infty} e^{-x\cosh(t)}\cosh\pa{\nu t}dt,$$
 		when $\nu,x>0$. Then $Y_a$ is a non-negative function in $L^1\pa{\R^k}$. Moreover, there exists an explicit $C_k>0$ such that 
		$$\wh{Y_a}\pa{\xi_1,\dots,\xi_k} = \frac{C_k}{a^2+\abs{\bm{\xi}}^2}.$$
 	\end{lemma}
 	
 	\begin{proof}
 		The fact that $Y_a$ is a non-negative follows from its definition. We continue by considering the case $k=1$ and $k\geq 2$ separately, starting with $k=1$. In the case $Y_a$ is clearly in $L^1\pa{\R}$. Moreover, we recall that $e^{-ar}e^{i \xi r}\in L^1\pa{[0,\infty)}$ for any $a>0$ and $\xi\in\R$ and
 		$$\int_{0}^\infty e^{-ar} e^{-i\xi r}dr = \frac{1}{a+i\xi}.$$
 		Consequently
 		$$\int_{\R}Y_a(x_1)e^{-i\xi x_1}dx_1 = \frac{1}{a}\int_{\R}e^{-a\abs{x_1}}e^{-i\xi x_1}dx_1 = \frac{1}{2a}\int_{\R}e^{-a\abs{x_1}}\pa{e^{-i\xi x_1}+e^{i\xi x_1}}dx_1$$
 		$$=\frac{1}{a}\int_{\R}e^{-a\abs{x_1}}\cos\pa{\xi x_1}dx_1 = \frac{2}{a}\int_{0}^\infty e^{-a\abs{x_1}}\cos\pa{\xi x_1}dx_1 = \frac{2}{a}\mathrm{Re}\pa{\int_{0}^\infty e^{-ar} e^{-i\xi r}dr } = \frac{2}{a^2+\xi^2},$$
 		which concludes the one dimensional case.
 		
 		We proceed by considering the case $k\geq 2$. Since
 		$$\int_{0}^\infty Y_a\pa{\bm{x}}d\bm{x} = a^{\frac{k-2}{2}}\int_{(0,\infty)\times \mathbb{S}^{k-1}}r^{\frac{k}{2}}K_{\frac{k-2}{2}}\pa{ar}drd\Omega_k$$
 		$$=a^{\frac{k-2}{2}}\abs{\mathbb{S}^{k-1}}\int_{0}^{\infty}r^{\frac{k}{2}}K_{\frac{k-2}{2}}\pa{ar}dr = \frac{2^{\frac{k-2}{2}}\abs{\mathbb{S}^{k-1}}\Gamma\pa{\frac{k}{2}}}{a^2}<\infty$$
 	 	(see, for instance \cite[Formula 16 in 6.561]{GR2007}), we find that $Y_a\in L^1\pa{\R^k}$.
 		
 		To calculate the Fourier transform of $Y_a$ we notice that due to the fact that $Y_a$ is radial, we can rotate our space so that the frequency variable $\bm{\xi}$ lies on the rotated $x_{k}-$th axis. As such, we can assume without loss of generality that 
 		$$\pa{\xi_1,\dots,\xi_k} = \pa{0,\dots,0,\abs{\bm{\xi}}}.$$
 		Using spherical coordinates $\pa{r,\theta_1,\dots,\theta_{k-2},\phi}$, with $\theta_1$ being the angle relative to the $x_k-$th axis and where $0\leq \theta_i<\pi$ for $i=1,\dots, k-2$ and $0\leq \phi<2\pi$, we find that 
 		$$\int_{\R^k}Y_a\pa{\bm{x}}e^{-i\bm{\xi}\cdot \bm{x}}d\bm{x} =a^{\frac{k-2}{2}} \int_{[0,\infty)\times [0,\pi)^{k-1}\times [0,2\pi)}r^{\frac{k}{2}}K_{\frac{k-2}{2}}\pa{ar}e^{-i\abs{\bm{\xi}}r\cos\pa{\theta_1}}\sin^{k-2}\pa{\theta_1}drd\theta_1 d\Omega_{k-2}\pa{\theta_2,\dots,\theta_{k-2},\phi}$$
 		$$=\abs{\mathbb{S}^{k-2}}a^{\frac{k-2}{2}}\int_{[0,\infty)\times [0,\pi)}r^{\frac{k}{2}}K_{\frac{k-2}{2}}\pa{ar}e^{-i\abs{\bm{\xi}}r\cos\pa{\theta_1}}\sin^{k-2}\pa{\theta_1}drd\theta_1$$
 		where we used the convention $\abs{\mathbb{S}^0}=2$\footnote{We need to be slightly careful when $k=2$ as in this case our radial variables are $r$ and $\phi$. However, as
 		$$\int_{0}^{\pi} e^{-i\abs{\bm{\xi}}r\cos\pa{\phi}}d\phi = \int_{\pi}^{2\pi} e^{-i\abs{\bm{\xi}}r\cos\pa{\phi}}d\phi, $$
 	the result still holds with our convention $\abs{\mathbb{S}^0}=2$.}. It is well known that for $\beta>0$ and $k\geq 2$
 		$$\int_{0}^\pi e^{-i\beta\cos\pa{\theta}}\sin^{k-2}\pa{\theta}d\theta=\sqrt{\pi}\pa{\frac{2}{\beta}}^{\frac{k-2}{2}}\Gamma\pa{\frac{k-1}{2}}J_{\frac{k-2}{2}}\pa{\beta},$$
 		where $J_{\nu}$ is the Bessel function of first kind of order $\nu$	(see, for instance \cite[Formula 5 in 3.915]{GR2007}). We find that when $\bm{\xi}\ne \pa{0,\dots,0}$
 		$$ \int_{\R^k}Y_a\pa{\bm{x}}e^{-i\bm{\xi}\cdot \bm{x}}d\bm{x} = C_k\frac{a^{\frac{k-2}{2}}}{\abs{\bm{\xi}}^{\frac{k-2}{2}}}\int_{0}^\infty rK_{\frac{k-2}{2}}\pa{ar}J_{\frac{k-2}{2}}\pa{\abs{\bm{\xi}}r}dr,$$
 		where $C_k$ is a fixed geometric constant. 
 		Since
 		$$\int_{0}^\infty rK_{\frac{k-2}{2}}\pa{ar}J_{\frac{k-2}{2}}\pa{\abs{\bm{\xi}}r}dr = \frac{\abs{\bm{\xi}}^{\frac{k-2}{2}}}{a^{\frac{k-2}{2}}\pa{a^2+\abs{\bm{\xi}}^2}} $$
 		(see, for instance \cite[Formula 2 in 6.521]{GR2007}) we conclude the desired result using the continuity of $\wh{Y_a}$ to include $\bm{\xi} = \pa{0,\dots,0}$.
 	\end{proof}
 	
\end{appendix}


\begin{thebibliography}{99}
	
	
	\bibitem{APD2021}
	Ayi, N and Pouradier Duteil, N., \textit{Mean-field and graph limits for collective dynamics models
		with time-varying weights}, J. Differential Equations \textbf{299} (2021). https://doi.org/10.1016/j.jde.2021.07.010
	
	\bibitem{BCC2011}
	Bolley, F. and Ca\~{n}izo, J. A. and Carrillo, J. A., \textit{Stochastic mean-field limit: non-{L}ipschitz forces and
		swarming}, Math. Models Methods Appl. Sci. \textbf{11} (2011). https://doi.org/10.1142/S0218202511005702
	
	
	\bibitem{BFFT2012}
	Baladron, J. and Fasoli, D. and Faugeras, O. and Touboul J., \textit{Mean-field description and propagation of chaos in networks of
		{H}odgkin-{H}uxley and {F}itz{H}ugh-{N}agumo neurons}, J. Math. Neurosci. \textbf{2} (2012). https://doi.org/10.1186/2190-8567-2-10
	
	\bibitem{CCLLV2010}
	Carlen E. A., Carvalho M. C., Le Roux J., Loss M. and Villani C., \textit{Entropy and chaos in the {K}ac model},Kinet. Relat. Models, \textbf{3} (2010). 
	
	\bibitem{CCDW2013}
	Carlen, E. A. and Chatelin, R. and Degond, P. and Wennberg, B., \textit{Kinetic hierarchy and propagation of chaos in biological swarm models}, Phys. D \textbf{260} (2013), 90--111. https://doi.org/10.1016/j.physd.2012.05.013
	
	\bibitem{CDW2013}
	Carlen, E. A. and Degond, P. and Wennberg, B., \textit{Kinetic limits for pair-interaction driven master equations
		and biological swarm models}, Math. Models Methods Appl. Sci. \textbf{23} (2013), 1339--1376. https://doi.org/10.1142/S0218202513500115
	
	\bibitem{CD(I)2022}
	Chaintron, L.-P. and Diez, A., \textit{Propagation of chaos: a review of models, methods and
		applications. I. Models and methods}, Kinet. Relat. Model \textbf{15} (2022), 895--1015. https://doi.org/10.3934/krm.2022017
		
	\bibitem{CD(II)2022}
	Chaintron, L.-P. and Diez, A., \textit{Propagation of chaos: a review of models, methods and
		applications. II. Applications}, Kinet. Relat. Models \textbf{15} (2022), 1017--1173. https://doi.org/10.3934/krm.2022018
	
	
	\bibitem{DLMT2017}
	P. Degond, J-G. Liu, S. Merino-Aceituno, and T. Tardiveau. \textit{Continuum dynamics of the intention field under weakly cohesive social interaction}. Math. Models Methods Appl. Sci., \textbf{27} (2017), No. 1,  159--182. https://doi.org/10.1142/S021820251740005X
	
	\bibitem{E2024}
	Einav, A., \textit{The Emergence of Order in Many Element Systems}. J. Stat. Phys., \textbf{191} (2024), No. 86.
	
	\bibitem{FL2016}
	N. Fournier and E L\"ocherbach. \textit{On a toy model of interacting neurons}. Ann. Inst. Henri Poincar\'e Probab. Stat., \textbf{52} (2016), No. 4, 1844--1876. https://doi.org/10.1214/15-AIHP701
	
	
	\bibitem{GR2007}
	Gradshteyn, I. S. and Ryzhik, I. M., \textit{Table of integrals, series, and products seventh edition}, Elsevier/Academic Press (2007).
	
%
	
	
	\bibitem{K1956}
	Kac, M., \textit{Foundations of kinetic theory}, Proceedings of the Third Berkeley Symposium on
	Mathematical Statistics and Probability, 1954--1955,
	vol. III (1956), 171--197.
	
%
	
	
	\bibitem{Rudin91}
	Rudin, W., \textit{Fourier Analysis on Groups}, John Wiley \& Sons (1991). 
	
%
\end{thebibliography}
\end{document}